\documentclass[aos, preprint]{imsart}

\RequirePackage{amsthm,amsmath,amsfonts,amssymb}
\RequirePackage[round,sort&compress]{natbib}
\RequirePackage[colorlinks,citecolor=blue,urlcolor=blue]{hyperref}
\RequirePackage{graphicx}

\startlocaldefs
\theoremstyle{plain}

\newtheorem{theorem}{Theorem}[section]
\newtheorem{lemma}[theorem]{Lemma}
\newtheorem{proposition}[theorem]{Proposition}
\theoremstyle{definition}

\newtheorem{remark}[theorem]{Remark}
\newtheorem*{example}{Example}

\usepackage{algorithm}
\usepackage{algpseudocode}
\usepackage{booktabs}
\usepackage{subcaption} 
\usepackage{multirow}  
\usepackage{caption}

\let\hat\widehat
\let\tilde\widetilde

\newcommand{\E}{\mbox{$\mathbb{E}$}}
\newcommand{\R}{\mbox{$\mathbb{R}$}}
\newcommand{\bX}{\mathbf{X}}
\newcommand{\bR}{\mathbf{R}}
\newcommand{\bT}{\mathbf{T}}
\newcommand{\bY}{\mathbf{Y}}
\newcommand{\EIF}{\mathbb{EIF}}

\endlocaldefs

\begin{document}

\begin{frontmatter}
\title{Masking criteria for selecting an imputation model}
\runtitle{Masking criteria}

\begin{aug}
\author[A]{\fnms{Yanjiao}~\snm{Yang} \ead[label=e1]{yjyang00@uw.edu}},
\author[A]{\fnms{Daniel}~\snm{Suen}\ead[label=e2]{dsuen@uw.edu}}
\and
\author[A]{\fnms{Yen-Chi}~\snm{Chen}\ead[label=e3]{yenchic@uw.edu}}

\address[A]{
Department of Statistics,
       University of Washington\printead[presep={ ,\ }]{e1,e2,e3}}

\end{aug}

\begin{abstract}
The masking-one-out (MOO) procedure, masking an observed entry and 
comparing it versus its imputed values, is a very common procedure for comparing imputation models. 
We study the optimum of this procedure and generalize it to 
a missing data assumption and establish the corresponding semi-parametric efficiency theory. 
However, MOO is a measure of prediction accuracy, which is not ideal for evaluating an imputation model. To address this issue, 
we introduce three modified MOO criteria, based on rank transformation, energy distance, and likelihood principle, that allow us
to select an imputation model that properly account for the stochastic nature of data. 
The likelihood approach further enables  an elegant framework of learning an imputation model from the data
and we derive its statistical and computational learning theories as well as consistency of BIC model selection.
We also show how MOO is related to the missing-at-random assumption.
Finally, we introduce the prediction-imputation diagram, a two-dimensional diagram visually comparing
both the prediction and imputation utilities for various imputation models.
\end{abstract}



\end{frontmatter}

\section{Introduction}

Missing data is a common problem across various scientific disciplines, including medical research \citep{bell2014handling}, social sciences \citep{molenberghs2014handbook}, and astronomy \citep{ivezic2020statistics}. To handle missing entries in the dataset, imputation \citep{little2019statistical, kim2021statistical, grzesiak2025needdozensmethodsreal} is a popular approach that is widely accepted in practice. An imputation model generates plausible values for each missing entry, transforming an incomplete dataset into a complete one. The critical importance of this task has led to the development of a wide array of imputation models, grounded in various modeling assumptions. These range from traditional approaches like hot-deck imputation \citep{little2019statistical} to more sophisticated methods such as Multiple Imputation via Chained Equations (MICE; \citealt{van2011mice}), random forest imputation \citep{stekhoven2012missforest}, techniques based on Markov assumptions on graphs \citep{yang2025markovmissinggraphgraphical}, and even generative adversarial networks \citep{yoon2018gain}.

Despite the proliferation of imputation models, the selection of an optimal imputation model for a given dataset remains a significant challenge, largely due to the unsupervised nature of the problem. Among the many proposed strategies for evaluating and selecting imputation models, masking has emerged as a particularly popular procedure \citep{gelman1998not,troyanskaya2001missing,honaker2011amelia, leek2012sva,qian2024unveiling,wang2024deep}. Masking involves intentionally creating missing values in observed entries 
to create a setting where imputation accuracy can be measured against a known ground truth.
This approach has demonstrated remarkable success and power in other domains, notably in language modeling \citep{yang2019xlnet,devlin2019bert} and image recognition \citep{vincent2010stacked,xie2022simmim, hondru2025masked} and prediction-powered inference \citep{angelopoulos2023ppi,wang2020methods}.

However, despite its practical appeal, 
there is a lack of theoretical understanding of how and why the masking approach works. 
In particular, it is unclear what is the optimal imputation model
under the masking criterion. 
Moreover, a notorious problem with masking is that
it tends to select an imputation model that ignores the stochastic nature of the data,
often resulting in selecting a model that just imputes the conditional mean.

In this paper, we provide a theoretical analysis of the masking procedure
and derive the corresponding optimal imputation model and show how it is associated with a missing-not-at-random assumption. 
To address the issue of ignoring the stochasticity of the data, we propose three modifications:
masking with rank transformation, energy distance, and a likelihood approach. 
The rank transformation and energy distance are easy to implement and
the likelihood approach offers a tractable framework 
for learning an imputation model.

\emph{Main results.} Our main results are as follows.
\begin{itemize}
\item {\bf Characterization of the masking optimum.} We provide a precise characterization of the minimizer under the conventional masking procedures (Theorems~\ref{thm::moo} and \ref{thm::mko}) and associate the masking optimum to a missing-not-at-random assumption (Proposition \ref{prop::prob}).
\item {\bf Semi-parametric efficiency.} We derive the underlying efficient influence function 
 (Theorem~\ref{thm::EIF}), which yields a multiply-robust estimator (Theorem~\ref{thm::MR}). 
\item {\bf Distributional imputation criteria.} 
We propose two criteria based on rank transformations and energy distance (Section \ref{sec::QT})
that lead to optimal distributional imputations (Theorems~\ref{thm::MOORT} and \ref{thm::MOOEN}). 
\item {\bf Likelihood-based imputation learning.} 
We introduce a masking likelihood framework (Section \ref{sec::learning}) for training parametric imputation models. We establish its theoretical guarantees, including asymptotic normality of the resulting estimators (Theorem~\ref{thm::MLE}), convergence of gradient ascent (Theorem~\ref{thm::GD}), parameter recovery under MCAR (Theorem \ref{thm::mcar}), and model selection consistency (Theorem \ref{thm::BIC}).
\item {\bf Connection to MAR.} We establish a novel link between masking and the missing-at-random (MAR) assumption. We show that under monotone missingness, the masking procedure is related to the available-case missing value assumption (Proposition \ref{prop::acmv}).
\item {\bf Prediction-Imputation diagram.}
We introduce the prediction-imputation diagram (PI diagram) in Section \ref{sec::empirical}
as a 2D visualization for comparing multiple imputation models' performance in terms of prediction and imputation. 

\end{itemize}

\subsection{A probability framework for imputation}

We first introduce  probability notations for the missing data problem.
Let $X\in\R^d$ be the 
vector of study variables of interest, such that any of its components may be missing.
Let $R\in\{0,1\}^d$ be the response vector
where $R_j=1$ if we observe $X_j$ and $R_j=0$ otherwise.
We use the notation $X_R = (X_j: R_j=1)$
to denote the observed variables under response pattern $R$. 
Let $\overline{R} = 1_d-R$ be the binary vector representing
the missing variables under $R$ and $1_d = (1,1,\cdots,1)\in \{0,1\}^d$  is the vector of $1$'s that corresponds
to the complete cases.  $R_{-j} = (R_\ell: \ell \neq j)$ is the vector without $j$-th element.
Table \ref{tab::non} provides an example of missing data with the corresponding response vector $R$.

\begin{table}
\center
\begin{tabular}{lllll}
\toprule
ID  & $X_1$ & $X_2$ & $X_3$ & $R$ \\
\midrule
001 & 13    & 0     & 2.2   & 111 \\
002 & 7     & \texttt{NA}    & 2.7   & 101 \\
003 & \texttt{NA}    & \texttt{NA}    & 2.5   & 001 \\
004 & 2     & 1     & 1.3   & 111 \\
005 & 8     & 0     & \texttt{NA}    & 110 \\
006 & \texttt{NA}    & 0     & \texttt{NA}    & 010 \\
007 & 15    & 1     & 2.2   & 111 \\
008 & \texttt{NA}    & 1     & 1.7   & 011\\
\bottomrule
\end{tabular}
\caption{An example of non-monotone missing data with three study variables $X_1,X_2,X_3$
and the corresponding response vector $R$.
For ID=001, $X_R = (X_1,X_2,X_3) =  (13,0,2.2)$ whereas for ID=002, $X_R = (X_1, X_3) = (7,2.7)$.
The extrapolation density of ID=002 is $p(x_2|X_1=7,X_3=2.7, R=101)$.
}
\label{tab::non}
\end{table}

Under this setup, the PDF/PMF $p(x_r,r) = p(x_r|R=r) P(R=r)$
is the observed-data distribution that describes the distribution of the observed entries.
By the decomposition
$$
p(x,r) = p(x_r, x_{\bar r}, r) = p(x_{\bar r}|x_r, r) p(x_r,r),
$$
the joint distribution of $(X,R)$ can be expressed as 
$p(x_{\bar r}|x_r, r) $ multiplied by the observed-data distribution $p(x_r,r)$.
The distribution 
$
p(x_{\bar r}|x_r,r)
$
is the distribution of the \emph{unobserved variables} under pattern $R=r$ and observed entries $x_r$.
$p(x_{\bar r}|x_r,r)$ is also known as the \emph{extrapolation distribution/density} \citep{little1993pattern}. 


\emph{Out-of-sample (OOS) imputation.}
An imputation model is capable of performing out-of-sample (OOS) imputation if, after being trained on one dataset, 
it can impute missing entries for a new observation (e.g., from another dataset) without retraining.
This is a desirable property particularly in the modern era of big data
because we may train an imputation model on a massive dataset with powerful computers and then
use it to impute on another dataset.
An imputation model with the OOS property can be formalized as a model to the true extrapolation density $p(x_{\bar r}|x_r,r)$. 
We therefore define an imputation model in this paper as a model
$q(x_{\bar r}|x_r,r)$. 
Throughout the entire paper, we assume that the imputation models are given and non-random except for 
the likelihood method in Section \ref{sec::learning}.
This mathematical form of imputation model is particularly useful
because it enables us to analyze statistical properties of an imputation procedure.
Many imputation methods have the OOS imputation property such as hot-deck imputation, MAR with parametric models \citep{little2019statistical}, 
pattern graphs \citep{chen2022pattern},  Markov missing graph,
and GAIN \citep{yoon2018gain}.
However, 
some popular methods such as MICE \citep{van2011mice} 
cannot perform OOS imputation without retraining the model.
In this paper, we only consider imputation models with OOS property.

\subsection{Outline}

In Section \ref{sec::MOO}, we formally introduce the mask-one-out (MOO) procedure and investigate its theoretical properties. These include the characterization of its optima, the probability model implied by MOO, and the associated semi-parametric efficiency theory.
In Section \ref{sec::QT}, we demonstrate the limitations of the MOO procedure and propose two remedies: the rank transformation and the energy distance. We show that these modified criteria select imputation models that properly account for the stochastic nature of the data.
In Section \ref{sec::learning}, we present a statistical learning framework based on a masking log-likelihood function. This framework enables us to learn an imputation model directly from the data, and we study the underlying theoretical properties.
In Section \ref{sec::mono}, we analyze the monotone missing data setting and draw meaningful connections between the masking procedure and the MAR assumption.
Finally, in Section \ref{sec::empirical}, we introduce the prediction-imputation (PI) diagram as a two-dimensional visualization tool. We use this diagram to summarize MOO risks in a simulation study and a real-data analysis.
Proof of theoretical results are deferred to Appendix \ref{sec::proofs}.

\section{Mask-one-out and its theoretical properties}	\label{sec::MOO}


The mask-one-out (MOO) is a procedure of intentionally masking one observed variable
at a time,  imputing the masked value, and  comparing the imputed value to the observed value. 
It shows some similarity to the conventional leave-one-out cross-validation method, so we call it mask-one-out. 
This approach has appeared in various works \citep{gelman1998not,troyanskaya2001missing,honaker2011amelia, leek2012sva,qian2024unveiling,wang2024deep}
but there is very limited theoretical understanding about it.

Before formally describing the MOO procedure,
we first introduce some notations.
For $j\in\{1,2,\cdots, d\}$ and $r \in \{0,1\}^d$,
we denote the binary vector
$r\ominus e_j  \in\{0,1\}^d$
to be the same as $r$ except that the $j$-th element is set to be $0$,
where $e_j $ is the 
the $j$-th standard basis vector (a vector of zeros with a one at the $j$-th position).
Similarly, $r\oplus e_j$ is the same as $r$ except that the $j$-th element is set to be $1$.
Also, we use the notation
$j \in r$ to represent $j \in \{k: r_k = 1\}$.

Here is a formal description of the MOO procedure.
Let  $(X_r = x_r, R=r)$ be an observation. 
For each variable $j \in r$, i.e., this variable is observed in $(x_r,r)$, 
we generate $\hat x_j \sim q(x_j|x_{r\ominus e_j}, r\ominus e_j)$, where $q$ is an imputation model that we want to evaluate its performance.
Then we compute the loss of this imputation
$
L(x_j, \hat x_j);
$
a very common example of such loss is the square loss $L(x_j, \hat x_j) = (x_j - \hat x_j)^2$.
By doing so for every $j\in r$, we obtain a loss for the imputation model $q$ 
for this observation 
\begin{equation}
L(q| x_r,r) =  \sum_{j\in r} L(x_j, \hat x_j).
\label{eq::L1}
\end{equation}
When we have many observations, we compute the total loss of all observations.

To avoid conflicts of notations, we use the boldface variables
$$
(\bX_{1, \bR_1}, \bR_1),\cdots, (\bX_{n,\bR_n}, \bR_n)
$$
to denote our observed data. Namely, each $(\bX_{i,\bR_i}, \bR_i )$
is an independent and identically distributed (IID) copy of $(X_R, R)$.
We write $\bX_{ij}$ to refer to the $j$-th variable in the $i$-th observation.
When applied to the entire dataset, 
the MOO procedure  sums over the loss evaluated at every observation, leading to an overall risk
$$
\hat{\mathcal{E}}_n(q ) = \frac{1}{n}\sum_{i=1}^n L(q| \bX_{i,\bR_i},\bR_i)
$$
for the imputation model $q$. 
To reduce the Monte Carlo errors due to imputing each $\hat \bX_{i,\bR_i}$, we may repeat the computation of $\hat{\mathcal{E}}_n(q )$ multiple times
and take the average of them.
The MOO procedure is summarized in Algorithm \ref{alg::moo}.
To avoid confusion with other MOO procedures introduced later, 
we call the procedure in Algorithm \ref{alg::moo} the \emph{naive MOO}.
Note that in practice, we often standardize the observed entries first so that 
the loss values are of the same order.

When we have multiple imputation models $q_1,\cdots, q_K$,
we  apply this procedure to each of them, which leads to 
$$
\hat{\mathcal{E}}_n(q_1),\cdots, \hat{\mathcal{E}}_n(q_K).
$$
These values are used as a criterion for selecting the optimal imputation model;
generally, we choose the model that has the smallest loss.

\begin{algorithm}
\caption{(Naive) Mask-one-out (MOO) procedure}
\label{alg::moo}
\begin{algorithmic}
\State  \textbf{Input:} Imputation model $q$.
\begin{enumerate}
\item For $i=1,\cdots, n$, we do the following:
	\begin{enumerate}
	\item For each $j\in \bR_i$: 
	\begin{enumerate}
	\item We mask the observed entry $\bX_{ij}$ and update the response pattern to be $\bR_i \ominus e_j$ (pretending $\bX_{ij}$ is missing).  
	\item We generate $ \hat \bX_{ij}$ by sampling from the conditional distribution 
	$$
	q(x_{j}|\bX_{i, \bR_i \ominus e_j}, \bR_i \ominus e_j)\equiv q(x_{j}|X_{\bR_i \ominus e_j} = \bX_{i, \bR_i \ominus e_j}, R= \bR_i \ominus e_j).
	$$
	Namely, we treat the data as if $\bX_{ij}$
	is a missing value and attempt to impute it.
	\item Compute the  loss $L( \bX_{ij}, \hat\bX_{ij}).$
	\end{enumerate} 
	\item Compute the total loss for this individual: $L(q| \bX_{i, \bR_i},\bR_i) = \sum_{j\in \bR_i}L( \bX_{ij}, \hat\bX_{ij})$.
	\end{enumerate}
\item Compute the risk of the imputation model $q$ as 
$$
\hat{\mathcal{E}}_n(q ) = \frac{1}{n}\sum_{i=1}^n L(q| \bX_{i,\bR_i},\bR_i).
$$
\item (Optional) Repeat the above procedure multiple times and take the average of $\hat{\mathcal{E}}_n(q )$
to reduce the Monte Carlo errors.
\end{enumerate}
\end{algorithmic}
\end{algorithm}

The MOO idea is based on the feature that
the observation should remain somewhat similar even if we mask one entry. 
So imputing the masked value and comparing with the actual value
may be a reasonable metric for evaluating the effectiveness of an imputation model.
Sometimes we may be interested in the imputation performance on
a specific variable. 
The MOO procedure can be modified to compute the loss for a particular variable; see Appendix \ref{sec::moo::v}
for more details.
Moreover, we may mask multiple variables at the same time; we provide a detailed discussion in Appendix \ref{sec::MKO}. 

\begin{example}
Consider the data in Table \ref{tab::non}
and let $q$ be an imputation model.
When we apply the MOO to ID=001, 
we will perform imputation on $\bX_{1,\bR_1} =(\bX_{11}, \bX_{12}, \bX_{13}) = (13,0,2.2)$ via
\begin{align*}
\hat \bX_{11} &\sim q(x_1|X_2= 0,X_3 = 2.2, R= 011),\\
\hat \bX_{12} &\sim q(x_2|X_1= 13,X_3 = 2.2, R= 101),\\
\hat \bX_{13} &\sim q(x_3|X_1=13, X_2= 0,R=110)
\end{align*}
and then compute the losses
$$
L(13, \hat \bX_{11}) + L(0, \hat \bX_{12}) + L(2.2, \hat \bX_{13}).
$$
For the individual ID=002, $\bX_{2, \bR_2} = (\bX_{21}, \bX_{23}) = (7, 2.7)$, 
we mask $\bX_{21}, \bX_{23}$ separately and impute them  via 
\begin{align*}
\hat \bX_{21} &\sim q(x_1|X_3 = 2.7, R= 001),\\
\hat \bX_{23} &\sim q(x_3|X_1=7,R=100)
\end{align*}
and compute the risk 
$$
L(7, \hat \bX_{21}) + L (2.7, \hat \bX_{23}).
$$
Namely, when imputing $\bX_{21}$, we mask $\bX_{21}$ and pretend the observation is $(\texttt{NA}, \texttt{NA}, 2.7)$. 
\end{example}

\subsection{Optimal imputation value}

The (naive) MOO approach in Algorithm \ref{alg::moo} can be viewed as a risk minimization procedure. 
In particular, the output of Algorithm \ref{alg::moo} is the quantity
$$
\hat{\mathcal{E}}_n(q ) = \frac{1}{n}\sum_{i=1}^n L(q| \bX_{i,\bR_i},\bR_i) = \frac{1}{n}\sum_{i=1}^n \sum_{j\in R_i} L(\bX_{ij}, \hat \bX_{ij}),
$$
which can be interpreted as an empirical risk. 
The corresponding test risk (also called population risk or true risk) is the following population quantity 
\begin{equation}
\begin{aligned}
\mathcal{E}(q) = \E\{\bar L(q| \bX_{i,\bR_i},\bR_i)\} &= \sum_{r:r\neq 1_d}\int \bar L(q|x_{r} , r) p(x_r,r)dx_r,\\
\bar L(q|x_{r} , r) & = \sum_{j \in  r} \int L(x_j, x_j') q(x'_j|x_{r\ominus e_j},R=r\ominus e_j)dx'_j,
\end{aligned}
\label{eq::MOO::TR}
\end{equation}
where $p(x_r,r)$ is the observed-data distribution. 
The quantity $\bar L(q|x_{r} , r)$ is the expectation of $L(q|x_{r},r)$ defined in equation \eqref{eq::L1}
that has no Monte Carlo errors.
Clearly, $\hat{\mathcal{E}}_n(q )$ is the empirical (and one-sample Monte Carlo approximation) version of $\mathcal{E}(q)$.

Since $\mathcal{E}(q)$ is  the population risk corresponding to the naive MOO procedure, its minimizer provides key insights into the properties of MOO. The following theorem characterizes the minimizer of the population risk $\mathcal{E}(q)$.


\begin{theorem}[Optimal imputation value of MOO]
\label{thm::moo}
For an observation $(x_r,r)$, let $j \in \bar r$ be the index of an unobserved variable. 
For the missing variable $x_j$, 
\begin{equation}
\hat x^*_j = {\sf argmin}_{\theta} \int L( x_j,\theta ) p(x_j|x_{r}, r\oplus e_j)dx_j
\label{eq::opt::value}
\end{equation}
is the optimal imputation value under the population risk $\mathcal{E}(q)$.
Namely, for the observation $(x_r,r)$, the optimal imputation model will impute
the missing variable $x_j$ with $\hat x^*_j$ for every $j\in \bar r$.
\end{theorem}

Theorem \ref{thm::moo} implies that
if we use the square loss $L(a,b) = (a-b)^2$,
$$
\hat x^*_j = \E(X_j|X_r = x_r, R = r\oplus e_j)
$$
will be the mean value of the conditional distribution $p(x_j|x_r, r\oplus e_j)$.
So the optimal $q$ will be a point mass at $\hat x^*_j$.
If we use the absolute loss $L(a,b) = |a-b|$,
$\hat x^*_j $ will be the median of $p(x_j|x_r, r\oplus e_j)$.
Also, Theorem \ref{thm::moo} shows that the optimal imputation model under MOO
is a deterministic imputation that ignores the data's stochastic nature.
Thus, the  MOO criterion in Algorithm \ref{alg::moo} is like
a measure of prediction performance, not  a measure of imputation performance, so it is not ideal for comparing
imputation models \citep{van2018flexible,naf2023imputation, grzesiak2025needdozensmethodsreal}.

\begin{example}
Suppose we have three variables $X = (X_1,X_2,X_3)^T$
and we have an observation $X = (\texttt{NA},  \texttt{NA}, z_3)$ with $R = 001$. 
Assume that we use the square loss $L(a,b) = (a-b)^2$.
Then the optimal imputation model will impute $x_1$ and $x_2$ with
$\E(X_1|X_3 = z_3, R = 101)$ and $\E(X_2|X_3 = z_3, R=011)$, respectively. 
Here is a high-level idea on why this is the optimal imputation value. 
The imputation model on $X = (\texttt{NA},  \texttt{NA}, z_3)$ can be written as $q(x_1,x_2|X_3=z_3, R=001)$. 
Under the MOO procedure, this imputation model will be used in two scenarios. 
The first scenario is the case where $R=101$. In this case, when we mask $X_1$,
the response pattern becomes $R=001$ and we will use the marginal $q(x_1|X_3=z_3, R=001)$
to impute $X_1$ and attempt to minimize the square loss. 
The unmasked value follows from the distribution of $p(x_1|x_3,R=101)$,
so under the square loss, the minimizer is the conditional mean of $p(x_1|x_3,R=101)$. 
The second scenario is $R=011$ and when we mask $X_2$, we obtain the response pattern $R=001$.
So by the same argument, the minimization procedure leads to the conditional mean of $p(x_2|x_3,R=011)$. 
\end{example}

\subsection{Optimal imputation model}

Theorem \ref{thm::moo} shows that 
the optimal imputation value depends on the loss function we use. 
However, the
distribution $p(x_j|x_r, r\oplus e_j)$
appears in equation \eqref{eq::opt::value}
is independent of the loss function. It can therefore be defined as the optimal loss-agnostic target distribution.
Therefore, we call the marginal imputation model (for variable $x_j$ such that $r_j = 0$)
\begin{equation}
q(x_j|x_r, r) = p(x_j|x_r, r\oplus e_j), 
\label{eq::MOO::marginal}
\end{equation}
the \emph{optimal MOO (marginal) imputation model}.
Note that equation \eqref{eq::MOO::marginal} only describes
an imputation model marginally
for  each individual variable. 
It does not specify any dependency among those variables to be imputed.

The imputation model  in equation \eqref{eq::MOO::marginal}
has another nice interpretation. For pattern $R=r$ and $x_j$ is a missing variable under $R=r$, the
pattern $R=r\oplus e_j$ is the response pattern most similar to $r$ with variable $x_j$ being observed. 
Therefore, it is reasonable to expect that the conditional distribution $p(x_j|x_r, r\oplus e_j)$ would be
similar to the imputation distribution $p(x_j|x_r, r)$.
The imputation model in equation \eqref{eq::MOO::marginal} just equates these two distributions.

With equation \eqref{eq::MOO::marginal},
we define the collection of \emph{optimal imputation models for the MOO procedure} as 
\begin{equation}
\mathcal{Q}^*_{MOO} =\{q: q(x_j|x_r, r) = p(x_j|x_r, r\oplus e_j), \quad \forall j \in \bar r,\quad r \in\{0,1\}^d\}.
\label{eq::MOO::opt::set}
\end{equation}
By construction, any imputation model
in $\mathcal{Q}^*_{MOO}$ 
satisfies equation \eqref{eq::MOO::marginal}
and can be used to construct an optimal imputation value via equation \eqref{eq::opt::value}
when the loss function is specified.
Later we will discuss three methods for finding an imputation model in $\mathcal{Q}^*_{MOO}$ (Sections \ref{sec::QT} and \ref{sec::learning}).

The collection $\mathcal{Q}^*_{MOO}$ is not 
an empty set. 
Here is a useful example inside $\mathcal{Q}^*_{MOO}$:
\begin{equation}
q_{PM}(x_{\bar r}|x_r,r) = \prod_{j\in \bar r} p(x_j|x_r, r\oplus e_j).
\label{eq::MOOPM}
\end{equation}
Namely, the imputation model $q_{PM}$
imputes every missing entry independently from each other
by the marginal $p(x_j|x_r, r\oplus e_j)$. 
We call this imputation model mask-one-out product model (MOOPM).
The product model in equation \eqref{eq::MOOPM} will be 
particularly useful when learning the imputation model
from the data; see Section \ref{sec::sepprod} for more details.


\subsection{Probability statement for the optimum}


The collection $\mathcal{Q}^*_{MOO}$ in equation \eqref{eq::MOO::opt::set} defines the imputation models that are optimal and loss-agnostic under the MOO. This set identifies $p(x_j|x_r, r\oplus e_j)$ as the target (marginal) imputation density. This target, however, is only equal to the true, unobserved extrapolation density $p(x_j|x_r, r)$ if the data-generating process $p(X,R)$ satisfies a specific assumption. The following proposition precisely characterizes this implicit assumption as a formal conditional independence statement.

\begin{proposition}
\label{prop::prob}
Any optimal imputation model in $\mathcal{Q}^*_{MOO}$ must satisfy the following conditional independence:
for  every $R$ and $j \in {\bar R }$,
\begin{equation}
X_j \perp R_j | X_{R}, R_{-j},
\label{eq::CI}
\end{equation}
where $R_{-j} = (R_\ell: \ell \neq j)$.
Also, the above conditional independence statement is equivalent to the following:
for  every $R$ and $j \in { R }$,
\begin{equation}
X_j \perp R_j | X_{R\ominus e_j}, R_{-j}.
\label{eq::CI2}
\end{equation}

\end{proposition}

Equation \eqref{eq::CI} shows similarity to the itemwise conditionally independent nonresponse (ICIN; \citealt{sadinle2017itemwise}; also known
as the no-self-censoring/NSC; \citealt{malinsky2022semiparametric}) condition
but there is a key difference. In ICIN/NSC, the probability statement is:
$$
X_j \perp R_j | X_{-j}, R_{-j}.
$$
Namely, ICIN/NSC require conditioning on \emph{all other variables}.
On the other hand, equation \eqref{eq::CI} only requires
conditioning on the observed variable $X_R$ under pattern $R$.
This key difference indicates that the conditional operation in equation \eqref{eq::CI}
cannot be expressed in a directed acyclic graph. 
So the optimal imputation model is not a missing data directed acyclic graph \citep{mohan2013graphical, nabi2020full}.


Since we know the optimal imputation model in $\mathcal{Q}^*_{MOO}$
is generally non-unique (only unique for each marginal), 
equation \eqref{eq::CI} is not a nonparametric identification assumption \citep{robins2000sensitivity}. 
To obtain a unique imputation model, 
we need to add additional assumptions that do not conflict with the observed data
as well as equation \eqref{eq::CI}.
One such additional assumption is as follows. 
For any $j,k \in \bar R$, we assume that
\begin{equation}
X_j\perp X_k|X_{R}, R.
\label{eq::PM::pair}
\end{equation}
One can easily see that equation \eqref{eq::PM::pair} does not conflict with equation \eqref{eq::CI}
nor the observed data. 
Also, with equations
\eqref{eq::CI} and \eqref{eq::PM::pair},
there is a unique imputation distribution--the MOOPM model in equation \eqref{eq::MOOPM}.
Thus, equations \eqref{eq::CI} and \eqref{eq::PM::pair} 
together form a nonparametric identifying assumption.

\subsection{Efficiency theory}

Because equations \eqref{eq::CI} and \eqref{eq::PM::pair}
form a nonparametric identification assumption \citep{robins2000sensitivity},
we are able to construct its inverse probability weighting (IPW) estimator
and regression adjustment estimator.
Moreover, we will be able to study
the underlying semi-parametric efficiency theory.
To simplify the problem, we consider estimating the mean of the first variable
and study its efficiency theory.
Namely, our parameter of interest is $\mu \equiv \E[X_1]$.

For any $r$ such that $r_1 = 0$ (i.e., $X_1$ is missing),
define the odds
$$
O_1(x_{r}, r) =  \frac{P(R_1=0|x_{r}, r_{-1})}{P(R_1=1|x_{r}, r_{-1})} = \frac{p(R_1=0,x_{r}, R_{-1}=r_{-1})}{p(R_1=1,x_{r}, R_{-1}=r_{-1})}
=  \frac{p(x_{r}, R = r)}{p( x_{r}, R=r\oplus e_1)}.
$$
Clearly, this odds is identifiable. 
We just need to perform a two-sample comparison over the variables $x_{r}$ where
the first sample is $R_1=0, R_{-1} = r_{-1}$ and the second sample is $R_1 = 1, R_{-1} = r_{-1}$.
This can be done easily by either estimating the density ratio \citep{sugiyama2012density}
or training a generative classifier for the binary outcome $R_1$ given $(X_r, r_{-1})$.

This odds quantity has an interesting property: for any $r$ with $r_1 =0$,
\begin{align*}
\E[X_1 O_1(X_{r}, r) I(R=r\oplus e_1)] & = \int x_1O_1(x_{r}, r) p(x_{r\oplus e_1}, r\oplus e_1)dx_1dx_r\\
  &=\int x_1  O_1(x_{r}, r)p(x_1|x_{r}, r\oplus e_1)  p(x_{r}, r\oplus e_1)dx_1 dx_r\\
    &=\int x_1  p(x_1|x_{r}, r\oplus e_1)  p(x_{r}, R=r)dx_1 dx_r\\
        &\overset{\eqref{eq::CI}}{=}\int x_1  p(x_1|x_{r}, r)  p(x_{r}, R=r)dx_1 dx_r\\
        & =\E[X_1  I(R=r)],
\end{align*}
which is the expected value of the missing value $X_1$ under $R=r$ when $R_1 = r_1 = 0$.

With the above result,
we  decompose 
\begin{align*}
\mu \equiv \E(X_1) & = \E(X_1 I(R_1=1)) + \E(X_1 I(R_1 = 0))\\
& = \E(X_1 I(R_1=1)) + \sum_{r_{-1}}\E(X_1 I(R_1 = 0, R_{-1} = r_{-1}))\\
& = \E(X_1 I(R_1=1)) + \sum_{r: r_1 = 0}\E(X_1 O_1(X_{r}, r) I(R = r\oplus e_1) ),
\end{align*}
which implies the IPW estimator of $\mu$ via the plug-in approach.

In addition to the IPW estimator,
we are able to construct a regression adjustment estimator.
Let
$$
\mu_1(X_r,r) \equiv \E[X_1 | X_r, R=r] 
$$
be the outcome regression model for the mean of $X_1$ given $X_r$ and $R=r$ with $r_1=0$.
One can easily see that 
$$
\E[\mu_1(X_r,r)I(R=r)] = \int \E[X_1 | X_r, R=r] p(x_r, r)dx_r  = \E[X_1I(R=r)]
$$
so the function $\mu_1(X_r,r)$ can be interpreted as an outcome regression model of $\mu_1$ under pattern $R=r$.
While $\mu_1(X_r,r)$ is in general unidentifiable from the data,
equation \eqref{eq::CI} 
implies
$$
\mu_1(X_r,r) \stackrel{\eqref{eq::CI}}{=} \mu_1(X_r,r\oplus e_1) \equiv  \E[X_1 | X_r, R=r\oplus e_1] ,
$$
which can be identified from the data.
So we can simply use the plug-in approach to construct a regression adjustment estimator. 

Now we have two sets of nuisance functions $O_1(x_{r}, r)$
and $\mu_1(X_r,r\oplus e_1)$ for each $r$ with $r_1 = 0$. 
One may expect from the semi-parametric efficiency theory that
some combination of these two nuisances leads to the efficient influence function.
The theorem below provides a positive answer to this.

\begin{theorem}[Efficient influence function for marginal mean]\label{thm::EIF}
		The efficient influence function for $\mu\equiv\E[X_1]$ under equation \eqref{eq::CI} is given by
		\begin{equation}
		\begin{aligned}
			\EIF(\mu) = I(R_1=1)X_1  
			 &+ \sum_{r:r_1=0} \biggr[I(R=r\oplus e_1) O_1(X_r,r) (X_1 - \mu_1(X_r,r\oplus e_1)) \\
			&+ I(R=r)  \mu_1(X_r,r\oplus e_1)\biggr] - \mu. 
		\end{aligned}
		\label{eq:marginalEIF}
		\end{equation}
	\end{theorem}

Here is an interesting fact: we only need equation \eqref{eq::CI} for constructing the EIF of $\mu_1$.
This is because the marginal mean $\mu_1$ only require 
depends on a single variable. As long as we can identify its marginal distribution,
we can identify this parameter of interest.
Therefore, equation \eqref{eq::CI} is enough and we do not need equation \eqref{eq::PM::pair}.
However, if the parameter of interest involves two or more study
variables, then equation \eqref{eq::CI} is not enough and we need additional assumptions  such as equation \eqref{eq::PM::pair}.

Based on equation \eqref{eq:marginalEIF}, 
we can construct a plug-in estimator
\begin{align*}
\hat \mu_{MR} = \frac{1}{n}\sum_{i=1}^n I(\bR_{i,1}=1)\bX_{i,1}  
		 &+ \sum_{r:r_1=0} \biggr[I(\bR_i=r\oplus e_1) \hat O_1(\bX_{i,r},r) (\bX_{i,1} - \hat \mu_1(\bX_{i,r},r\oplus e_1)) \\
			&+ I(\bR_i=r)  \hat \mu_1(\bX_{i,r},r\oplus e_1)\biggr],
\end{align*}
where $\hat O_1$ and $\hat \mu_1$ are estimators corresponding to $O_1$ and $\mu_1$. 
The estimator $\hat \mu_{MR}$
is a multiply-robust estimator, as illustrated in the following theorem. 

	\begin{theorem}[Multiple robustness]\label{thm::MR}
		The plug-in estimator $\hat \mu_{MR}$ is $(2^{d-1}-1)$-multiply robust.
	\end{theorem}
	
	The meaning of multiply-robustness in Theorem~\ref{thm::MR}
	is as follows. 
	For each $r$ with $r_{1} = 0$, 
	we have a pair of nuisance parameters: $O_1(x_r,r)$ and $\mu_1(x_r,r\oplus e_1)$.
	We need at least one of the two nuisances to be correct
	to obtain the consistency of the estimator. 
	Using a bit more algebra, let $\mathcal{F}_{O, r}$ be the collection of distributions such that $O_1(x_r,r)$ is the correct model
	and $\mathcal{F}_{\mu, r}$ be the collection of distributions such that $\mu_1(x_r,r\oplus e_1)$ is the correct model. 
	The multiply-robustness means that as long as the true distribution that generates our data falls within the intersection
	$
	\bigcap_{r: r_1 = 0 } \left(\mathcal{F}_{O, r}\cup \mathcal{F}_{\mu, r}\right),
	$
	our estimator is consistent.
	Since there will be a total of $2^{d-1}$ patterns  for $\{r: r_1 = 0\}$ and we do not need to model the case $r=0_d$ (all variables are missing, which can be
	estimated by the empirical ratio), 
	the model is $(2^{d-1}-1)$-multiply robust.


\section{Evaluating imputation via rank transformation and energy distance}	\label{sec::QT}

\subsection{Limitation of loss minimization}

While the naive MOO in Algorithm \ref{alg::moo} is easy to implement, it has a severe limitation:
the optimal imputation model is a deterministic imputation that ignores the stochastic nature of data.
While the deterministic imputation may have a lower variance for a particular parameter of interest,
it could lead to a biased estimate
when the parameter of interest does not align with the loss function \citep{von2025imputing}.
The following is a concrete example illustrating this problem.

\begin{example}[Failure of deterministic imputation]
Consider a simple missing data problem where we have two study variables $(X,Y)\in\R^2$
and $Y$ is subject to missing and $X$ is always observed.
Let $R\in\{0,1\}$  be the response indicator for $Y$, i.e., $R=1$ if $Y$ is observed. 
In this case, equation \eqref{eq::CI} will imply $Y\perp R|X$, 
which agrees with the conventional missing-at-random assumption. 
One can easily show that under the square loss, the optimal imputation for $Y$ given $X$ and $R=0$ under the naive MOO
is $ \mu_1(X) = \E(Y|X, R=1)$. 
Suppose the parameter of interest is $\theta = \E(Y^2)$, the second moment of $Y$. 
Based on the imputation, our estimate of $\theta $ will be 
$$
\frac{1}{n}\sum_{i=1}^n \left(R_i Y_i^2 + (1-R_i) \mu^2_1(X_i)\right).
$$
However, it is easy to see that the imputed part
$$
\E(\mu^2_1(X_i)) = \E(\E^2(Y|X= X_i))\leq \E(\E(Y^2|X= X_i))  = \E(Y^2).
$$
The difference of the inequality is 
$$
\E(\E(Y^2|X= X_i)) - \E(\E^2(Y|X= X_i)) = \E({\sf Var}(Y|X=X_i))\geq 0.
$$
The equality holds only if ${\sf Var}(Y|X=X_i) = 0$.
Thus, the mean imputation gives a biased estimate for $\theta$.
\end{example}

\subsection{Masking with rank transformation}
To resolve the above issue, we need a procedure where the minimizer is
a stochastic imputation rather than a deterministic imputation.
And ideally, such minimizer shall recover an imputation model in $\mathcal{Q}^*_{MOO}$.
To obtain such an imputation model, we propose a procedure called \emph{masking-one-out with rank transformation (MOORT)}. 
The procedure is summarized in Algorithm \ref{alg::MOORT}.

\begin{algorithm}
\caption{Masking-one-out with rank transformation (MOORT)}
\label{alg::MOORT}
\begin{algorithmic}
\State \textbf{Input:} Imputation model $q$ and a distributional metric $d$ (e.g. Kolmogorov distance, maximal mean discrepancy).
\begin{enumerate}
\item For each individual $i=1,\cdots, n$, 
we randomly pick one observed entry $j \in \bR_i$. 
\item We mask $\bX_{ij}$, pretending it to be a missing value.
\item 
We 
sample $M$ times from the conditional distribution 
$$
q(x_j | \bX_{i, \bR_i \ominus e_j},  \bR_i \ominus e_j)\equiv q(x_j | X_{ \bR_i \ominus e_j} = \bX_{i, \bR_i \ominus e_j}, R= \bR_i \ominus e_j)
$$ to generate $M$ imputed values:
$
\hat \bX_{ij}^{(1)},\cdots, \hat \bX_{ij}^{(M)}.
$
\item We compute the empirical cumulative distribution function (EDF) of these $M$ values: $\hat G_{\bX_{ij}}(x) = \frac{1}{M}\sum_{m=1}^M I\left(\hat \bX_{ij}^{(m)} \leq x\right)$. 
\item We  compute the (normalized) rank $\hat S_i = \hat G_{\bX_{ij}}(\bX_{ij})$.
\item By doing so for every individual, we obtain $\hat S_1,\cdots, \hat S_n$ and the corresponding empirical distribution $\hat H(t; q)= \frac{1}{n}\sum_{i=1}^n I(\hat S_i\leq t)$.
\item We use metric $d$ to obtain
$
\hat{\mathcal{R}}(q) = d\left(\hat H(\cdot; q), {\sf Uni}[0,1]\right).
$
\end{enumerate}

\end{algorithmic}
\end{algorithm}

The high level idea of MOORT is that
when the imputation model is correct, $\bX_{ij}$ should be a random draw from the imputation distribution. 
Therefore, the (normalized) rank $\hat S_i$ should be 
(asymptotically) distributed as
a uniform distribution over $[0,1]$. 
The independence among different individuals allows us to compare the distribution of $\hat S_i$ to the uniform distribution.

A feature of MOORT in Algorithm \ref{alg::MOORT} is that we only pick one variable per individual
because different observed variables in a single individual may  be dependent.
Note that we may use all observed variables in the computation of MOORT. Namely, in Step 1 of Algorithm \ref{alg::MOORT},
we consider every $j\in \bR_i$. 
While this reduces the Monte Carlo errors, the resulting normalized ranks will be have a block-dependent structure
and individuals with more observed variables will have a higher weight in the final output. 
The random selection of one observed variable per individual resolve this issue with  the cost of a slightly increased Monte Carlo errors. 
Alternatively, we may perform MOORT for each variable separately and combine them together
to reduce the Monte Carlo errors.
See Appendix \ref{sec::MOORT::v} for more details.

The MOORT is related to the following multiple testing problem:
$$
H_{0,i}: Z_i \sim Q_i,
$$
where $Q_i$ is a distribution we can sample from.
Our goal is to test the global null
that 
$H_{0,i}: Z_i \sim Q_i$ is true for all $i$.
In our case, $Z_i$ is the masked variable and $Q_i$ is the corresponding imputation distribution.
Under this framework, we can rewrite the global null as 
$$
H_0: Q_i(Z_i)\sim {\sf Uni}[0,1]
$$
for all $i=1,\cdots, n$.
Thus, MOORT is a Monte Carlo method for testing the above global null.

%

\begin{theorem}[Consistency of  MOORT procedure]
\label{thm::MOORT}
Consider the MOORT procedure.
Suppose we use Kolmogorov distance in Algorithm \ref{alg::MOORT}.
When $n,M\rightarrow\infty$, 
$$
\hat{\mathcal{R}}(q)
\overset{P}{\rightarrow} 0
$$
for any $q \in \mathcal{Q}^*_{MOO}$.
\end{theorem}

The Kolmogorov distance in Theorem~\ref{thm::MOORT} can be replaced by other distributional metrics such as Wasserstein, 
maximal mean discrepancy, and energy distance.
Moreover, we may use distance derived from Anderson-Darling or Cram{\'e}r-von Mises tests.
Theorem \ref{thm::MOORT} confirms that MOORT is consistent for any $q \in \mathcal{Q}^*_{MOO}$. 
However, like the original MOO, it is insensitive to the joint dependency structure of the imputed variables.

Since the MOORT value $\hat{\mathcal{R}}(q)$ measures the departure of an imputation model
from $\mathcal{Q}^*_{MOO}$, 
the population version quantity of $-\hat{\mathcal{R}}(q)$ is a proper imputation score in the definition of \cite{naf2023imputation}
when the true full-data distribution is $ p^*(x,r) = q(x_{\bar r}|x_r, r) p(x_r,r)$ for any $q\in \mathcal{Q}^*_{MOO}$.

\subsection{Masking with energy distance}

In addition to the rank transformation, we may use the energy distance \citep{rizzo2016energy,szekely2013energy} for assessing the performance
of an imputation model \citep{grzesiak2025needdozensmethodsreal}. This is similar to the engression approach \citep{shen2025engression}.
We call this approach \emph{Masking-one-out with energy distance (MOOEN).}
Algorithm \ref{alg::MOOEN} provides a summary of this procedure.

\begin{algorithm}
\caption{Masking-one-out with energy distance (MOOEN)}
\label{alg::MOOEN}
\begin{algorithmic}
\State \textbf{Input:} Imputation model $q$.
\begin{enumerate}
\item For each individual $i=1,\cdots, n$, 
and each observed variable $j \in \bR_i$. 
\item We mask $\bX_{ij}$, pretending it to be a missing value.
\item We sample $M$ times from the conditional distribution
$$
q(x_j | \bX_{i, \bR_i \ominus e_j},  \bR_i \ominus e_j)\equiv q(x_j | X_{ \bR_i \ominus e_j} = \bX_{i, \bR_i \ominus e_j}, R= \bR_i \ominus e_j)
$$
to generate the first set:
$\hat \bX_{ij}^{(1)},\cdots, \hat \bX_{ij}^{(M)}$.
We repeat this process to generate a second, independent set:
$\hat \bX_{ij}^{\dagger(1)},\cdots, \hat \bX_{ij}^{\dagger(M)}.$
\item We compute the energy distance loss for $\bX_{ij}$ as
$$
L_{\sf EN}(q|\bX_{ij}) = \frac{1}{M}\sum_{m=1}^M \left|\bX_{ij} - \hat \bX_{ij}^{(m)}\right| - \frac{1}{2M (M-1)} \sum_{m<m'}\left|\hat \bX_{ij}^{(m)} - \hat \bX_{ij}^{\dagger(m')}\right| .
$$
\item The final MOOEN of the whole data is 
$$
\hat{\mathcal{R}}_{\sf EN}(q) = \frac{1}{n}\sum_{i=1}^n \sum_{j\in \bR_i} L_{\sf EN}(q|\bX_{ij}).
$$

\end{enumerate}

\end{algorithmic}
\end{algorithm}

The MOOEN is based on the energy distance. 
The energy distance between $P_X,P_Y$ 
is 
$$
d_{\sf EN}(P_X, P_Y) = 2\E\|X-Y\| - \E\|X-X^\dagger\|  - \E\|Y-Y^\dagger\|,
$$
where  $X,X^\dagger\sim P_X$ and $Y,Y^\dagger\sim P_Y$. 

The energy distance can be written as an expected negative score
\begin{align*}
d_{\sf EN}(P_X, P_Y) &= \E_X[-{\sf ES}(X, P_Y)]\\
-{\sf ES} (x, P_Y) & = \E\|x-Y\|- \frac{1}{2} \E\|Y-Y^\dagger\| .
\end{align*}
The quantity ${\sf ES} (x, P_Y)$ is called the \emph{energy score} \citep{gneiting2007strictly, rizzo2016energy} and $\E_X$ is the expectation with respect to $X$.

When we have a random sample $\bX_1,\cdots, \bX_n$, their empirical energy distance to $P_Y$ is 
$$
\hat {\E_X} [-{\sf ES} (X, P_Y)] = \frac{1}{n}\sum_{i=1}^n \E_Y\|\bX_i-Y\|- \frac{1}{2} \E_{Y,Y^\dagger}\|Y-Y^\dagger\|
$$
A Monte Carlo approximation to $-{\sf ES} (x, P_Y)$ is via sampling 
$$
\bY^{(1)},\cdots, \bY^{(M)}, \bY^{\dagger(1)},\cdots, \bY^{\dagger(M)} \sim P_Y
$$
and computing
$$
\tilde {-{\sf ES} (x, P_Y)} = \frac{1}{M}\sum_{m=1}^M \|x-\bY^{(m)}\| - \frac{1}{2M (M-1)} \sum_{m<m'} \|\bY^{(m)} - \bY^{\dagger(m')}\|,
$$
which is essentially the step 4 in Algorithm \ref{alg::MOOEN}.
Note that since MOO only masks one variable at a time, 
the $L_2$-norm reduces to the absolute value, 
making the computation a lot easier. 


The MOOEN and the original MOO criteria share an interesting similarity. 
The loss of $\bX_{ij}$ under original MOO (using $L_1$ loss) after averaging $M$ times and energy distance are 
\begin{align*}
\text{(MOO)}\qquad L_{\sf MOO}(q|\bX_{ij}) & =  \frac{1}{M}\sum_{m=1}^M \left|\bX_{ij} - \hat \bX_{ij}^{(m)}\right|\\
\text{(MOOEN)}\qquad L_{\sf EN}(q|\bX_{ij}) &= \frac{1}{M}\sum_{m=1}^M \left|\bX_{ij} - \hat \bX_{ij}^{(m)}\right| - \frac{1}{2M (M-1)} \sum_{m<m'}\left|\hat \bX_{ij}^{(m)} - \hat \bX_{ij}^{\dagger(m')}\right|.
\end{align*}
This second term, $-\frac{1}{2M (M-1)} \sum_{m}$, acts as a reward for stochasticity. 
A deterministic imputation model (a point mass) has zero internal variance, so this term is 0, and it receives no reward. A stochastic model receives a `bonus' (a lower, i.e., better, score) proportional to its internal variance. Because energy score is a proper scoring rule \citep{gneiting2007strictly}, this bonus is maximized when the imputation distribution's variance matches the true data-generating variance.



Since the energy distance is a distance of distributions, 
we have the following consistency result for the MOOEN procedure. 
\begin{theorem}[Consistency of MOOEN procedure]
\label{thm::MOOEN}
Consider the MOOEN procedure. 
When $n,M\rightarrow\infty$, 
$$
\hat{\mathcal{R}}_{\sf EN}(q)
\overset{P}{\rightarrow} 0
$$
for any $q \in \mathcal{Q}^*_{MOO}$.
\end{theorem}

Note that energy distance is not the only possible option, other scoring criteria \citep{gneiting2007strictly}
and distributional distances are applicable. 
We choose the energy distance for its similarity to the MOO under $L_1$ loss.

\section{Learning imputation model with MOO likelihood}	\label{sec::learning}

While the rank transformation and energy distance in Section \ref{sec::QT} can be used to compare multiple imputation models,
they are not ideal for training an imputation model because
we need many Monte Carlo evaluations.
To resolve this issue, we introduce a likelihood method based on the masking procedure.

We assume that the imputation model $q  = q_\theta$ is parameterized by $\theta$. 
This means that given $\theta$ and
any $(x_r,r)$, we are able to impute the missing variables $x_{\bar r}$ by generating from $q_\theta(x_{\bar r}|x_r,r)$.
In this setup, training the imputation model is the same as learning the underlying parameter $\theta$.

Let $(\bX_{i,R_i}, \bR_i)$ be an observation.
For the imputation model $q_\theta$,
we define its \emph{MOO log-likelihood function} to be 
\begin{equation}
\begin{aligned}
\ell(\theta|\bX_{i,\bR_i},\bR_i) 
&= \sum_{j \in \bR_i} \log q_\theta(x_j = \bX_{ij}| x_r = \bX_{i, r},  r= \bR_i\ominus e_j)\\
&\equiv \sum_{j \in \bR_i} \log q_\theta(x_j = \bX_{ij}|X_{\bR_i\ominus e_j} = \bX_{i, \bR_i\ominus e_j}, R = \bR_i\ominus e_j).
\end{aligned}
\label{eq::MOO::likelihood}
\end{equation}
Note that the expression 
$
q_\theta(x_j = \bX_{ij}| x_r = \bX_{i, r},  r= \bR_i\ominus e_j)
$
will be used frequently in the rest of the paper 
since it avoids confusion when taking expectation.


$\ell(\theta|\bX_{i,\bR_i},\bR_i)$ is the logarithm of the predictive probability (density) on the masked variable
given the other observed variables based on the idea of MOO. 
The {MOO log-likelihood} of the entire data is 
\begin{equation}
\begin{aligned}
\ell_n(\theta)&=  \sum_{i=1}^n \ell(\theta|\bX_{i,\bR_i},\bR_i)  \\
&= \sum_{i=1}^n\sum_{j \in \bR_i} \log q_\theta(x_j = \bX_{ij}|x_r =  \bX_{i, r}, r=\bR_i\ominus e_j).
\end{aligned}
\label{eq::MOO::likelihood2}
\end{equation}
With equation \eqref{eq::MOO::likelihood2}, we  estimate $\theta$ by the maximum likelihood estimator (MLE)
\begin{equation}
\hat \theta_n = {\sf argmax}_\theta \,\,\ell_n(\theta).
\label{eq::MLE}
\end{equation}
We provide a Gaussian example of this framework in Appendix \ref{sec::sep}.

The MLE is a minimizer of an empirical risk, so we can define its population analog:
\begin{equation}
\theta^*  = {\sf argmax}_\theta  \,\, \bar \ell(\theta), \qquad \bar \ell(\theta) = \E\{\ell(\theta|\bX_{1,\bR_1},\bR_1)\}.
\label{eq::pMLE}
\end{equation}
$\theta^*$ can be estimated by the MLE $\hat \theta_n$ under proper assumptions (Theorem~\ref{thm::MLE}).

The population parameter $\theta^*$ has a useful interpretation. The log-likelihood is a strictly proper scoring rule. Therefore, maximizing the expected MOO log-likelihood $\bar \ell(\theta)$ is equivalent to finding the parameter $\theta$ that minimizes the Kullback-Leibler divergence between the model's marginals $q_\theta(x_j|x_r, r)$ and the true target marginals $p(x_j|x_r, r\oplus e_j)$.This means that if the model is well-specified (i.e., there exists a unique $\theta_0$ such that $q_{\theta_0} \in \mathcal{Q}^*_{MOO}$), then $\theta^* = \theta_0$ (Theorem~\ref{thm::MOO::ID}). If the model is misspecified, $\theta^*$ is the parameter that makes $q_{\theta^*}$ the closest possible approximation to the optimal set $\mathcal{Q}^*_{MOO}$ within the given parametric family. 


\begin{theorem}[Asymptotic normality of MOO-MLE]
Assume the following conditions:
\begin{itemize}
\item[(A1)] The MLE $\theta^*$ in equation \eqref{eq::pMLE} is unique and lies in the interior of a compact parameter space $\Theta$
and satisfies the score equation $ \nabla \bar \ell(\theta^*) = 0$. 
\item[(A2)] The Hessian matrix $\bar H(\theta) = \nabla\nabla \bar \ell(\theta)  = \E[\nabla_\theta\nabla_\theta\ell(\theta|\bX_{1,\bR_1},\bR_1)]$ is invertible at $\theta = \theta^*$.
\item[(A3)] There exists a function $\Lambda(X_R, R)$ such that  $\sup_{\theta\in\Theta}\max_{j_1,j_2,j_3}\left|\frac{\partial}{\partial \theta_{j_1}} \frac{\partial}{\partial \theta_{j_2}} \frac{\partial}{\partial \theta_{j_3}}\ell(\theta|X_R,R)\right|\leq \Lambda(X_R, R)$ and $\E[|\Lambda(X_R, R)|]<\infty$. 
\end{itemize}
Then we have 
$$
\sqrt{n}(\hat \theta_n - \theta^*) \overset{d}{\rightarrow} N(0, \Sigma(\theta^*)),
$$
where $\Sigma(\theta)  = \bar H^{-1}(\theta)\E\left[(\nabla_\theta \ell(\theta|\bX_{1,\bR_1},\bR_1) )(\nabla_\theta \ell(\theta|\bX_{1,\bR_1},\bR_1) )^T\right]\bar H^{-1}(\theta)$.
\label{thm::MLE}
\end{theorem}
Theorem \ref{thm::MLE} shows the asymptotic normality of the MLE when the MLE is a unique maximizer. 
(A1) requires that the MLE is the unique maximizer, 
which is a standard identifiability assumption. This could be violated if the parametric model $q_\theta$ is such that multiple $\theta$ values produce the same optimal marginals.
(A2) is a mild condition that requires the curvature of the MOO log-likelihood around the population MLE
to behave well. It is a standard assumption in MLE theory. 
(A3) requires a third-order derivative to be bounded, which is also a mild condition.
We assume this form to ensure algorithmic convergence and model selection consistency as well (see Theorems \ref{thm::GD} and~\ref{thm::BIC}).
(A3) ensures that the remainder terms in the Tayler expansion
around the MLE are negligible; also, under compact parameter space from (A1), this condition implies that the Hessian matrix
of the MOO log-likelihood is also
uniformly bounded in expectation. 
In Section \ref{sec::sep}, we provide an example where all the conditions are satisfied.

\subsection{Gradient ascent and its algorithmic convergence}

Numerically, we may use gradient ascent to find the MLE $\hat \theta_n$ when no closed-form solution is available. 
This can be done easily by utilizing the score function (gradient of the log-likelihood function)
\begin{equation}
S_n(\theta) = \nabla\ell_n(\theta)   = \sum_{i=1}^n\sum_{j \in \bR_i} \nabla_\theta \log q_\theta(x_j = \bX_{ij}|x_r = \bX_{i, r},  r=\bR_i\ominus e_j),
\label{eq::MOO::likelihood3}
\end{equation}
which is generally easy to compute.
Specifically, we start with an initial guess $\theta^{(0)}$ and iterate the following procedure until convergence:
\begin{equation}
\theta^{(t+1)} = \theta^{(t)}  + \xi \cdot \frac{1}{n}S_n (\theta^{(t)})
\label{eq::GD}
\end{equation}
where $\xi>0$ is an appropriate step size. Note that we divide the gradient by $n$ because the score function in equation \eqref{eq::MOO::likelihood3}
is additive over all observations, which grows at rate $O_P(n)$.

\begin{theorem}[Algorithmic convergence of gradient ascent]
Under assumptions (A1-3) in Theorem~\ref{thm::MLE}, there exists
a radius $\zeta_0>0$ and a stepsize threshold $\xi_0>0$ such that 
if the initial point  $\theta^{(0)}\in B(\hat \theta_n, \zeta_0)$
and the step size $\xi< \xi_0$,
then with a probability tending to $1$, the gradient ascent algorithm in equation \eqref{eq::GD} satisfies
$$
\|\theta^{(t)} - \hat \theta_n\| \leq \rho_\xi^t \|\theta^{(0)} - \hat \theta_n\|
$$
for some $\rho_\xi \in(0,1)$.
\label{thm::GD}
\end{theorem}
Theorem \ref{thm::GD} shows a local linear convergence \citep{boyd2004convex} of the gradient ascent algorithm in equation \eqref{eq::GD}.
The high level idea of the proof is to show that the sample MOO log-likelihood function $\ell_n(\theta)$
is locally concave within $B(\hat \theta_n, \zeta_0)$ with a probability tending to $1$.
Once we have established this result, the algorithmic convergence follows
from the conventional analysis of algorithmic convergence for a (locally) strongly concave function.
$\zeta_0$ and $\xi_0$ can be chosen as
\begin{align*}
\zeta_0 = \frac{-\lambda^*_{\max}}{6\psi_3},\qquad
 \xi_0 =\min\left\{\frac{-3}{\lambda^*_{\max}}, \frac{1}{2H_{\max}}\right\},\qquad \rho_\xi =\sqrt{1+\frac{1}{3}\lambda^*_{\max}\xi},
\end{align*}
where 
$ \lambda^*_{\max} = \lambda_{\max}( \bar H(\theta^*))<0$
is the largest eigenvalue of the Hessian matrix
$\bar H(\theta)  = \nabla\nabla \bar \ell(\theta)$ at $\theta =\theta^*$,
$ \psi_3  =  \sup_{\theta\in\Theta} \max_{j_1,j_2,j_3}\left|\frac{\partial}{\partial \theta_{j_1}}\frac{\partial}{\partial \theta_{j_2}}\frac{\partial}{\partial \theta_{j_3}}\bar \ell(\theta)\right|$
is the maximal third-order derivative,
$ H_{\max} = \sup_{\theta\in\Theta} \|\bar H(\theta)\|_2$ is the maximal $2$-norm of the Hessian matrix.
Assumption (A3) guarantees that $\psi_3, H_{\max}<\infty$.
It is also possible to obtain the speed on how fast the probability tends to $1$
since we have an explicit characterization on the events that are needed for the linear convergence in Theorem~\ref{thm::GD};
see the proof in Section \ref{sec::thm::GD} for more details.

\subsection{Optimal MOO imputation models and log-likelihood}

The MOO log-likelihood is applicable for any imputation model $q$ admitting a PDF or PMF via
$$
\ell_n(q)=  \sum_{i=1}^n \ell(q|\bX_{i,\bR_i},\bR_i)  = \sum_{i=1}^n\sum_{j \in \bR_i} \log q(x_j = \bX_{ij}|x_{r} = \bX_{i, r}, r = \bR_i\ominus e_j).
$$
When evaluating $q$ is costly but sampling is tractable,  
the MOO log-likelihood  can be approximated via a Monte Carlo approach; see Appendix \ref{sec::MCMOO} for more details.

With this definition, the MOO log-likelihood has an interesting identification property. 
Let
\begin{equation}
\bar \ell(q) = \E\{\ell(q|\bX_{1,R_1},\bR_1)\} = \E\left\{\sum_{j \in \bR_1} \log q(x_j = \bX_{1j}|x_{r} = \bX_{1, r}, r = \bR_1\ominus e_j)\right\}
\label{eq::MOO::likelihood3}
\end{equation}
be the population MOO log-likelihood 
for any imputation model $q$.
\begin{theorem}
Any imputation model $q\in  \mathcal{Q}^*_{MOO}$ maximizes
the MOO log-likelihood, i.e., 
$\inf_{q\in \mathcal{Q}^*_{MOO}} \bar \ell(q) = \sup_{q} \bar\ell(q).$
\label{thm::MOO::ID}
\end{theorem}

Theorem~\ref{thm::MOO::ID} offers another view on the optimal imputation model $\mathcal{Q}^*_{MOO}$--any models
inside $\mathcal{Q}^*_{MOO}$ will maximize the MOO log-likelihood.
Thus, these models are optimal  from the perspective of likelihood principle. 



\subsection{Identification under missing completely at random}


In this section, we  study the behavior of MOO likelihood under missing completely at random (MCAR). 
Suppose we have a parametric model for the marginal distribution of $X$ only, i.e., $p(x) =f_\theta(x)$,
where $\theta$ is the underlying parameter. 
The MCAR requires 
$P(R=r|X=x) = P(R=r)$, i.e., $R\perp X$.
Clearly, the imputation model under MCAR is
$$
p(x_{\bar r}|x_r, r)= p(x_{\bar r}|x_r) = f_\theta(x_{\bar r}|x_r) = \frac{f_\theta(x)}{f_{\theta}(x_r)},
$$
which is the implied conditional model under the joint model $f_\theta(x)$.


If such parametric model is correct, i.e., $p(x) = f_\theta(x)$, and true missing mechanism is MCAR, then this model also maximizes the MOO log-likelihood.
\begin{theorem}[Recovery under MCAR]
\label{thm::mcar}
Suppose the true joint distribution that generates our data is  $p(x) = f_{\theta^*}(x)$ for some unknown parameter $\theta^*$
and the missingness is MCAR. 
Then we have the following result:
$$
\bar \ell(f_{\theta^*}) = \sup_q \bar \ell(q).
$$
\end{theorem}

Theorem~\ref{thm::mcar} shows that the correct parametric model under MCAR indeed 
maximizes the MOO log-likelihood. 
Since many modern imputation models
are trained under the assumption of MCAR \citep{yoon2018gain},
the MOO criterion offers an alternative objective 
in the training process.
Informally, Theorem~\ref{thm::mcar} also implies that
when all data are complete and we are just using masking to 
train the full model (this occurs in training a large language model or image model; \citealt{devlin2019bert,vincent2010stacked}), 
the true generative model maximizes the MOO log-likelihood.
Thus, maximizing the masked log-likelihood can be a method for learning the data-generating model.

\begin{remark}[MOO likelihood and MAR]	\label{rm::mar}
If we assume the joint model to be $f_\theta(x)$ and missing mechanism is missing-at-random (MAR), i.e., $P(R=r|X=x)= P(R=r|X_r = x_r) $, 
the imputation model will be 
$$
p(x_{\bar r}|x_r, r)= p(x_{\bar r}|x_r) = f_\theta(x_{\bar r}|x_r),
$$
which is similar to MCAR.
Suppose the data are from $f_{\theta^*}(x)$ and MAR is correct,
one may be wondering if the imputation model $f_{\theta^*}$ maximizes  the MOO log-likelihood?
Unfortunately, the answer is no unless the missingness is monotone (see Section \ref{sec::mono} and Equation \eqref{eq:MOOBL2}).
The major problem is that for a pattern $R=r$ and we attempt to impute $x_j$ where $j \in \bar r$,
the optimal imputation model under MOO is 
$$
q(x_j|x_r, R=r) = p(x_j|x_r, R=r \oplus e_j).
$$
Under MAR, the optimal imputation model is $p(x_j|x_r) = f_{\theta^*}(x_j|x_r)$.
Thus, the MOO optimal model
\begin{align*}
p(x_j|x_r, R=r \oplus e_j) & = \frac{p(x_j,x_r, R=r \oplus e_j)}{p(x_r, R=r \oplus e_j)}\\
& = \frac{P(R=r\oplus e_j|x_r, x_j) f_{\theta^*}(x_r,x_j)}{\int  P(R=r\oplus e_j|x_r, x_j) f_{\theta^*}(x_r,x_j)dx_j}.
\end{align*}
This quantity will be the same as $f_{\theta^*}(x_j|x_r)$ only if 
$$
f_{\theta^*}(x_r) P(R=r\oplus e_j|x_r, x_j)  = \int  P(R=r\oplus e_j|x_r, x_j) f_{\theta^*}(x_r,x_j)dx_j,
$$
which is generally not the case because the left-hand-side depends on $x_j$ while the right-hand-side does not. 
\end{remark}



\subsection{Selecting imputation models with MOO likelihood}

The likelihood function in equation \eqref{eq::MOO::likelihood2} can be used as a selection criterion 
for different imputation models
as long as we can evaluate $q$ easily. 
Specifically, suppose we have $q_1,\cdots, q_K$ and we want to select an imputation model. 
We compute their MOO log-likelihoods as in equation \eqref{eq::MOO::likelihood2}:
$$
\ell_n(q_k) = \sum_{i=1}^n\sum_{j \in R_i} \log q_k(x_j = \bX_{ij}|x_{r} = \bX_{i, r}, r = \bR_i\ominus e_j)
$$
and choose the model that has the highest log-likelihood. 
Note that the above quantity is a sample analogue of equation \eqref{eq::MOO::likelihood3}.
However, this  suffers from overfitting problem because 
a complex model tends to have a higher likelihood, 
so we should not directly use $\ell_n(q_k)$ for comparing different models.
We need to add a penalization/regularization such as the AIC \citep{akaike1974aic} or BIC \citep{schwarz1978bic} to $\ell_n(q_k)$ for model selection. 

For the case of BIC, the MOO criterion is 
\begin{equation}
\ell_{n, BIC}(q_k) = \ell_n(q_k) - \frac{1}{2} d(q_k) \log n,
\label{eq::BIC}
\end{equation}
where $d(q_k) $ is the number of parameters (dimension of free parameters) of the imputation model $q_k$.
We choose the model $q_{\hat k}$ via $\hat k = {\sf argmax}_k\,\, \ell_{n, BIC}(q_k)$.

When the models being compared are nested and the true model belongs to one of them, the BIC 
can select the correct model asymptotically. 
\begin{theorem}[Model selection consistency]
Suppose we have $K$ nested models $\mathcal{Q}_1\subset \mathcal{Q}_2\subset \cdots\subset \mathcal{Q}_K$
such that each model $\mathcal{Q}_k = \{q_{\theta_{[k]}}: \theta_{[k]} \in \Theta_{[k]}\subset \R^{d_k}\}$ is indexed by  $\theta_{[k]}$
with $d_k$ free parameters and $d_1<d_2<\cdots<d_K$. 
Assume that 
\begin{itemize}
\item[(AS)]  conditions (A1-3) hold for every model $\mathcal{Q}_k$ and 
\item[] there exists $k^*$ such that
\begin{itemize}
\item[(B1)] $\mathcal{Q}_k\cap \mathcal{Q}^*_{MOO} =\emptyset$ for all $k=1,\cdots, k^*-1$.
\item[(B2)] For model $\mathcal{Q}_{k^*}$, there exists a parameter $\theta^*_{[k^*]}$ such that $q_{\theta^*_{[k^*]}}\in \mathcal{Q}^*_{MOO}$. 
\end{itemize}
\end{itemize}
Let $q_k \in \mathcal{Q}_k$ be the model corresponding to the MLE under model $\mathcal{Q}_k$,
i.e., 
$$
q_k  = q_{\hat \theta_{[k]}},\qquad \hat \theta_{[k]} = {\sf argmax}_{\theta_{[k]} \in \Theta_{[k]}} \ell_n(q_{\theta_{[k]}}),
$$
and
$\ell_{n, BIC}(q_k)$ be the BIC  in equation \eqref{eq::BIC}
and $\hat k = {\sf argmax}_k\,\, \ell_{n, BIC}(q_k)$ be the model selected by the BIC. 
Then 
$
P\left(\hat k = k^*\right) {\rightarrow}1
$
as $n\rightarrow \infty$.
\label{thm::BIC}
\end{theorem}

Theorem~\ref{thm::BIC} shows that the BIC  has model selection consistency for nested models. 
A technical challenge of this proof is that since MOO likelihood is not the conventional likelihood function,
we cannot apply the Wilk's theorem \citep{wilks1938large}. So we need some extra conditions to ensure model selection consistency. 
Condition (AS) is needed so that the MLE of each model $q_k$ is well-behaved.
This condition implies two useful results. First, we have a uniform bound 
$\sup_{\theta_{[k]}\in\Theta_{[k]}}\frac{1}{n}\left|\ell_n(q_{\theta_{[k]}}) - \E[\ell_n(q_{\theta_{[k]}})]\right|\overset{P}{\rightarrow}0$,
which will be useful in controlling the errors when $k<k^*$ .
Secondly, the asymptotic normality of each MLE $q_k$ will lead to a finite-order
stochastic fluctuations for $k>k^*$, which eventually leads to a control over the empirical MOO log-likelihood value. 
Condition (B1) means that the optimal imputation model under MOO likelihood does not 
appear before model $\mathcal{Q}_{k^*}$. Condition (B2) states that $\mathcal{Q}_{k^*}$ is the minimal model 
that contains the optimal imputation model. Thus, $\mathcal{Q}_{k^*}$ can be interpreted as the `correct' model
for imputation under MOO likelihood.
Note that (B2) can be relaxed so that we do not need $q_{\theta^*_{[k^*]}}\in \mathcal{Q}^*_{MOO}$
but instead, we require $q_{\theta^*_{[k^*]}}$ is the closest model to $\mathcal{Q}^*_{MOO}$
under KL divergence
and for larger models, $q_{\theta^*_{[k]}}$ with $k>k^*$,
this gap in the KL divergence is not improved. 
The nested model assumptions (B1-2) are common assumptions for BIC to be consistent \citep{nishii1984asymptotic, yang2005can}.

\section{Masking in monotone missing data}	\label{sec::mono}

The monotone missing data is a special scenario where the variables are ordered and missingness occurs
in a way that if one variable is missing, all subsequent variables are missing \citep{little2019statistical}.
This occurs  frequently in health-related data due to participants dropout of the study. 

In the monotone missing data scenario, the response vector $R\in\{0,1\}^d$ can be summarized 
by $T = \sum_{j} R_j$, the total number of observed variables,
because $R_j = 0$ will imply $R_k = 0$ for all $k> j$. 
In this case, the observed variable
$X_R = (X_j : j \in R) = X_{\leq T} = (X_j: j=1,\cdots, T)$.
The observed data distribution is $p(x_{\leq t}, t)$
and the extrapolation distribution is $p(x_{>t}|x_{\leq t}, t).$

Because of the monotone missing pattern,
we can no longer mask any variable arbitrarily 
otherwise we will end up with patterns that do not exist in monotone missing data. 
For MOO, there are two  ways of performing masking while
maintaining the monotone missingness.

{\bf MOO with latest case (MOOLC).}
Given an observation $(x_{\leq t},t)$, the MOO-LC
just mask the latest variable $x_t$. 
Namely, after masking, the data becomes $(x_{\leq t-1}, t-1)$. 
So the MOO procedure will only mask one variable for every individual
and attempt to impute the masked variable. 

{\bf MOO with blocking (MOOBL).}
Given an observation $(x_{\leq t},t)$, the MOOBL
allows masking any observed variable $j =1,\cdots, t$.
If variable $x_j$ is masked,
we block out all subsequent variables (pretending subsequent variables are missing)
and impute $x_j$ via $q(x_j|x_{\leq j-1}, T= j-1).$
Note that we only impute $x_j$; the other blocked variables $x_{j+1},\cdots, x_{t}$ will not be imputed.

\begin{example}
Suppose we have an observation $X = (1.3, 2.5, 1.5, 3.1, \texttt{NA})$, which implies $T=4$ ($R=11110$).
The MOOLC has only one possible masking scenario: masking  $X_4= 3.1$
and attempt to impute $X_4$ with $p(x_4|X_1=1.3, X_2=2.5, X_3=1.5, T=3).$
In the case of MOOBL, we have 4 possible ways of masking: masking $X_1,X_2,X_3,X_4$, separately.
Suppose we mask $X_3=1.5$, we will  block $X_4=3.1$ at the same time
and impute $X_3$ with $p(x_3|X_1=1.3, X_2=2.5, T=2)$ and compare the imputed value to
the observed value. Note that we do not impute $X_4$--we just ignore (block) it. 
MOOBL will apply this masking to all four observed variables and
use summation over all losses.
\end{example}

%

The MOOLC and MOOBL lead to different optimal imputation models. 
Their optimal models in the form of equation \eqref{eq::MOO::opt::set} are
\begin{equation}
\mathcal{Q}^*_{MOOLC} =\{q: q(x_{t+1}|x_{\leq t}, t) = p(x_{t+1}|x_{\leq t},T=t+1), \quad t=0,\cdots, d-1\}
\label{eq::MOOLC::opt::set}
\end{equation}
and 
\begin{equation}
\mathcal{Q}^*_{MOOBL} =\{q: q(x_{t+1}|x_{\leq t}, t) = p(x_{t+1}|x_{\leq t},T\geq t+1), \quad t=0,\cdots, d-1\},
\label{eq::MOOBL::opt::set}
\end{equation}
respectively.
Note that the MOO procedure does not constrain any imputation on $x_{t+2},\cdots, x_d$
when the the last observed variable is at $T=t$.

It is clear why $\mathcal{Q}^*_{MOOLC}$ is the optimal imputation model
because the only constraint from MOOLC on the imputation model $q(x_{t+1}|x_{\leq t},T=t)$ 
occurs when we observe $x_1,\cdots, x_{t+1}$ and $T=t+1$. The masking of this case will lead to 
an imputation of $x_{t+1}$ using  $x_1,\cdots, x_t$  and $T=t$.
For the case of $\mathcal{Q}^*_{MOOBL}$, the imputation model $q(x_{t+1}|x_{\leq t}, t)$
will be used in all of the following cases:
we observed $x_1,\cdots, x_s$ ($T=s$) with $s\geq t+1$. 
When we mask  $x_{t+1}$, the MOOBL procedure will require imputing $x_{t+1}$ with $x_1,\cdots, x_t$ and $T=t$,
which is the imputation model $q(x_{t+1}|x_{\leq t}, t)$.


Here is an interesting connection to the classical nearest-case missing value assumption (NCMV)
and the available-case missing value assumption (ACMV; \citealt{molenberghs1998monotone}).
The NCMV and ACMV corresponds to the imputation models such that 
for every $t$ and $\tau\geq t$, 
\begin{equation}
\begin{aligned}
q_{NCMV}(x_{\tau+1}|x_{\leq \tau}, t)  &=p(x_{\tau+1}|x_{\leq \tau}, T= \tau+1),\\
q_{ACMV}(x_{\tau+1}|x_{\leq \tau}, t)  &=p(x_{\tau+1}|x_{\leq \tau}, T\geq \tau+1).
\end{aligned}
\label{eq::acmv}
\end{equation}
The above two imputation models can be viewed as a sequential imputation procedure that
when we observed only up to time point $T=t$ (i.e., $x_{\leq t}$ is observed), we sequentially impute $x_{t+1},x_{t+2},\cdots, x_d$
from the conditional distribution $q(x_{\tau+1}|x_{\leq \tau}, T=t)$ for $\tau=t,t+1,\cdots, d-1$.
At each $\tau$,
the NCMV requires that such imputation model only uses individuals where only the set of variables $(x_1,\cdots, x_{\tau+1})$ is observed and
the subsequent variables are missing (i.e., $T=\tau+1$). 
The ACMV, on the other hand, uses any individual as long as the vector $(x_1,\cdots, x_{\tau+1})$
is observed.

\begin{proposition}
The imputation model implied by NCMV belongs to $\mathcal{Q}^*_{MOOLC}$, i.e., $q_{NCMV} \in \mathcal{Q}^*_{MOOLC}$.
The imputation model implied by ACMV belongs to $\mathcal{Q}^*_{MOOBL}$, i.e., $q_{ACMV}\in \mathcal{Q}^*_{MOOBL}$.
\label{prop::acmv}
\end{proposition}

Proposition \ref{prop::acmv} implies a very interesting result about the missing-at-random (MAR) and MOO
log-likelihood. In the monotone missing data problem,
the imputation model under MAR maximizes the log-likelihood function under MOOBL (in contrast to Remark \ref{rm::mar}).
To see this,
the population MOO log-likelihood in equation \eqref{eq::MOO::likelihood3} 
under monotone missing data is
\begin{equation*}
\bar \ell_{MOOBL}(q)=  \E\left\{\sum_{j \leq \bT_1} \log q(x_j = \bX_{1j}|x_{<j} = \bX_{1, <j}, t = j-1)\right\},
\end{equation*}
where $(\bX_1,\bT_1)$ are identically distributed as $(X,T)$.
By Theorem~\ref{thm::MOO::ID}, any imputation model $q\in \mathcal{Q}^*_{MOOBL}$
maximizes $\bar \ell_{MOOBL}(q)$.
Using the fact that under monotone missing data, MAR and ACMV are equivalent \citep{molenberghs1998monotone},
Proposition \ref{prop::acmv} implies that the imputation model based on MAR
will maximize $\bar \ell_{MOOBL}(q)$. Namely, we can revise Theorem \ref{thm::mcar} to:  
under monotone missing data,
if the true joint distribution that generates our data is $f_{\theta^*}(x) = p(x)$ for some $\theta^*$ and missing mechanism is MAR,
then 
\begin{equation}
\bar \ell_{MOOBL}(f_{\theta^*}) = \sup_q\bar \ell_{MOOBL}(q).
\label{eq:MOOBL2}
\end{equation}
Thus, in the monotone missing data case,
we may use the sample version of $\bar \ell_{MOOBL}(q)$  as an objective function 
to learn the optimal imputation model under MAR. 

\section{Empirical studies}\label{sec::empirical}


\subsection{Simulations}
We evaluate the proposed masking criteria using five real datasets from the UCI Machine Learning Repository, with varying dimensions summarized in Table \ref{tab::sim_data}. Variables that are numeric are kept and standardized to ensure the loss values are of the same order. Missing values are introduced under the missing completely at random (MCAR) mechanism, with each data entry independently set to missing with probability 0.3.
We adopt a cross-fitting strategy (see also \citealp{zhao2025imputationpoweredinference}) to separate model training from evaluation and thereby prevent overfitting. The data are partitioned into $K=5$ folds. For each fold $k$, we train the imputation models on the remaining $K-1$ folds and apply the trained imputers to that fold for evaluation using masking criteria. 

We consider a collection of deterministic and stochastic imputation methods, including mean imputation, the Expectation-Maximization (EM) algorithm, nearest-neighbor hot deck (NN HD; \citealt{little2019statistical}), complete-case missing value (CCMV; \citealt{tchetgen2018discrete}), Markov missing graph (MMG; \citealp{yang2025markovmissinggraphgraphical}), and multiple imputation by chained equations (MICE; \citealt{van2011mice}). The EM, CCMV, and MMG methods fit Gaussian models, where MMG uses the Gaussian-MMG specification. Because MICE lacks an explicit OOS property, we append each test observation to the training data and rerun MICE to impute the masked entries using the ``mice'' function in R.
Finally, we compute the naive MOO risk using the squared loss function, the MOORT risk using the Kolmogorov distance, and the MOOEN risk. Each simulation is repeated $n_{\mathrm{sim}} = 100$ times. We set the number of neighbors to $k = 10$ for the nearest neighbor hot deck and the number of multiple imputations to $M=20$ for all three criteria (MOO, MOORT, MOOEN). 
We do not use the MOO log-likelihood since the log-likelihood function for some imputation models 
such as MICE and NN HD
are either not well-defined or would require a non-trivial amount of Monte Carlo approximations.

\begin{table}
\centering
\caption{Datasets with number of observations $n$ and number of variables $p$ in the simulation.}
\begin{tabular}{|l|c|c|}
\hline
\text{data set} & \text{n} & \text{p} \\ \hline
Iris & 150 & 4 \\ \hline
Seeds & 210 & 6 \\ \hline
Yacht & 308 & 7 \\ \hline
Concrete Compression & 1030 & 9 \\ \hline
Red Wine Quality& 1599 & 11 \\ \hline
\end{tabular}
\label{tab::sim_data}
\end{table}

To illustrate the relative performance of imputation methods under different masking criteria, we introduce the \emph{Prediction-Imputation (PI) diagram} (see Figure \ref{fig:PI_sim}). In this diagram, the x-axis is the MOO risk, which measures the prediction error of each imputation method, and the y-axis is the imputation risk under MOORT or MOOEN, which measures the departure of the imputation model from $\mathcal{Q}^*_{MOO}$. Methods near the lower-left region of the diagram achieve better performance in terms of prediction and imputation. 
Across datasets, MICE and EM mostly occupy this favorable region, which suggests good prediction and imputation performance. In contrast, mean imputation generally performs the worst (as expected), while CCMV, MMG, and hot deck method show intermediate performance that vary by dataset. Beyond comparing methods, the PI diagrams also reveal a clear distinction between prediction and imputation. High predictive accuracy does not necessarily correspond to distributional closeness to the optimal imputation model. For instance, mean imputation achieves a relatively small MOO prediction risk on datasets such as Concrete and Red Wine because the imputed values are close to true conditional mean, yet it exhibits the largest imputation risk due to its inability to account for the uncertainty. This also reflects the aforementioned limitation of the MOO criterion, which may favor deterministic imputation method.
In Appendix \ref{sec::appendix_sim}, we provide additional analysis on these data by comparing the MOO criteria
to the `oracle' imputation performance.

\begin{figure}
    \centering
    \includegraphics[height=3in]{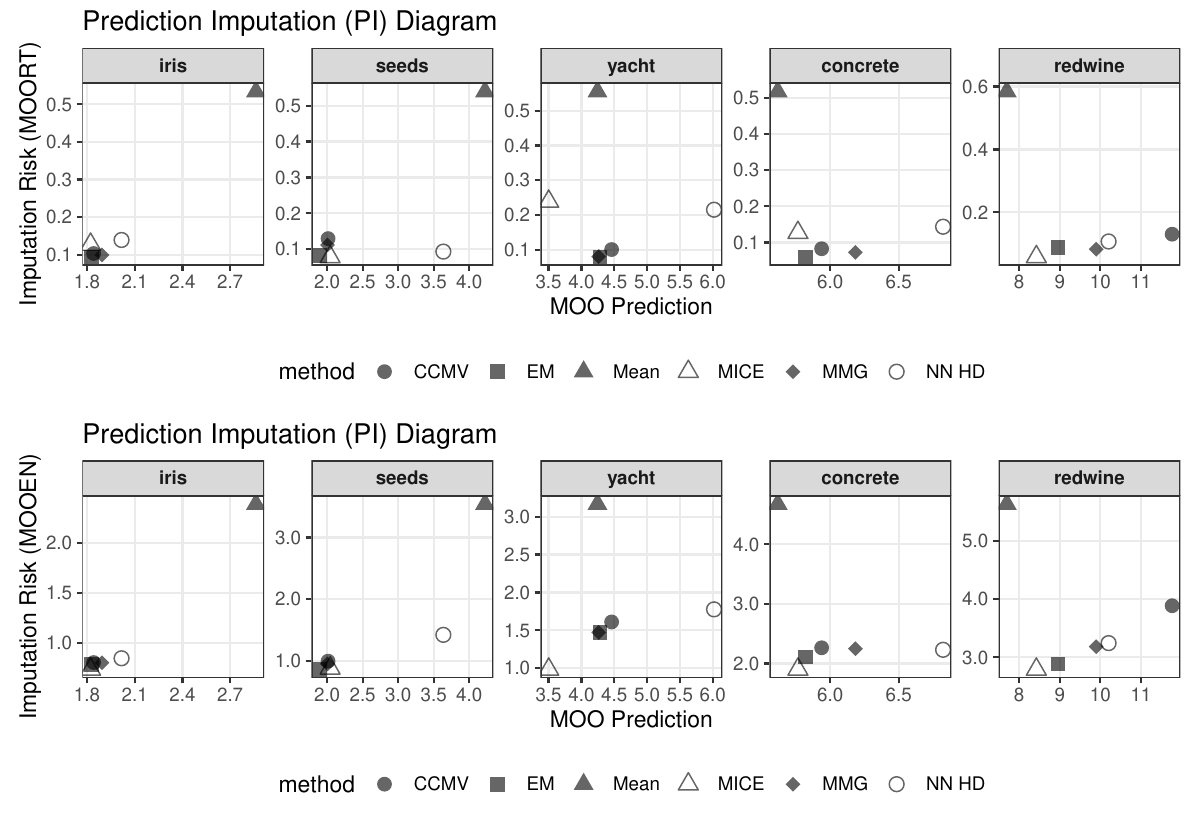}
    \caption{Prediction-Imputation (PI) Diagram comparing imputation methods (CCMV, EM, mean imputation, MICE, MMG, and nearest-neighbor hot deck) under the MOO, MOORT, and MOOEN criteria across simulation datasets. Methods closer to the lower-left region indicate lower risks and better performance.}
    \label{fig:PI_sim}
\end{figure}




\subsection{Real data: NACC data}

\begin{figure}
    \centering
    \includegraphics[height=2in]{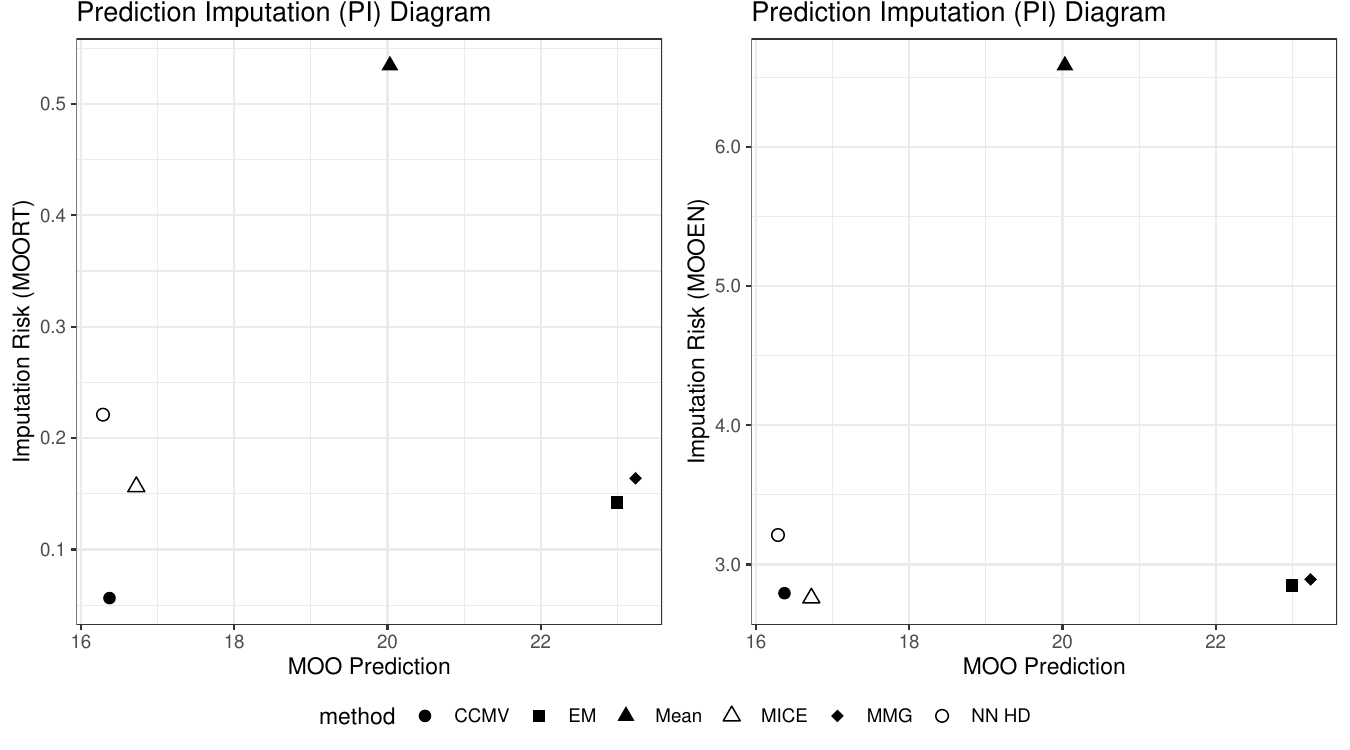}
    \caption{Prediction-Imputation (PI) Diagram comparing imputation methods (CCMV, EM, mean imputation, MICE, MMG, and nearest-neighbor hot deck) under the MOO, MOORT, and MOOEN criteria on the NACC dataset for the DIGFORCT variable. }
    \label{fig:PI_nacc}
\end{figure}

We further illustrate the masking criteria using the National Alzheimer's Coordinating Center (NACC) dataset\footnote{\url{https://naccdata.org/}}  from the years 2005 to 2024. The NACC data include longitudinal cognitive test scores collected across multiple visits. To focus on a well-defined neuropyschological outcome, we analyze the Number Span Test-Forward (DIGFORCT), which is a numeric measure of memory and attention ranging from 0 to 25. We restrict the analysis to the first five visits, resulting in 2777 participants with approximately 30\% missing entries. To avoid overfitting, we use the cross-fitting with $K=3$ folds. We use the same number of multiple imputation $M=20$ for the three MOO methods.  Because the test score is discrete, we use the EM algorithm with a mixture of binomial product experts model \citep{suen2023modelingmissingrandomneuropsychological} and the MMG method with the mixture-of-product MMG specification. The graph structure for MMG is estimated using graphical lasso applied to the complete cases. Figure~\ref{fig:PI_nacc} shows that CCMV has the lowest imputation risks under both MOORT and MOOEN while maintaining a small prediction risk follow by MICE and hot deck. MMG and EM perform comparably well in terms of imputation but show a higher prediction risk, whereas mean imputation  exhibits the largest imputation risk. 

\begin{figure}
    \centering
    \includegraphics[height=3in]{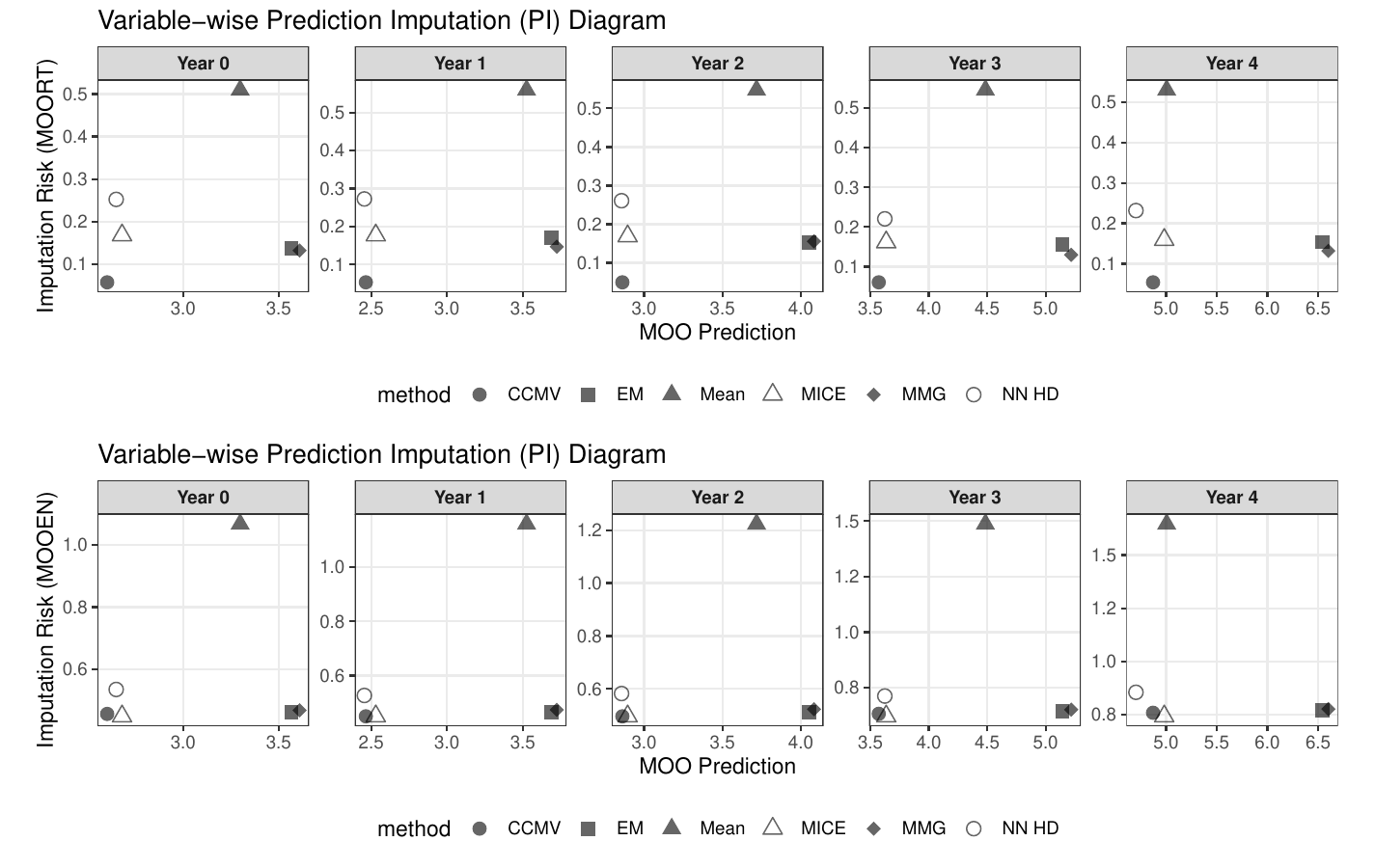}
    \caption{Prediction-Imputation (PI) Diagram under variable-wise MOO, MOORT, and MOOEN criteria on the NACC dataset (DIGFORCT).
    Each variable is the DIGFORCT score at different visits.}
    \label{fig:PI_nacc_varwise}
\end{figure}
The variable-wise masking procedure (Appendix \ref{sec::variablewise}) computes the risks separately for each variable and provides a detailed comparison of imputation methods. Figure~\ref{fig:PI_nacc_varwise} displays the PI diagrams for each visit under variable-wise MOO, MOORT, and MOOEN procedures. Compared to Figure~\ref{fig:PI_nacc}, 
the patterns are similar.
CCMV generally achieves the lowest risks, whereas mean imputation has the largest imputation risk. We note that MOO risks in the variable-wise PI diagrams are not directly comparable across variables, as each variable can have a different missingness rate and thus contribute differently to the risk computation.

\section{Discussion}

In this paper, we analyze four masking criteria:  naive MOO, MOORT, MOOEN, and MOO log-likelihood. 
All these criteria are computable for observations with missing entries. 
We rigorously study theoretical properties of these criteria
and investigate the  corresponding optimal imputation model.
For comparing imputation models in practice, 
we will recommend MOORT and MOOEN 
because they are distributional measures and straightforward to implement. 
The naive MOO is more like an assessment of prediction,
which is not a suitable criterion for comparing imputation models \citep{van2018flexible,naf2023imputation, grzesiak2025needdozensmethodsreal}. 
While the MOO log-likelihood has elegant theoretical properties and is useful in learning an imputation model,
it is not ideal for comparing different models because an imputation model may not
have a well-defined likelihood function or could be costly to evaluate. 
This is a particularly severe problem for modern generative models.



In what follows, we discuss some possible future work.

\begin{itemize}

\item {\bf Evaluating the joint dependency of imputed variables.} The proposed criteria (MOORT, MOOEN, and MOO likelihood) are designed to recover the optimal marginal imputation distributions. They do not, however, constrain the joint dependency structure among the variables being imputed. While we generalized this to Mask-K-Out (MKO) in Appendix \ref{sec::MKO}, the use of an additive loss means it also fails to evaluate this joint structure. Furthermore, as we showed in Appendix \ref{sec::max::mao}, defining a consistent optimal joint distribution as a target is a non-trivial task. Therefore, an open problem is the design of a principled masking criterion that can successfully evaluate and constrain the joint dependencies of an imputation model.


\item {\bf Constructing imputation risk via masking.}
The masking approach offers an elegant way to turn the problem of learning an imputation
model into a risk minimization procedure under different missing data assumptions.
In Section \ref{sec::learning}, we have demonstrated that
masking can be used to construct an imputation risk via a likelihood approach that acts like an empirical risk,
leading to an objective for learning an imputation model. 
There are two key ingredients for this procedure. 
The first is how we mask observed entries.
A masking procedure (masking only one entry versus multiple entries)
corresponds to a missing data assumption.
So different masking procedure refers to different missing data assumptions.
The second ingredient is how we construct the loss function. 
The log-likelihood in Section \ref{sec::learning} is just an example of a loss function; other options
such as Energy distance may be applicable.

\item {\bf Deep generative models. }
The MOO log-likelihood method enables us to use deep generative models
to learn an imputation method as an alternative approach to the GAIN approach \citep{yoon2018gain}.
For learning via the MOO log-likelihood, we need a model that is easy to sample from and the log-likelihood is tractable. 
Normalizing flow \citep{papamakarios2021normalizing} is an excellent model for this task--it is easy to sample from and the evaluation
of log-likelihood function can be done efficiently.
Variational autoencoder \citep{pu2016variational} is another good alternative if we use a variational approximation
to the MOO log-likelihood function
for training the model. 
It will be of great interest to use these deep generative model
in constructing an imputation model under MOO framework.

\item {\bf Complex masking and missing procedure.}
When the variables are associated in certain ways such as text or image or functional data, 
the missing data patterns often occur in a structured way. 
Not all possible response patterns may occur in the data.
In this case, the masking procedure has to be modified according to the
missingness structure. 
The monotone missing data problem in Section \ref{sec::mono} presents
an example of MOO under monotone missing data, which has two possible variants. 
Thus, how to properly modify the masking procedure
and analyze the underlying optimal imputation distribution remains
an open problem.



\end{itemize}

\begin{acks}[Acknowledgments]
The authors would like to thank Yikun Zhang for helpful comments on the paper.
\end{acks}
%
\begin{funding}
YY and DS were partially supported by NSF Grant DMS-2141808.
YC was supported by NSF Grant DMS-2141808, 2310578 and NIH Grant U24 AG072122. 

\end{funding}
%
%

\bibliography{MOO_new}

\bibliographystyle{abbrvnat}


%

\begin{supplement}
\stitle{Supplementary materials}
\sdescription{This appendix contains the theoretical proofs and additional contents related to the main paper. }
\end{supplement}

\begin{appendix}

{\bf Contents of Appendix.}
\begin{itemize}
\item {\bf Section \ref{sec::MKO}: Masking multiple variables.}
We discuss the procedure of masking multiple variables. 

\item {\bf Section \ref{sec::variablewise}: Variable-wise MOO.}
We include the procedure for performing a variable-wise MOO (and MOORT).

\item {\bf Section \ref{sec::sep}: Learning a separable imputation model using MOO likelihood.}
We analyze the case of learning a separable imputation model via MOO likelihood and
provide a Gaussian example.

\item {\bf Section \ref{sec::MCMOO}: Monte Carlo approximation of MOO likelihood.}
We include a Monte Carlo method for computing the MOO likelihood function 
for general imputation model. .

\item {\bf Section \ref{sec::proofs}: Proofs.}
This section includes all the details of technical proofs.

\item {\bf Section \ref{sec::appendix_sim}: Additional details for the simulation studies.}
We include additional details on the simulation studies. 

\end{itemize}

\section{Masking multiple variables}	\label{sec::MKO}

The idea of masking is not limited to a single variable.
We can perform masking of multiple variables at the same time.
For a response pattern $R\in\{0,1\}^d$, let $J_K(R) = \{r\in\{0,1\}^d: r\leq  R, r\neq 0,  |r|\leq K\}$,
where $|r| = \sum_{j}r_j$ is the summation of elements in the binary vector
and for binary vectors $r,s\in\{0,1\}^d$, we write $r< s$ if $r_j\leq s_j$ for all $j$ and there exist at least strict inequality ($r\leq s$ means
$r_j\leq s_j$ for all $j$).
The element in $J_K(R)$ can be interpreted as the set of all possible variables in $R$
when we can mask at most $K$ variables at a time.
For instance, suppose $R=00111$,
then the set $J_2(R) = \{00110,  00101, 00011, 00100, 00010,00001\}$
and $J_3(R)  = J_2(R) \cup \{00111\}$.
The element $00011$ in $J_2(R)$ is the case where we are masking $x_4$ and $x_5$.
In the case of MOO, we are masking with respect to $J_1(R)$, which in the above example is $J_1(R) = \{00100, 00010, 00001\}$. 


With the notation of $J_K$, we  formally define the \emph{mask-K-out (MKO)} procedure
in Algorithm \ref{alg::mko}.

\begin{algorithm}
\caption{Mask-K-out (MKO) procedure}
\label{alg::mko}
\begin{algorithmic}
\State  \textbf{Input:} Imputation model $q$ and an integer $K>0$. 
\begin{enumerate}
\item For $i=1,\cdots, n$, we do the following:
	\begin{enumerate}
	\item For each $s\in J_K(\bR_i) = \{r\in\{0,1\}^d: r\leq  \bR_i, r\neq 0,  |r|\leq K\}$: 
	\begin{enumerate}
	\item We mask the observed entries $\bX_{i, s}$ and updated the response pattern to be $\bR_i \ominus s$.  
	\item We impute $ \bX_{i,s}$ by sampling from $q(x_{s}|X_{\bR_i \ominus s} = \bX_{i, \bR_i \ominus s}, R = \bR_i \ominus s)$.
	Namely, we treat the data as if $\bX_{i, s}$
	is a missing value and attempt to impute it.
	\item Compute the  loss value $L(\bX_{i, s}, \hat \bX_{i, s}) = \sum_{j \in s}L(\bX_{ij}, \hat \bX_{ij})$
	\end{enumerate} 
	\item Compute the total loss for this individual: $L_K(q| \bX_{i, \bR_i},\bR_i) = \sum_{s\in J_K(R)} L(\bX_{i, s}, \hat \bX_{i, s})$. 
	\end{enumerate}
\item Compute the total risk of the imputation model $q$ as 
$$
\hat{\mathcal{E}}_{K, n }(q ) = \frac{1}{n}\sum_{i=1}^n L_K(q| \bX_{i,\bR_i},\bR_i).
$$
\item (Optional) Repeat the above multiple times and take the average value of $\hat{\mathcal{E}}_{K, n }(q )$ to reduce Monte Carlo error. 
\end{enumerate}

\end{algorithmic}
\end{algorithm}


Note that the loss value $L(\bX_{i, s}, \hat \bX_{i, s}) = \sum_{j \in s}L(\bX_{ij}, \hat \bX_{ij})$
makes it easy to compute and such loss function avoid the need of 
specifying different loss function when input is a vector of different lengths.
One may use a loss function that takes into account dependency among different variables
when imputing multiple variables. However, this will require specifying different losses
when we mask different numbers of variables.

The minimization problem in Algorithm \ref{alg::mko} can be viewed as an empirical risk minimization
and the corresponding population/test risk of $\hat{\mathcal{E}}_{K, n }(q )$ is 
$$
 {\mathcal{E}}_{K }(q ) = \E\left(L_K(q| \bX_{1,\bR_1},\bR_1)\right).
$$
Understanding the minimizer of $ {\mathcal{E}}_{K }(q )$ provides us useful information on the MKO procedure.

To investigate the optimal imputation value under MKO,
we define 
$$
\mathbb{U}_K(r, j)  = \{s\in\{0,1\}^d: s \geq r\oplus e_j, |s-r|\leq K\}
$$
for each $j \in \bar r$,
where for a binary vector $s\in\{0,1\}$, $|s| = \sum_j s_j$.
The set $\mathbb{U}_K(r, j)$ is the collection of response patterns where variable
$X_j$ and every observed variables in $r$ are observed
and there are at most $K$ more variables being observed relative to $r$.
In the special case $K=1$, $\mathbb{U}_1(r, j) = r\oplus e_j$.
For $K>1$, $\mathbb{U}_K(r, j)$ will be those response patterns $s$ such that the variables $x_{r\oplus e_j}$ are observed (i.e., $s> r\oplus e_j$)
and $x_s$ has $K-1$ additional observed variables than $x_{r\oplus e_j}$.

\begin{theorem}[Optimal imputation value under MKO]
\label{thm::mko}
For an observation $(x_r,r)$, let $j \in \bar r$ be the index of an unobserved variable. 
Then the imputation value
\begin{equation}
\hat x^*_j = {\sf argmin}_{\theta} \int L( x_j,\theta ) p(x_j|x_{r}, R\in \mathbb{U}_K(r, j))dx_j
\label{eq::opt::value::mko}
\end{equation}
is the optimal imputation value when minimizing $ {\mathcal{E}}_{K }(q )$.
\end{theorem}

A key technique in the proof of Theorem~\ref{thm::mko}
is the reparameterization technique in Lemma \ref{lem::repar}  (also appear in Section~\ref{sec::sep})
that we changes the individual-view of the loss (summation over index of individual $i$ first)
to imputation model-view (summation over $r$ first). 
This technique is how   $J_K(r)$ and 
$\mathbb{U}_K(r, j)$
are associated: for individual with missing pattern $r$, $J_K(r)$
is all possible masked pattern that can occur while the set $\mathbb{U}_K(r, j)$
is what patterns can contribute to imputing variable $ x_j$ when $ x_j$  is masked.

Based on Theorem~\ref{thm::mko}, if we use the square loss $L(a,b) = (a-b)^2$,
$$
\hat x^*_j = \E(X_j|X_r = x_r, R\in \mathbb{U}_K(r, j))
$$
is the mean value of the conditional distribution $p(x_j|x_r, R\in \mathbb{U}_K(r, j))$.
Thus, the conditional distribution $p(x_j|x_r, R\in \mathbb{U}_K(r, j))$
can be viewed as the optimal imputation distribution
for the MKO procedure. 

It is critical to note that the MKO procedure, by using an additive marginal loss, still only evaluates the marginal properties of an imputation model. It does not evaluate the joint dependency structure of the imputed variables. As Theorem \ref{thm::mko} shows, the optimal imputation model is still a deterministic, point-wise imputation (e.g., the conditional mean). This demonstrates that simply increasing the number of masked variables does not, by itself, solve the deterministic imputation problem that we identified as the core limitation of MOO.

{\bf Mask-all-out (MAO).} Now we consider a special case $K=d$
which we call masking-all-out (MAO).
In this case, the set 
$$
\mathbb{U}_d(r, j) = \{s\in\{0,1\}^d: s\geq r\oplus e_j\}
$$
is the collection of all patterns where the variable $x_{r\oplus e_j}$ is observed. 
The resulting optimal imputation value has a nice form:
\begin{equation}
\hat x^*_j = {\sf argmin}_{\theta} \int L( x_j,\theta ) p(x_j|x_{r}, R\geq r \oplus e_j)dx_j.
\label{eq::opt::value::mao}
\end{equation}
The imputation model
uses every case as long as the variables $x_r, x_j$ are both observed. 

%
%
Therefore, the resulting optimal imputation value can be easily estimated
since we have the highest effective sample size. 
On the other hand, the MOO is the case where we have least information in
estimating the optimal imputation value because we are using only
the pattern $R = r \oplus e_j$.
Note that we can also construct an MKO or MAO likelihood similar to the MOO likelihood.
We provide an example of learning an imputation model based on this in Section \ref{sec::learning::mao}.

\begin{example}
Consider the data in Table \ref{tab::non}
and let $q$ be an imputation model.
When we apply the MKO to ID=001 with $K=2$, 
we have 6 possible masking scenario: 
masking $\bX_{11}, \bX_{12}, \bX_{13}$ individually,
and masking two-variables at the same time $(\bX_{11}, \bX_{12}), (\bX_{11}, \bX_{13}), (\bX_{12}, \bX_{13})$
which corresponds to
$$
J_2(111) = \{100,010,001,110,101,011\}.
$$
When we mask $\bX_{11}, \bX_{12}$, we will draw 
$$
(\hat \bX_{11}, \hat \bX_{12}) \sim q(x_1,x_2|X_3 = \bX_{13}, R = 001)
$$
and compute the loss $L(\bX_{11}, \hat \bX_{11}) + L(\bX_{12}, \hat \bX_{12})$.
The total  loss value of this individual will be 
the summation of the losses
under the six masking scenarios. 
Note that if we apply MAO to ID=001, 
we will have one additional masking scenario: masking all three variables:
$$
J_3(111) = \{100,010,001,110,101,011, 111\}.
$$
Consider the observations  with $\bR_i=010$ and  suppose our goal is to impute $X_1$, 
the MKO with $K=2$ leads to
$$
\mathbb{U}_2(010, 1) =\{110, 111\},
$$
which means that we will use observations with $\bR_i=110$ or $111$
to evaluate the performance of the imputation model on imputing $X_1$  from $\bR_i = 010$.

\end{example}

\begin{remark}
Under MKO, the losses are generally higher
for observations with more observed entries.
Take the data in Table \ref{tab::non} as an example. 
Individual ID=001 has three observed variables. 
So under MKO with $K=2$, 
we will evaluate a total of $3 + 3\times 2 = 9$ loss functions:
the first three comes from masking one variables
and the latter six ($3 \times 2$) is from masking two variables--when masking two variables,
the loss function will be evaluated twice per each imputation. 
On the other hand, for ID=002, there are only two observed entries.
So the loss function will be evaluated only $2 + 2\times 2 = 6$ times. 
Elementary calculation shows that for a variable with $L$ observed entries,  MKO will require evaluating
the loss function $\sum_{k=1}^K {L \choose k} \cdot k$ times.
\end{remark}

\subsection{Challenges on characterizing the maximizer of the MAO likelihood}	\label{sec::max::mao}

While the MAO likelihood defines a criterion that incorporates interactions among variables 
during imputations, its maximizer cannot be easily characterized.
Similar to equation \eqref{eq::MOO::likelihood3}, we may define the (population) MAO log-likelihood for any imputation model $q$
as 
$$
\bar \ell_{MAO}(q)  = \E\left\{\sum_{s < \bR_1} \log q_\theta(x_s =\bX_{1,s}|x_{s} = \bX_{1, s}, r = \bR_1\ominus s)\right\}.
$$
Theorem~\ref{thm::MOO::ID}
shows that the class $\mathcal{Q}^*_{MOO}$
characterizes a class of distribution such that $\bar \ell(q)$
is maximized. 
One may wonder whether we can also find a similar set 
$\mathcal{Q}^*_{MAO}$
that maximizes this likelihood. 
A naive approach of generalizing $\mathcal{Q}^*_{MOO}$ to MAO is the following set
\begin{equation}
\mathcal{Q}^\dagger_{MAO} =\{q: q(x_s|x_r, r) = p(x_s|x_r, R\geq r\oplus s), \quad \forall s \leq  \bar r,\quad r \in\{0,1\}^d\}.
\label{eq::MAO::opt::set}
\end{equation}
While seemingly reasonable, this set may be an empty set!
The high level idea  is that
we cannot guarantee that an imputation  model for
multiple variables can be marginalized into an imputation model on 
fewer variables that are compatible with other patterns.
The following is an example highlighting this issue.

\begin{example}[Failure of $\mathcal{Q}^\dagger_{MAO} $]
Consider a bivariate data $X=(X_1,X_2)$.
In this case, we have four response patterns $R \in\{11,10,01,00\}$. 
Now we assume that 
$$
p(x_1|R=10)  = \begin{cases}
2,\quad\mbox{with a probability of } 0.5\\
3,\quad\mbox{with a probability of } 0.5
\end{cases}
$$
and 
$$
p(x_1,x_2|R=11)\sim {\sf Uni}[0,1]^2.
$$
In this case, there is no imputation model $q(x_1,x_2|R=00)$ in $\mathcal{Q}^\dagger_{MAO} $
that can perfectly fit to both distributions.
This is because when we mask $x_1$ in $R=10$, the optimal imputation model $q(x_1|R=00)$
should always predict a value that is either $2$ or $3$. 
On the other hand, when we mask $x_1,x_2$ together in $R=11$,
the optimal imputation model $q(x_1,x_2|R=00)$ should be a uniform distribution of the region $[0,1]^2$. 
Because we need to cover both $R=11$ and $R=10$ when imputing $x_1$,
$q(x_1|R=00)$ has a support on $[0,1]\cup\{2,3\}$. 
This conflicts with the support of the joint distribution $q(x_1,x_2|R=00)$!
Note that this does not conflict with Theorem~\ref{thm::mko}. There are optimal imputation values under say squared loss.
For instance, suppose $P(R=11) = P(R=10)$ then the optimal imputation value $\hat x_1 = 1.5$, which is the average of the mean value $X_1$
under the two models. 
\end{example}

\section{Variable-wise MOO procedure}	\label{sec::variablewise}

\subsection{Variable-wise MOO}	\label{sec::moo::v}

The MOO procedure in Algorithm~\ref{alg::moo} can be applied to a specific variable, leading to a risk value for that particular variable.
The idea is very simple. 
Instead of looking  at every individual and every variable in the MOO procedure, 
we focus on individuals with the  variable of interest, for instance $X_j$, being observed. 
We only mask this variable and impute it to compare the difference.
This variable-wise MOO procedure for an imputation model $q$ is summarized in Algorithm \ref{alg::moo::v}.

\begin{algorithm}
\caption{Variable-wise MOO procedure}
\label{alg::moo::v}
\begin{algorithmic}
\State \textbf{Input:} Imputation model $q$ and the variable of interest $X_j$.
\State Let $D_j =\{i: \bR_{ij} = 1\}$ be those individuals with $X_j$ being observed.
\begin{enumerate}
\item For $i\in D_j$, we do the following:
	\begin{enumerate}
	\item We mask the observed entry $\bX_{ij}$, pretending it to be missing.
	\item We impute $\hat \bX_{ij}$ by sampling from $q(x_{j}|X_{\bR_i \ominus e_j} = \bX_{i, \bR_i \ominus e_j}, R = \bR_i \ominus e_j)$. 
	\item Compute the loss value $L(\bX_{ij}, \hat \bX_{ij})$. Note that such loss function may vary from variable to variable. 
	\end{enumerate}
\item Compute the total risk of the imputation model $q$ as 
$$
\hat{\mathcal{E}}_{j, n}(q ) =\frac{1}{n} \sum_{i\in D_j} L(\bX_{ij}, \hat \bX_{ij})
$$
\end{enumerate}
\end{algorithmic}
\end{algorithm}

The quantity $\hat{\mathcal{E}}_{j, n}(q )$ from the output of Algorithm \ref{alg::moo::v}
is the risk of imputation model $q$ for variable $X_j$. 
It can be viewed as the risk of imputing variable $X_j$ under the imputation model $q$.

Interestingly, if we use the risk of all variables, we have the equality
$$
\sum_{j=1}^d\hat{\mathcal{E}}_{j, n}(q ) = \hat{\mathcal{E}}_{n}(q ),
$$
which is the total risk in the original MOO procedure.
This is because the original MOO procedure can be viewed as computing the loss for each row of the data matrix
and then summing over every row. 
This risk is the same if we perform a column-wise MOO of the data matrix and sum over every column. 
The variable-wise MOO is exactly the column-wise MOO of the data matrix.

\subsection{Variable-wise MOORT and MOOEN}	\label{sec::MOORT::v}

The idea of MOORT and MOOEN can be modified into criteria
for comparing imputation models for a specific variable as well.
Here we describe the procedure for MOORT; the case of MOOEN follows in a similar way.

Suppose we are interested in variable $X_j$. 
We modify Step 1 in Algorithm \ref{alg::MOORT} 
so that we only consider observations where $X_j$ is observed, i.e., $\bR_{ij}=1$,
and instead of randomly select a variable to mask, 
we always mask variable $X_{j}$. 
The rest steps remain the same and we will obtain a risk score 
for each imputation model under this variable. 
Algorithm \ref{alg::MOORT::v} summarizes the whole procedure.

\begin{algorithm}
\caption{Variable-wise MOORT}
\label{alg::MOORT::v}
\begin{algorithmic}
\State \textbf{Input:} An imputation model $q$ and a distance of distribution $d$ and a variable to be compared $X_j$
\State Let $D_j = \{i: \bR_{ij}=1\}$ be the indices of the observations where variable $X_j$ is observed.
\begin{enumerate}
\item For each $i\in D_j$, we do the following.
\item We mask $\bX_{ij}$, pretending it to be a missing value.
\item Using imputation model $q$, we sample $M$ times from the conditional distribution $q(x_j | X_{\bR_i \ominus e_j} =  \bX_{i, \bR_i \ominus e_j}, R=\bR_i \ominus e_j)$ to generate $M$ imputed values:
$$
\hat \bX_{ij}^{(1)},\cdots, \hat \bX_{ij}^{(M)}.
$$
\item We compute the EDF of these $M$ values: $\hat G_{\bX_{ij}}(x) = \frac{1}{M}\sum_{m=1}^M I(\hat \bX_{ij}^{(m)}\leq x)$. 
\item We compute the score $\hat S_i = \hat G_{\bX_{ij}}(\bX_{ij})$.
\item By doing so for every individual, we obtain $\{ \hat S_i: i \in D_j\}$ and the corresponding empirical distribution $\hat H_j(t; q) = \frac{1}{|D_j|}\sum_{i\in D_j} I(\hat S_i\leq t),$
where $|D_j| = \sum_{i=1}^n I(i\in D_j)$ is the cardinality of $D_j$.
\item We compute
$$
\hat{\mathcal{R}}_j(q) = d\left(\hat H_j(\cdot; q), {\sf Uni}[0,1]\right).
$$
\end{enumerate}

\end{algorithmic}
\end{algorithm}

The output $\hat{\mathcal{R}}_j(q)$ from Algorithm \ref{alg::MOORT::v}
can be viewed as a measure of imputation performance on variable $X_j$.
We may apply Algorithm \ref{alg::MOORT::v} to every variable and compute
$\bar{\mathcal{R}}(q) = \sum_{j=1}^d \hat{\mathcal{R}}_j(q)$ as
an alternative criterion for evaluating the overall imputation performance.
The following theorem shows a recovery result under $\bar{\mathcal{R}}(q)$.

\begin{theorem}
\label{thm::MOORT::v}
Consider the variable-wise MOORT procedure for every variable with stochastic rank for categorical and discrete random variables.
Let $d$ be a metric for distribution. 
When $n,M\rightarrow\infty$, 
$$
\bar{\mathcal{R}}(q) \overset{P}{\rightarrow} 0
$$
for any $q \in \mathcal{Q}^*_{MOO}$.
\end{theorem}

Thus, 
we may minimize $\bar{\mathcal{R}}(q) = \sum_{j=1}^d \hat{\mathcal{R}}_j(q)$ to select the best imputation model. 
Compared with MOORT, 
this procedure will  have a higher computational cost (since we have to go through all observed entries)
but it has less randomness.

\section{Learning a separable imputation model using MOO likelihood}	\label{sec::sep}

When the imputation model's parameters are separable among different response patterns, 
learning/estimating the parameters using the MOO likelihood in equation \eqref{eq::MOO::likelihood2} can be decomposed into
several sub-problems, reducing the complexity of learning.
We start with a simple Gaussian example highlighting this feature.

\begin{example}[A separable Gaussian model]	\label{ex::sep}
Suppose that every individual has two variables $X = (X_1,X_2) \in \R^2$
and there are a total of four possible missing patterns: $R \in \{11, 10, 01, 00\}$.
Let the data be $(\bX_{1, \bR_1},\bR_1),\cdots, (\bX_{n, \bR_n}, \bR_n)$.
Among these missing patterns, any imputation model consists of three submodels:
$$q(x_1|x_2, 01), q(x_2|x_1, 10), q(x_1,x_2|00).
$$
We assume that these three models are parametrized as 
\begin{align*}
q(x_1|x_2, 01) &\sim N(\mu_{01} + \beta_{01} x_2, \sigma^2_{01})\\
q(x_2|x_1, 10) &\sim N(\mu_{10} + \beta_{10} x_1, \sigma^2_{10})\\
q(x_1,x_2| 00) &\sim N((\mu_{00,1}, \mu_{00,2})^T, \sigma^2_{00} \mathbf{I}_2).
\end{align*}
The above parametric model satisfies assumptions (A1-3) in Theorem \ref{thm::MLE}. 
Note that the off-diagonal term in the covariance matrix of $q(x_1,x_2| 00) $ has to 
be a fixed quantity (like $0$ in the above model) otherwise 
the MLE is not unique--this is because the MOO log-likelihood procedure does not learn
the dependency.

In this case, the model's parameters  $\theta = (\mu_{01}, \mu_{10}, \mu_{00,1}, \mu_{00,2}, \beta_{01}, \beta_{10}, \sigma^2_{01},\sigma^2_{10},\sigma^2_{00})$.
The parameters are separable in this case because each submodel uses a different set of parameters: $(\mu_{01}, \beta_{01}, \sigma^2_{01})$,
$(\mu_{10}, \beta_{10}, \sigma^2_{10})$, and $(\mu_{00,1}, \mu_{00,2}, \sigma^2_{00})$.
Thus, we can reparametrize the MOO log-likelihood function as 
\begin{equation}
\begin{aligned}
\ell_n(\theta) & = \sum_{i=1}^n\sum_{j \in \bR_i} \log q_\theta(x_j = \bX_{ij}|x_{\bR_i\ominus e_j} = \bX_{i, \bR_i\ominus e_j}, r = \bR_i\ominus e_j)\\
& = \sum_{i=1}^n I(\bR_i = 11) \log q(x_1 = \bX_{i1}|x_2 = \bX_{i2}, 01) \\
&\qquad+ \sum_{i=1}^n I(\bR_i = 11)\log q(x_2 = \bX_{i2}|x_1 = \bX_{i1}, 10) \\ 
&\qquad+ \sum_{i=1}^n I(\bR_i = 10) \log q(x_1 = \bX_{i1}| 00) + I(\bR_i = 01) \log q(x_2 = \bX_{i2}| 00).
\end{aligned}
\label{eq::sep}
\end{equation}
In this easy case, 
the MLE has the following closed-form solution: 
\begin{equation}
\begin{aligned}
(\hat  \mu_{01} ,\hat \beta_{01}) &= {\sf argmin}_{\mu ,\beta}\,\,  \sum_{i=1}^n I(\bR_i = 11) (\bX_{i1} - \mu - \beta\cdot \bX_{i2})^2\\
\hat \sigma^2_{01} &= \frac{\sum_{i=1}^n I(\bR_i = 11) (\bX_{i1} - \hat\mu_{01} - \hat\beta_{01}\cdot \bX_{i2})^2}{n_{11} }\\
(\hat  \mu_{10} ,\hat \beta_{10}) &= {\sf argmin}_{\mu ,\beta}\,\,  \sum_{i=1}^n I(\bR_i = 11) (\bX_{i2} - \mu - \beta\cdot \bX_{i1})^2\\
\hat \sigma^2_{10} &= \frac{\sum_{i=1}^n I(\bR_i = 11) (\bX_{i2} - \hat\mu_{10} - \hat\beta_{10}\cdot \bX_{i1})^2}{n_{11}}\\
\hat \mu_{00,1} & = \frac{1}{n_{10}}\sum_{i=1}^n I(\bR_i = 10) \bX_{i1}\\
\hat \mu_{00,2} & = \frac{1}{n_{01}}\sum_{i=1}^n I(\bR_i = 01) \bX_{i2}\\
\hat \sigma^2_{00} & = \frac{1}{n_{10}+n_{01}}\sum_{i=1}^n [I(\bR_i = 10) (\bX_{i1}-\hat \mu_{00,1})^2 + I(\bR_i=01) (\bX_{i2}-\hat \mu_{00,2})^2]
\end{aligned}
\label{eq::mle}
\end{equation}
where $n_r = \sum_{i=1}^n I(\bR_i=r)$ is the number of observations of response pattern $r$.
\end{example}

The reparameterization in equation \eqref{eq::sep}
is useful because we rewrite the likelihood function from an individual-view (summation over $i$ first) to pattern-view (summation over $r$ first).
A more general form of reparameterization is given in Lemma \ref{lem::repar}.
For a general MOO log-likelihood model, the same reparameterization leads to
\begin{equation}
\begin{aligned}
\ell_n(\theta) & = \sum_{i=1}^n\sum_{j \in \bR_i} \log q_\theta(x_j = \bX_{ij}|x_{\bR_i\ominus e_j} = \bX_{i, \bR_i\ominus e_j}, r = \bR_i\ominus e_j)\\
& = \sum_{r:r\neq 1_d} \sum_{i=1}^n   \sum_{j \in \bar r} I(\bR_i = r\oplus e_j)\log q_\theta(x_j = \bX_{ij}|x_{r} = \bX_{i,r}, r) .
\end{aligned}
\label{eq::sep2}
\end{equation}
If the parameters in the imputation model are separable  in the sense that $\theta = (\theta_r: r \in \{0,1\}^d\backslash 1_d)$, then 
equation \eqref{eq::sep2} can be further decomposed into 
\begin{equation}
\begin{aligned}
\ell_n(\theta)
& = \sum_{r:r\neq 1_d} \sum_{i=1}^n   \sum_{j \in \bar r} I(\bR_i = r\oplus e_j)\log q_\theta(x_j = \bX_{ij}|x_{r} = \bX_{i,r}, r) \\
& = \sum_{r:r\neq 1_d} \sum_{i=1}^n   \sum_{j \in \bar r} I(\bR_i = r\oplus e_j)\log q_{\theta_r}(x_j = \bX_{ij}|x_{r} = \bX_{i,r}, r)\\
& = \sum_{r:r\neq 1_d} \ell_{r, n}(\theta_r),
\end{aligned}
\label{eq::sep3}
\end{equation}
where $\ell_{r, n}(\theta_r) = \sum_{i=1}^n   \sum_{j \in \bar r} I(\bR_i = r\oplus e_j)\log q_{\theta_r}(x_j = \bX_{ij}|x_{r} = \bX_{i,r}, r)$
is the log-likelihood function of the parameter $\theta_r$.
A major benefit of the separable model is that we can learn different sets of 
parameters separately. This is particularly useful when the number of parameters is huge.

Note that even if we choose a model with separable parameters,
the estimated parameters may still be dependent. 
In the Gaussian example, the complete observations, $\bR_i = 11$,
are used in estimating both $q_{01}$ and $q_{10}$. 
Thus, the uncertainty of estimated parameters is correlated.

\subsection{Separable product imputation model}	\label{sec::sepprod}
The imputation model may be further separated if the submodel parameters $\theta_r$ for pattern $r$
can be  decomposed into $\theta_r = (\theta_{r,j}: j \in \bar r)$
and the imputation model has a product form as in equation \eqref{eq::MOOPM}:
$$
q_{\theta_r}(x_{\bar r} |x_r, r) = \prod_{j \in \bar r}q_{\theta_{r,j}}(x_j |x_r, r).
$$
In this case, 
$$
q_{\theta_r}(x_j = \bX_{ij}|x_{r} = \bX_{i,r}, r) = q_{\theta_{r,j}}(x_j = \bX_{ij}|x_{r} = \bX_{i,r}, r)
$$
and
equation \eqref{eq::sep3} is  decomposed into
\begin{align*}
\ell_n(\theta) &=\sum_{r:r\neq 1_d} \sum_{i=1}^n   \sum_{j \in \bar r} I(\bR_i = r\oplus e_j)\log q_{\theta_r}(x_j = \bX_{ij}|x_{r} = \bX_{i,r}, r)\\
&=\sum_{r:r\neq 1_d} \sum_{i=1}^n   \sum_{j \in \bar r} I(\bR_i = r\oplus e_j)\log q_{\theta_r,j}(x_j = \bX_{ij}|x_{r} = \bX_{i,r}, r)\\
& =  \sum_{r:r\neq 1_d} \sum_{j \in \bar r} \ell_{r,j, n}(\theta_{r,j}),
\end{align*}
where $\ell_{r,j, n}(\theta_{r,j}) = \sum_{i=1}^n I(\bR_i = r\oplus e_j)\log q_{\theta_r,j}(x_j = \bX_{ij}|x_{r} = \bX_{i,r}, r)$.
Thus, we can separately learn each $\theta_{r,j}$
and the observations contributing to learning this parameter have the same response pattern $R = r\oplus e_j$.

In the  simple Gaussian example of Example \ref{ex::sep}, this occurs when we modify the covariance matrix of $q(x_1,x_2|00)$ 
so that
$$
q(x_1,x_2| 00) \sim N\left(\begin{bmatrix}\mu_{00,1}\\
\mu_{00,2}
\end{bmatrix},
\begin{bmatrix}
\sigma_{00,1}^2 & 0\\
0 &\sigma_{00,2}^2
\end{bmatrix}\right).
$$
It is clear that $q(x_1,x_2| 00) = q(x_1| 00)q(x_2| 00)$
and 
$q(x_j| 00)\sim N(\mu_{00,j}, \sigma_{00,j}^2)$ for $j=1,2$
and the last part of equation \eqref{eq::sep}
becomes
$$
\sum_{i=1}^n I(\bR_i = 10) \log q(x_1 = \bX_{i1}| 00) +\sum_{i=1}^n I(\bR_i = 01) \log q(x_2 = \bX_{i2}| 00),
$$
so learning $(\mu_{00,1},\sigma_{00,1}^2)$ and $(\mu_{00,2}, \sigma_{00,2}^2)$
can be done separately.

%

\subsection{Shared parameters in imputation model}	\label{sec::pool}

In the previous section, we introduce the idea of separable product imputation model that
is flexible for imputation. 
However, this approach may suffer from low sample size for learning certain parameters,
leading to an unstable estimate of the model's parameters. 
To address this issue, we consider a shared parameter approach. Specifically,
we equate some imputation submodel's parameters so that we may pool observations with different missing patterns
to estimate the parameter.

To illustrate the idea, again we consider the Gaussian example in Section \ref{sec::sepprod} (Example \ref{ex::sep}):
\begin{equation}
\begin{aligned}
q(x_1|x_2, 01) &\sim N(\mu_{01} + \beta_{01} x_2, \sigma^2_{01})\\
q(x_2|x_1, 10) &\sim N(\mu_{10} + \beta_{10} x_1, \sigma^2_{10})\\
q(x_1,x_2| 00) &\sim N\left(\begin{bmatrix}\mu_{00,1}\\
\mu_{00,2}
\end{bmatrix},
\begin{bmatrix}
\sigma_{00,1}^2 & 0\\
0 &\sigma_{00,2}^2
\end{bmatrix}\right).
\end{aligned}
\label{eq::shared}
\end{equation}
Suppose we include a constraint that 
\begin{equation}
\sigma^2_{01} = \sigma_{00,1}^2,
\label{eq::sigma}
\end{equation}
which means that the variance parameters are shared across the two models. 
Under this constraint, one can easily show that the MLE of $\sigma^2_{01}$ in equation \eqref{eq::mle} is updated to 
$$
\tilde \sigma^2_{01} = \frac{\sum_{i=1}^n I(\bR_i = 11) (\bX_{i1} - \hat\mu_{01} - \hat\beta_{01}\cdot \bX_{i2})^2 + I(\bR_i = 10) (\bX_{i1}- \hat \mu_{00,1})^2}{n_{11}+n_{10} }.
$$
The original variance estimator in equation \eqref{eq::mle} has an effective sample size $n_{11}$ while
the effective sample size of the new variance estimator $\tilde \sigma^2_{00,1}$ is $n_{11}+n_{10}$. 
When the complete case sample size $n_{11}$ is small and the partial response's sample size $n_{10}$ is large,
the equality constraint in equation \eqref{eq::sigma}
offers a huge boost to the sample size.

The shared parameter model approach is particularly useful when the 
number of missing variables is high. 
To see this, when there are $d$ variable that can be missing,
there are a total number of $2^d$ possible missing patterns,
so we need  $2^d$ distinct imputation submodels  $q_{\theta_r}$. 
Consequently, the size of all parameters scale at the order of $O(2^d)$,
so the complexity diverges quickly with respect to $d$.

\subsection{Learning imputation model under MAO}	\label{sec::learning::mao}
The statistical learning approach in Section \ref{sec::learning}
can be generalized to the MKO scenario. 
To simplify the problem, we consider the case of the MAO.
Specifically, we  modify the MOO log-likelihood in equations \eqref{eq::MOO::likelihood}
and \eqref{eq::MOO::likelihood2}
to the MAO log-likelihood
\begin{equation}
\begin{aligned}
\ell_{MAO,n}(\theta)&=  \sum_{i=1}^n \ell_K(\theta|\bX_{i,\bR_i},\bR_i) \\
& = \sum_{i=1}^n\sum_{\ell < \bR_i} \log q_\theta(x_\ell =\bX_{i,\ell}|x_{ \bR_i\ominus \ell} = \bX_{i, \bR_i\ominus \ell}, r = \bR_i\ominus \ell).
\label{eq::MAO::likelihood}
\end{aligned}
\end{equation}
Using the reparameterization technique in equation \eqref{eq::sep2},
we can rewrite the MAO log-likelihood from
the individual-view (summation over observation index $i$ first)
to the imputation model-view (summation over response pattern $r$):
\begin{equation}
\begin{aligned}
\ell_{MAO,n}(\theta)& = \sum_{r:r\neq 1_d} \sum_{i=1}^n \sum_{s\leq  \bar r} I( \bR_i =  s+r) \log q_{\theta_r}(x_s = \bX_{i,s}|x_r = \bX_{i,r}, R=r)\\
& =  \sum_{r:r\neq 1_d} \ell_{MAO, r,n} (\theta_r),\\
\ell_{MAO, r,n} (\theta_r) & = \sum_{i=1}^n \sum_{s\leq  \bar r} I( \bR_i =  s+r) \log q_{\theta_r}(x_s = \bX_{i,s}|x_r = \bX_{i,r}, R=r).
\end{aligned}
\end{equation}
Note that in the above scenario, we consider a separable imputation model (Section \ref{sec::sep})
so that we can partition the parameter $\theta = (\theta_r: r\in\{0,1\}^d\backslash 1_d)$. 

The function $\ell_{MAO, r,n} (\theta_r)$
is the MAO log-likelihood function
for the imputation submodel $ q_{\theta_r}$. 
Comparing $\ell_{MAO, r,n} (\theta_r)$ to the MOO version $\ell{r,n}(\theta_r)$ in equation \eqref{eq::sep3},
the MAO considers jointly imputing multiple missing variables 
whereas the MOO only consider imputing a single variable. 
Therefore, MAO allows learning imputation model's parameters on the 
interactions. 


\begin{example}[Example \ref{ex::sep} revisited]
To see how MAO learns the joint distribution of missing variables, we consider again the 
Gaussian example in Section \ref{sec::sep}.
Now we consider the imputation models
\begin{align*}
q(x_1|x_2, 01) &\sim N(\mu_{01} + \beta_{01} x_2, \sigma^2_{01})\\
q(x_2|x_1, 10) &\sim N(\mu_{10} + \beta_{10} x_1, \sigma^2_{10})\\
q(x_1,x_2| 00) &\sim N\left(\underbrace{\begin{bmatrix}\mu_{00,1}\\
\mu_{00,2}
\end{bmatrix}}_{=\mu_{00}},
\underbrace{\begin{bmatrix}
\Sigma_{00,11} & \Sigma_{00,12}\\
\Sigma_{00,21} &\Sigma_{00,22}
\end{bmatrix}}_{=\Sigma_{00}}\right).
\end{align*}
Note that we now allow the covariance matrix of $q(x_1,x_2| 00)$
to have off-diagonal terms.
The MAO log-likeihood function will be 
\begin{equation}
\begin{aligned}
\ell_{MAO, n}(\theta)
& = \sum_{i=1}^n I(\bR_i = 11) \log q(x_1 = \bX_{i1}|x_2 = \bX_{i2}, 01) \\
&\qquad+ \sum_{i=1}^n I(\bR_i = 11)\log q(x_2 = \bX_{i2}|x_1 = \bX_{i1}, 10) \\ 
&\qquad+ \sum_{i=1}^n I(\bR_i = 10) \log q(x_1 = \bX_{i1}| 00) + I(\bR_i = 01) \log q(x_2 = \bX_{i2}| 00)\\
&\qquad+ \sum_{i=1}^n I(\bR_i = 11) \log q(x_1 = \bX_{i1}, x_2 = \bX_{i2}| 00).
\end{aligned}
\label{eq::mao::sep}
\end{equation}
Comparing to Equation \eqref{eq::sep}, we have an additional term, $\sum_{i=1}^n I(\bR_i = 11) \log q(x_1 = \bX_{i1}, x_2 = \bX_{i2}| 00)$,
which is the likelihood when we mask both variables from complete data.
Under the MAO procedure, learning parameters of $q(x_1|x_2, 01)$ and $q(x_2|x_1, 10)$
remains the same as MOO. 
However, learning parameters of $q(x_1,x_2| 00)$
will also use the data of complete case ($\bR_i = 11$).

It turns out that the MLEs of $\mu_{00}$ and $\Sigma_{00}$ do not have a closed-form solution
but we can easily maximize it by a gradient ascent approach. 
Here are the gradients of them:
\begin{align*}
\nabla_{\mu_{00}} \ell_{MAO, n}(\theta) &=
\sum_{i=1}^n  \begin{bmatrix}
\Sigma^{-1}_{00,11}&0\\
0& \Sigma^{-1}_{00,22}
\end{bmatrix}
\begin{bmatrix}
 I(\bR_i = 10) (\bX_{i1}- \mu_{00,1})\\
 I(\bR_i = 01) (\bX_{i2}- \mu_{00,2})
 \end{bmatrix}
 +
 \Sigma_{00}^{-1} \begin{bmatrix}
I(\bR_i = 11) (\bX_{i1}- \mu_{00,1})\\
I(\bR_i = 11) (\bX_{i2}- \mu_{00,2})
 \end{bmatrix}\\
\nabla_{\Sigma_{00}} \ell_{MAO, n}(\theta) &=
\begin{bmatrix}
S_{10,n}(\theta)
&0\\
0&S_{01,n}(\theta)
\end{bmatrix}
-\frac{n_{11}}{2}\Sigma_{00}^{-1} + \frac{1}{2}\Sigma_{00}^{-1} \left[\sum_{i=1}^n I(\bR_i=11) (X_{i} - \mu_{00})(X_{i} - \mu_{00})^T\right]\Sigma_{00}^{-1}\\
S_{10,n}(\theta)& =-\frac{n_{10}}{2}\Sigma^{-1}_{00,11} +\frac{1}{2}\Sigma^{-2}_{00,11}\sum_{i=1}^n I(\bR_i=10)(X_{i1} - \mu_{00,1})^2  \\
S_{01,n}(\theta)& = -\frac{n_{01}}{2}\Sigma^{-1}_{00,22} +\frac{1}{2}\Sigma^{-2}_{00,22}\sum_{i=1}^n I(\bR_i=01)(X_{i2} - \mu_{00,2})^2.
\end{align*}
\end{example}

\section{Monte Carlo approximation of MOO likelihood}	\label{sec::MCMOO}

When evaluating $q$ is intractable while sampling $q$ is easy, 
the rank transformation or energy distance approaches (Section \ref{sec::QT}) are generally recommended but we may still use the likelihood method. 
In this case, we can perform a Monte Carlo approximation of the likelihood. 
Suppose $\bX_{ij}$ is the masked variable and we have an imputation model $q$ that we want to evaluate the likelihood.
We first  perform a multiple imputation: $\hat \bX_{ij}^{(1)},\cdots, \hat \bX_{ij}^{(M)}$
and then estimate the density using $\hat \bX_{ij}^{(1)},\cdots, \hat \bX_{ij}^{(M)}$ evaluated at $\bX_{ij}$. 
One simple example is the kernel density estimation:
\begin{equation}
\tilde q(x_j = \bX_{ij}|x_r = \bX_{i, r}, r= \bR_i\ominus e_j) = \frac{1}{Mh}\sum_{m=1}^M K\left(\frac{\bX_{ij} -  \hat \bX_{ij}^{(m)}}{h}\right),
\label{eq::kde}
\end{equation}
where $h>0$ is the smoothing bandwidth and $K(\cdot) >0$ is the smoothing kernel such as a Gaussian. 
Fortunately, in the case of MOO, we only need to apply a univariate KDE so 
there will be no curse-of-dimensionality and the choice of $h$ is a relatively
easy problem.
Algorithm \ref{alg::MC::likelihood} summarizes this procedure.
When using this Monte Carlo method,
it is like performing a multiple imputation on the observed entries
and then evaluating the density.

\begin{algorithm}
\caption{Monte Carlo approximation of MOO log-likelihood}
\label{alg::MC::likelihood}
\begin{algorithmic}
\State \textbf{Input:} Imputation model $q$ that is easy to sample from.
\begin{enumerate}
\item For each individual $i=1,\cdots, n$, 
and each observed entry $j \in \bR_i$. 
\item We mask  $\bX_{ij}$, pretending it to be a missing value.
\item Using imputation model $q$, we sample $M$ times from the conditional distribution 
$$
q(x_j |x_r = \bX_{i, r},r=\bR_i \ominus e_j) \equiv q(x_j |X_{\bR_i \ominus e_j} = \bX_{i, \bR_i \ominus e_j},R=\bR_i \ominus e_j)
$$ 
to generate $M$ imputed values:
$$
\hat \bX_{ij}^{(1)},\cdots, \hat \bX_{ij}^{(M)}.
$$
\item We compute the estimated density $\hat q(x_j = \bX_{ij}| \bX_{i, \bR_i\ominus e_j},  \bR_i\ominus e_j)$.
For continuous variable, we may use the KDE in equation \eqref{eq::kde}:
$$
\tilde q(x_j = \bX_{ij}| x_r = \bX_{i, r}, r= \bR_i\ominus e_j) = \frac{1}{Mh}\sum_{m=1}^M K\left(\frac{\bX_{ij} -  \hat \bX_{ij}^{(m)}}{h}\right).
$$
\item We compute the MOO likelihood
$$
\hat \ell_n(q) = \sum_{i=1}^n\sum_{j \in R_i} \log \tilde q(x_j = \bX_{ij}|x_r = \bX_{i, r}, r=\bR_i\ominus e_j).
$$
\end{enumerate}

\end{algorithmic}
\end{algorithm}

\section{Proofs}	\label{sec::proofs}

\subsection{Proof of Theorem~\ref{thm::moo}}
\begin{proof}
The key to this proof is that
we can view the problem of finding optimal imputation value as a risk minimization problem.
So finding the optimal imputation value is
similar to the derivation of a Bayes classifier.

By equation \eqref{eq::MOO::TR},
the population risk  is
\begin{align*}
\mathcal{E}(q) = \E\{L(q| \bX_{i,\bR_i},\bR_i)\} &= \sum_{r:r\neq 1_d}\int \bar L(q|x_{r} , r) p(x_r,r)dx_r,\\
\bar L(q|x_{r} , r) & = \sum_{j \in  r}\int L(x_j, x_j')  q(x'_j|x_{r\ominus e_j},R=r\ominus e_j)dx'_j,
\end{align*}
which further implies 
\begin{align*}
\mathcal{E}(q) & = \sum_{r:r\neq 1_d} \sum_{j \in  r} \int L(x_j, x_j')  q(x'_j|x_{r\ominus e_j},R=r\ominus e_j) p(x_r,r)dx'_j dx_r\\
&= \sum_{r:r\neq 1_d} \sum_{j \in  r} \underbrace{\left\{\int L(x_j, x_j')  q(x'_j|x_{r\ominus e_j},R=r\ominus e_j) p(x_j|x_{r\ominus e_j},r)dx_jdx'_j \right\}}_{ = \mathbf{L}_{r,j}} p(x_{r\ominus e_j}, r)dx_{r\ominus e_j}.
\end{align*}

Under the above decomposition, one can clearly see that
the imputation model $q(x_j|x_{r\ominus e_j},{r\ominus e_j})$ is applied separately for every variable $x_j$
and  every $(x_r,r)$. 
Thus, to investigate the optimal imputation model $q(x_j|x_r,r)$, we only need to consider the term 
\begin{equation}
\begin{aligned}
\mathbf{L}_{r,j}
 &=\int  L(x_j, x_j') p(x_j|x_{r\ominus e_j}, r) q(x'_j|x_{r\ominus e_j},R=r\ominus e_j)dx_jdx'_j,\\
&= \int \left\{\int L(x_j, x_j') p(x_j|x_{r\ominus e_j}, r) dx_j \right\}q(x'_j|x_{r\ominus e_j},R=r\ominus e_j)dx'_j.
\end{aligned}
\label{eq::Lj}
\end{equation}
The quantity in the bracket
has a deterministic minimizer 
$$
x'_j = \hat x^*_j = {\sf argmin}_{\theta} \int L( x_j,\theta ) p(x_j|x_{r\ominus e_j}, r)dx_j.
$$
Namely, we should choose $q(x'_j|x_{r\ominus e_j},R=r\ominus e_j)$ so that  
it is a degenerate density (point mass) at $x'_j = \hat x^*_j$.

Let $s = r\ominus e_j$ be the response pattern that differs from $r$ by $e_j$.
Clearly, $r = s\oplus e_j$
and we can rewrite equations \eqref{eq::Lj} 
as 
\begin{align*}
\mathbf{L}_{r,j} = \int \left\{\int L(x_j, x_j') p(x_j|x_s, s\oplus e_j) dx_j \right\}q(x'_j|x_{s},s)dx'_j,
\end{align*}
so the optimal imputation model $q(x_j|x_{s},s)$
should be the point mass $\hat x^*_j$ for $j \in \bar s$,
which completes the proof.

\end{proof}

\subsection{Proof of Proposition \ref{prop::prob}}
\begin{proof}

	\emph{Part 1: Conditional independence statement of $\mathcal{Q}^*_{MOO}$.}
	By construction, any imputation model $q\in\mathcal{Q}^*_{MOO}$ must satisfies
	$$
	p (x_j|x_r, r) = q(x_j|x_r, r) = p(x_j|x_r, r\oplus e_j)
	$$
	for any $j\in \bar r$.
	Since $j\in \bar r$, $r_j = 0$, so the above equality can be written as
	$$
	p(x_j|x_r, r_{-j}, R_j = 0) = p(x_j|x_r, r\oplus e_j) = p(x_j|x_r, r_{-j}, R_j = 1),
	$$
	which is the same as $X_j \perp R_j | X_{R}, R_{-j}$. 
	This holds for every $r$ and $j\in \bar r$, so we have the desired result.

	\emph{Part 2: Equivalence of the two conditional statements.} 
	We consider the direction $\eqref{eq::CI} \Rightarrow \eqref{eq::CI2}$ first.  Let $R\in\mathcal{R}$, and pick $j\in R$.  Then, define $R' \equiv R \ominus e_j$ as flipping the $j$-th element of $R$, so that $R_j' \equiv 1-R_j$ and $j\in \bar{R'}$.  From equation \eqref{eq::CI}, we have
		$$X_j \perp R_j' | X_{R'}, R'_{-j}.$$
		Observe that $R'_{-j} = R_{-j}$ and $X_{R'} = X_{R_{-j}}$.  We now have a series of equivalences
		\begin{align*}
			X_j \perp R_j' | X_{R'}, R'_{-j} &\iff X_j \perp R_j' | X_{R_{-j}}, R_{-j} \\
			&\iff X_j \perp (1-R_j) | X_{R_{-j}}, R_{-j} \\
			&\iff X_j \perp R_j | X_{R}, R_{-j},
		\end{align*}
		which is exactly equation \eqref{eq::CI2}.
		A similar argument proves the $\eqref{eq::CI2} \Rightarrow \eqref{eq::CI}$ direction.
	\end{proof}



\subsection{Proof of Theorem~\ref{thm::EIF}}
	\begin{proof}

		Since we can decompose 
		\begin{align*}
		\theta \equiv \E[X_1]  = \E\left[X_1\sum_r I(R=r)\right]  =  \sum_r \E[X_1I(R=r)] = \sum_r \theta_r.
		\end{align*}
		The EIF of $\theta$ will be the summation of EIF of each $\theta_r$.
		So we first consider the EIF for $\theta_r$.  

		To write the parameter in terms of the observables, we have the following decomposition
		for $r$ with $r_1 = 0$:
		\begin{align*}
			\theta_r &= \E[X_1 \cdot I(R=r)] \\
			&= \int x_1 \cdot p(x_1,r) dx_1\\
			&= \int x_1 \cdot p(x_1|x_r,r) \cdot p(x_r,r) dx_1dx_r \\
			&\stackrel{\eqref{eq::CI}}{=}  \int x_1  p(x_1|x_r,r\oplus e_1)dx_1 \cdot p(x_r,r) dx_r\\
			&= \int \mu_1(x_r, r\oplus e_1) \cdot p(x_r,r)dx_r.
		\end{align*}
		To avoid tedious computation that arises from deriving from first principles, it suffices to appeal to the approaches outlined in Section 3.4.3 of \cite{kennedy2023semiparametricdoublyrobusttargeted}.
		
		For all $r$ such that $r_1=0$, we have
		\begin{align*}
			\EIF(\theta_r) &= \EIF\left(\int \mu_1(x_r, r\oplus e_1) \cdot p(x_r,r)dx_r \right) \\
			&= \int \left[ \EIF(\mu_1(x_r, r\oplus e_1)) \cdot p(x_r,r) + \mu_1(x_r, r\oplus e_1) \cdot \EIF(p(x_r,r)) \right]dx_r \\
			&= \int \biggr[\frac{I(X_r=x_r,R=r\oplus e_1)}{p(x_r,r\oplus e_1)}(X_1 - \mu_1(x_r, r\oplus e_1)) \cdot p(x_r,r) \\
			&\qquad + \mu_1(x_r, r\oplus e_1) \cdot [I(X_r=x_r, R=r) - p(x_r,r)] \biggr] dx_r\\
			&= \frac{I(R=r\oplus e_1)}{p(X_r, r\oplus e_1)} (X_1 - \mu_1(X_r, r\oplus e_1)) \cdot p(X_r,r) \\
			& \qquad + I(R=r) \cdot \mu_1(X_r, r\oplus e_1) - \theta_r \\
			&= I(R=r\oplus e_1)\cdot O_1(X_r,r) (X_1 - \mu_1(X_r, r\oplus e_1)) + I(R=r) \cdot \mu_1(X_r, r\oplus e_1) - \theta_r.
		\end{align*}
		For patterns such that $r_1=1$, it is trivial to estimate $\theta_r$ using $I(R=r)X_1$.\\
	
		Now, summing over all patterns, we have the final EIF
		\begin{align*}
			\EIF(\theta) &= I(R_1=1)X_1 \\
			&\qquad + \sum_{r:r_1=0} \biggr[I(R=r\oplus e_1)\cdot O_1(X_r,r) (X_1 - \mu_1(X_r, r\oplus e_1)) \\
			&\qquad+ I(R=r) \cdot \mu_1(X_r, r\oplus e_1)\biggr] - \theta.
		\end{align*}
	\end{proof}

	\begin{proof}[Proof of Theorem~\ref{thm::MR}]
		We show unbiasedness for $\theta_r$ when at least one of the nuisance functions is correctly specified.  For all $r$ such that $r_1=0$, we show that $I(R=r\oplus e_1)\cdot O_1(X_r,r) (X_1 - \mu_1(X_r, r\oplus e_1)) + I(R=r) \cdot \mu_1(X_r, r\oplus e_1)$ is unbiased for $\theta_r$ as long as one of $O_1(X_r,r)$ and $\mu_1(X_r, r\oplus e_1)$ are correctly specified.  Under other appropriate regularity conditions, consistency and $\sqrt{n}$-asymptotic normality can be shown.\\
		
		\textit{Case 1: odds models are correctly specified.} Assume that $\mu_1(x_r,r\oplus e_1)$ may be misspecified but $O_1(x_r,r)$ is correctly specified.  
		Namely, the estimator $\hat \mu_1(x_r,r\oplus e_1)$ converges to $\tilde\mu_1(x_r,r\oplus e_1)\neq \mu_1(x_r,r\oplus e_1)$ while $\hat O_1(x_r,r)$ converges to the true $O_1(x_r,r)$.  We have
		\begin{align*}
			&\E[I(R=r\oplus e_1)\cdot O_1(X_r,r) (X_1 - \tilde\mu_1(X_r,r\oplus e_1)) + I(R=r) \cdot \tilde\mu_1(X_r,r\oplus e_1)] \\
			&\quad = \E[I(R=r\oplus e_1)\cdot O_1(X_r,r) \cdot X_1 + \tilde\mu_1(X_r,r\oplus e_1) \cdot (I(R=r)- I(R=r\oplus e_1) O_1(X_r,r))] \\
			&\quad =  \E[I(R=r\oplus e_1)\cdot O_1(X_r,r) \cdot X_1] \\
			&\quad = \theta_r.
		\end{align*}
	
		\textit{Case 2: outcome models are correctly specified.} 
		Assume the estimator $\hat \mu_1(x_r,r\oplus e_1)$ converges to true $\mu_1(x_r,r\oplus e_1)$ while $\hat O_1(x_r,r)$ converges to  $\tilde O_1(x_r,r)\neq O_1(x_r,r)$.
		We have
		\begin{align*}
			&\E[I(R=r\oplus e_1)\cdot \tilde O_1(X_r,r) (X_1 - \mu_1(X_r, r\oplus e_1)) + I(R=r) \cdot \mu_1(X_r, r\oplus e_1)] \\
			&\quad = \E[\tilde O_1(X_r,r) \cdot I(R=r\oplus e_1)\cdot (X_1 - \mu_1(X_r, r\oplus e_1)) + I(R=r) \cdot \mu_1(X_r, r\oplus e_1)] \\
			&\quad = \E[\tilde O_1(X_r,r) \cdot I(R=r\oplus e_1)\cdot (X_1 - \mu_1(X_r, r\oplus e_1))] + \theta_r \\
			&\quad = \theta_r.
		\end{align*}
		Note that for every pattern $r$ such that $r_1=0$ and $r\neq 0_d$, we have two nuisance functions.  This leads to the $[2^{d-1}-1]$-multiple robustness.  Of particular interest, each pair of nuisance functions for a given $\theta_r$ is variation independent.
	\end{proof}

\subsection{Proof of Theorem~\ref{thm::MOORT}}
\begin{proof}
For simplicity, we assume that all variables are continuous. The case of categorical and discrete cases follow
similarly when we use the stochastic rank.

Consider an observation $(x_r, r) = (\bX_{i, \bR_i}, \bR_i)$, where
the distribution of $\bX_{i, \bR_i}$ given its response vector $ \bR_i$ is $p(x_{r}|r = \bR_i)$.  

Let $j \in \bR_i$ be the variable we choose in MOORT to perform the rank transformation in step 1. 
Namely, we mask $\bX_{ij}$, and pretending that our observation has a response vector $R = \bR_i\ominus e_j$. 
In this case, 
the  conditional PDF of $\bX_{ij}$ given all other observed variables is 
$$
p(x_j|X_{\bR_i\ominus e_j} = \bX_{i,\bR_i\ominus e_j}, R = \bR_i).
$$ 
Let $G_{\bX_{ij}}(t)$ denote the CDF of $p(x_j|X_{\bR_i\ominus e_j} = \bX_{i,\bR_i\ominus e_j}, R = \bR_i)$.

Now, for any imputation model $q\in \mathcal{Q}^*_{MOO}$, 
imputed value  $\hat \bX_{ij}$ is drawn from 
\begin{align*}
q(x_{j}|X_r = \bX_{i, \bR_i\ominus e_j}, R = \bR_i\ominus e_j) &= q(x_{j}|X_{\bR_i\ominus e_j} = \bX_{i, \bR_i\ominus e_j}, R = \bR_i\ominus e_j\oplus e_j) \\
&= p(x_{j}|X_{\bR_i\ominus e_j} = \bX_{i, \bR_i\ominus e_j}, R = \bR_i),
\end{align*}
which coincides with the true conditional PDF of $\bX_{ij}$.

As the number of imputation $M\rightarrow \infty$,
the empirical CDF satisfies 
$$
\sup_t|\hat G_{\bX_{ij}}(t)- G_{\bX_{ij}}(t)|\overset{P}{\rightarrow} 0
$$ 
by the Glivenko-Cantelli theorem.
Thus, the rank $\hat S_{i} = \hat G_{\bX_{ij}}(\bX_{ij})$ satisfies
$$
\max_{i=1,\cdots, n}\left| \hat G_{\bX_{ij}}(\bX_{ij})-G_{\bX_{ij}}(\bX_{ij})\right| \overset{P}{\rightarrow}0 
$$
and each $G_{\bX_{ij}}(\bX_{ij})\equiv U_i \sim {\sf Uni}[0,1]$.


Since exactly one variable is selected per individual and individuals are independent, the collection
$
\hat S_{1},\cdots, \hat S_{n}
$
asymptotically behaves like IID random variables $U_1,\cdots, U_n\sim {\sf Uni}[0,1]$ as $M\rightarrow \infty$. 

Thus, as $M \to \infty$, the empirical distribution of $\hat S_{1},\cdots, \hat S_{n}$
converges under Kolmogorov distance to the empirical distribution of $U_1,\cdots, U_n$.
As $n\rightarrow\infty$, 
the empirical distribution of $U_1,\cdots, U_n$.
converges  to the uniform distribution on $[0,1]$.
So the result follows.

\end{proof}

\subsection{Proof of Theorem \ref{thm::MOOEN}}
\begin{proof}
The proof is a direct consequence of the definition of $\mathcal{Q}^*_{MOO}$ and the fact that the energy distance is a strictly proper scoring rule.

Recall that from Algorithm \ref{alg::MOOEN},
$$
L_{\sf EN}(q|\bX_{ij}) = \frac{1}{M}\sum_{m=1}^M \left|\bX_{ij} - \hat \bX_{ij}^{(m)}\right| - \frac{1}{2M (M-1)} \sum_{m<m'}\left|\hat \bX_{ij}^{(m)} - \hat \bX_{ij}^{\dagger(m')}\right| 
$$
and
the final output is
$$
\hat{\mathcal{R}}_{\sf EN}(q) = \frac{1}{n}\sum_{i=1}^n \sum_{j\in \bR_i} L_{\sf EN}(q|\bX_{ij}).
$$

As $M\rightarrow\infty$, it is clear that for each $\bX_{ij}$,
\begin{align*}
L_{\sf EN}(q|\bX_{ij})&\overset{P}{\rightarrow} \E_{\hat \bX_{ij}}\left|\bX_{ij} - \hat \bX_{ij}\right| - \frac{1}{2}\E\left| \hat \bX_{ij} -  \hat \bX_{ij}^\dagger\right|\\
& = -{\sf ES} \left(\bX_{ij}, Q_{j, \bX_{i, \bR_i}, \bR_i}\right)
\end{align*}
where $\E_{\hat \bX_{ij}}$ refers to expectation only to random variable $\hat \bX_{ij}$
and 
$
Q_{j, \bX_{i, \bR_i}, \bR_i}
$ 
is the distribution function corresponding to the imputation model $q(x_j | X_{\bR_i\ominus e_j} = \bX_{i, \bR_i\ominus e_j}, R=\bR_i \ominus e_j)$.
The above convergences is for the Monte Carlo errors.


By the definition of $\mathcal{Q}^*_{MOO}$ (equation \eqref{eq::MOO::opt::set}), if $q \in \mathcal{Q}^*_{MOO}$, its imputation distribution is 
\begin{align*}
q(x_j|X_{\bR_i\ominus e_j} = \bX_{i, \bR_i\ominus e_j}, R = \bR_i\ominus e_j) &\overset{\eqref{eq::MOO::opt::set}}{=} p(x_j|X_{\bR_i\ominus e_j} = \bX_{i, \bR_i\ominus e_j}, R = (\bR_i\ominus e_j)\oplus e_j)\\
& = p(x_j|X_{\bR_i\ominus e_j} = \bX_{i, \bR_i\ominus e_j}, R = \bR_i),
\end{align*}
which equals $p(x_j|X_{\bR_i\ominus e_j} = \bX_{i, \bR_i\ominus e_j}, R = \bR_i)$, 
the distribution of the masked value $\bX_{ij}$.
Therefore, each masked value $\bX_{ij}$ is  from the distribution $Q_{j, \bX_{i, \bR_i}, \bR_i}$
as well. 

This implies that the summation 
$$
\sum_{j\in \bR_i} L_{\sf EN}(q|\bX_{ij})\overset{P}{\rightarrow} \sum_{j\in \bR_i} -{\sf ES} \left(\bX_{ij}, Q_{j, \bX_{i, \bR_i}, \bR_i}\right)
$$
when $M\rightarrow\infty$ (Monte Carlo errors disappear)
and each $\bX_{ij}$ is from $Q_{j, \bX_{i, \bR_i}, \bR_i}$. 
Thus, 
we conclude that 
\begin{align*}
\hat{\mathcal{R}}_{\sf EN}(q) &= \frac{1}{n}\sum_{i=1}^n \sum_{j\in \bR_i} L_{\sf EN}(q|\bX_{ij})\\
& \overset{P}{\rightarrow}  \frac{1}{n}\sum_{i=1}^n\sum_{j\in \bR_i} -{\sf ES} \left(\bX_{ij}, Q_{j, \bX_{i, \bR_i}, \bR_i}\right)\qquad \mbox{(as $M\rightarrow\infty$)}\\
& \overset{P}{\rightarrow}  \E\left(\sum_{j\in \bR_i} -{\sf ES} \left(\bX_{ij}, Q_{j, \bX_{i, \bR_i}, \bR_i}\right)\right)\qquad \mbox{(as $n\rightarrow\infty$)}\\
& = 0
\end{align*}
since each $\bX_{ij}$  is from $Q_{j, \bX_{i, \bR_i}, \bR_i}$
so the expectation is $0$ due to the energy score ${\sf ES}$ being a strictly proper scoring rule.

%
%
%
%
%
%

\end{proof}

\subsection{Proof of Theorem~\ref{thm::MLE}}

Before we proceed, we want to note again that 
we have the following notational conventions:
\begin{equation}
\begin{aligned}
 q(x_j = \bX_{ij}|x_{r} = \bX_{i, r}, r = \bR_i\ominus e_j)
& \equiv  q(x_j = \bX_{ij}|X_{r} = \bX_{i, r}, R = \bR_i\ominus e_j)\\
& \equiv  q(x_j = \bX_{ij}|\bX_{i, \bR_i\ominus e_j}, \bR_i\ominus e_j).
\end{aligned}
\label{eq::convention}
\end{equation}
The first expression, $ q(x_j = \bX_{ij}|x_{r} = \bX_{i, r}, r = \bR_i\ominus e_j)$,
makes the proofs of likelihood inference a lot cleaner since the expectation operator is clearly only applies to the random variable $(\bX_i, \bR_i)$.

\begin{proof}

We will derive the asymptotic normality using the standard procedure for Z-estimation (score equation). 
Recall the score function in equation \eqref{eq::MOO::likelihood3} is
\begin{equation*}
S_n(\theta) = \nabla\ell_n(\theta)   = \frac{1}{n}\sum_{i=1}^n\sum_{j \in \bR_i} \nabla_\theta \log q_\theta(x_j = \bX_{ij}|x_{r} = \bX_{i, r}, r = \bR_i\ominus e_j),
\end{equation*}
which implies the population score 
$$
\bar S(\theta) = \nabla \bar \ell(\theta) = \E\left(\sum_{j \in \bR_1} \nabla_\theta \log q_\theta(x_j = \bX_{1j}|x_{r} = \bX_{1, r}, r = \bR_1\ominus e_j) \right).
$$
In $S_n(\theta)$, we additionally divide it by $n$
to make it align with the population score. 
This will not influence the location of the maximizer (MLE).

By assumption (A1), the MLE and population MLE both solve the score equation, namely,
\begin{align*}
S_n(\hat\theta_n) = 0 ,\qquad \bar S(\theta^*) = 0.
\end{align*}
Thus, using the Taylor expansion, we have the following decomposition:
\begin{align*}
S_n(\theta^*) - \bar S(\theta^*) & = S_n(\theta^*) - S_n(\hat \theta_n)\\
& = - (\hat \theta_n - \theta^*) \nabla S_n(\theta^*) + o_P(\|\hat \theta_n - \theta^*\|).
\end{align*}
Note that we need the uniformly bounded  second-order derivatives on $\bar S(\theta)$
condition in (A3) to ensure the remainder terms is $o_P(\|\hat \theta_n - \theta^*\|)$.

By rearrangements, 
$$
\hat \theta_n - \theta^* = -  \nabla S^{-1}_n(\theta^*) (S_n(\theta^*) - \bar S(\theta^*)) + o_P(\|\hat \theta_n - \theta^*\|). 
$$
Note that the term 
\begin{equation}
\begin{aligned}
S_n(\theta^*) - \bar S(\theta^*)  &= \frac{1}{n}\sum_{i=1}^n \Gamma(\theta| \bX_i, \bR_i) - \E[\Gamma(\theta| \bX_i, \bR_i)],\\
  \Gamma(\theta| \bX_i, \bR_i) &= \sum_{j \in \bR_i} \nabla_\theta \log q_\theta(x_j = \bX_{ij}|x_{r} = \bX_{i, r}, r = \bR_i\ominus e_j)\\
  & = \nabla_\theta  \ell(\theta|\bX_{i,\bR_i},\bR_i),
\end{aligned}
\label{eq::MLE::G}
\end{equation}
are summation of IID random variables, so it has asymptotic normality by the central limit theorem:
\begin{equation}
\sqrt{n}(S_n(\theta^*) - \bar S(\theta^*)) \overset{d}{\rightarrow}N (0, \E\left[\Gamma(\theta^*| \bX_i, \bR_i) \Gamma(\theta^*| \bX_i, \bR_i)^T \right ]). 
\label{eq::MLE::N}
\end{equation}
So all we need is to control the inverse matrix $ \nabla S^{-1}_n(\theta^*)$.

Using the fact that each element of the matrix
$$
\nabla S_n(\theta^*) = \frac{1}{n}\sum_{i=1}^n\sum_{j \in \bR_i} \nabla_\theta \nabla_\theta \log q_\theta(x_j = \bX_{ij}|x_{r} = \bX_{i, r}, r = \bR_i\ominus e_j)
$$
is a sample average of i.i.d. random variables, by Assumption (A2) and the law of large numbers, 
we have 
$$
\nabla S_n(\theta^*) \overset{P}{\rightarrow} \nabla \bar S(\theta^*)  = \bar H(\theta^*) = \E\left\{\sum_{j \in \bR_i} \nabla_\theta \nabla_\theta \log q_\theta(x_j = \bX_{ij}|x_{r} = \bX_{i, r}, r = \bR_i\ominus e_j)\right\}.
$$
Again, (A2) require the Hessian matrix $ \bar H(\theta^*)$ to be invertible, so we have 
$$
\nabla S^{-1}_n(\theta^*)\overset{P}{\rightarrow} \nabla \bar S^{-1}(\theta^*) = \bar H^{-1}(\theta^*)
$$
by continuous mapping theorem.
Combing this with equation \eqref{eq::MLE::N} and applying the Slutsky theorem, we conclude that 
\begin{align*}
\sqrt{n}(\hat \theta_n - \theta^*) &= -  \sqrt{n}\nabla S^{-1}_n(\theta^*) (S_n(\theta^*) - \bar S(\theta^*)) + o_P(\sqrt{n}\|\hat \theta_n - \theta^*\|)\\
&\overset{d}{\rightarrow}  N (0, \bar H^{-1}(\theta^*)\E\left[\Gamma(\theta^*| \bX_i, \bR_i) \Gamma(\theta^*| \bX_i, \bR_i)^T \right ]\bar H^{-1}(\theta^*)),
\end{align*}
which completes the proof. 
\end{proof}

\subsection{Proof of Theorem~\ref{thm::GD}}	\label{sec::thm::GD}

Before we prove Theorem~\ref{thm::GD}, we first introduce a useful lemma on the uniform convergence.
\begin{lemma}
Under assumptions (A1) and (A3), we have the following uniform bounds:
\begin{align*}
\sup_{\theta \in \Theta} \left|\frac{1}{n}\ell_n(\theta) - \bar \ell(\theta)\right| &\overset{P}{\rightarrow}0\\
\sup_{\theta \in \Theta} \left\|\frac{1}{n}\nabla\ell_n(\theta) - \nabla\bar \ell(\theta)\right\|_{\max} &\overset{P}{\rightarrow}0\\
\sup_{\theta \in \Theta} \left\|\frac{1}{n}\nabla\nabla\ell_n(\theta) - \nabla\nabla\bar \ell(\theta)\right\|_{\max} &\overset{P}{\rightarrow}0.
\end{align*}
\label{lem::ERM}
\end{lemma}

\begin{proof}
The proof is an application of Example 19.7 of \cite{van2000asymptotic}. 
The key is to observe that 
$$
\frac{1}{n}\ell_n(\theta)  = \frac{1}{n}\sum_{i=1}^n \ell(\theta|\bX_{i,\bR_i},\bR_i)
$$
is an empirical average. 
So we can use empirical process theory to obtain this bound.

Under (A1) and (A3), there exists $\Lambda_1(X_R,R)$ such that 
$$
\sup_{\theta \in \Theta}\|\nabla_\theta \ell(\theta|X_R,R)\|_{\max} \leq \Lambda_1(X_R,R)
$$
and $\E[|\Lambda_1(X_R,R)|]<\infty$.
Thus, 
the collection of function 
$$
\mathcal{L} = \{ \ell(\theta|x_r,r): \theta\in\Theta\}
$$
has an $\epsilon-$bracketing number shrinking at rate $O(\epsilon^{-{\sf dim}(\Theta)})$. 
So the collection $\mathcal{L}$ forms a Glivenko-Cantelli class and by Theorem 19.4 of \cite{van2000asymptotic},
we have 
$$
\sup_{\theta \in \Theta} \left|\frac{1}{n}\ell_n(\theta) - \bar \ell(\theta)\right| \overset{P}{\rightarrow}0.
$$

The case of gradient and Hessian convergence follows similarly; we just focus on each element 
and use the fact that the parameter space is compact (from (A1)) and the third-order derivative is integrable (from (A3)).

\end{proof}


With Lemma \ref{lem::ERM}, we are able to prove Theorem \ref{thm::GD}.

\begin{proof}
Our proof consists of three parts.
In Part 1, we will show that regions around the population MOO log-likelihood is strongly concave (all eigenvalues of the Hessian matrix
are negative).
In Part 2, we will extend the result of Part 1 to sample log-likelihood.
We will utilize the uniform bounds in 
Lemma \ref{lem::ERM} in this part.
Part 3 is the analysis on the algorithmic convergence.
A technical challenge here is that our objective function is the sample MOO log-likelihood,
which is a random quantity while our assumptions (A1-3) are on the population MOO log-likelihood.
So we have to use uniform bounds to convert concavity of the population MOO log-likelihood to
the sample MOO log-likelihood.


{\bf Part 1: local concavity of the population log-likelihood.}
Assumption (A2) requires that the Hessian matrix $\bar H(\theta) = \nabla \nabla \bar \ell(\theta)$
is invertible at $\theta =\theta^*$ and $\theta^*$ is an interior point in the parameter space by (A1).
Assumption (A3) further requires that the third-order derivatives of $\bar \ell(\theta)$
is uniformly bounded.
Thus, there exists a constant $\zeta_1>0$ such that all eigenvalues of 
$\nabla\nabla\bar \ell(\theta)$ is 
negative for $\theta \in B(\theta^*, \zeta_1)$. 

Specifically, we  construct $\zeta_1$  as follows. The uniform third-order derivative implies
that the Hessian $ \bar H(\theta)$ is Lipschitz in the sense that 
$$
\| \bar H(\theta_1) - \bar H(\theta_2)\|_{2} \leq \psi_3 \|\theta_1-\theta_2\|,
$$
where $\psi_3 = \sup_{\theta\in\Theta} \max_{j_1,j_2,j_3}\left|\frac{\partial}{\partial \theta_{j_1}}\frac{\partial}{\partial \theta_{j_2}}\frac{\partial}{\partial \theta_{j_3}}\bar \ell(\theta)\right|<\infty$
by assumption (A3).

Let $\lambda^*_{\max}<0$ be the largest eigenvalue of $ \bar H(\theta^*)$. 
We pick 
$$
\zeta_1 = \frac{-\lambda^*_{\max}}{3 \psi_3}
$$
so that by Weyl's theorem (see, e.g., Chapter 4 of \citealt{horn2012matrix}), 
$$
|\lambda_{\max}(\bar H(\theta)) - \lambda^*_{\max}|\leq \| \bar H(\theta) - \bar H(\theta^*)\|_2 \leq \psi_3 \|\theta-\theta^*\|
$$
and we conclude that 
\begin{equation}
\lambda_{\max}(\bar H(\theta)) \leq \lambda^*_{\max} + \psi_3 \|\theta-\theta^*\| \leq \frac{2}{3} \lambda^*_{\max}<0 
\label{eq::GA::pf1}
\end{equation}
for any $\theta \in B(\theta^*, \zeta_1)$.
As a result, all eigenvalues of $\bar H(\theta)$ are negative when $\theta \in B(\theta^*, \zeta_1)$,
so $\bar \ell(\theta)$ is strongly concave within $B(\theta^*, \zeta_1)$.

{\bf Part 2: local concavity of the sample log-likelihood.}
To convert the results on population log-likelihood to sample log-likelihood,
we need to use uniform bounds. 
Let $\bar \ell_n(\theta) = \frac{1}{n}\ell_n(\theta)$ be the normalized log-likelihood function.
Clearly, the maximizer of $\bar \ell_n(\theta)$ is the same as $ \ell_n(\theta)$,
so we will focus on analyzing the gradient ascent on $\bar \ell_n(\theta)$ 
since this quantity has a  limiting behavior easier to analyze.

Let $\bar H_n(\theta)  = \nabla\nabla \bar \ell_n(\theta)$ be the Hessian matrix. 
Under assumption (A3) and by Lemma \ref{lem::ERM},
$$
\sup_{\theta\in\Theta} \left|\bar H_n(\theta) - \bar H(\theta)\right|\overset{P}{\rightarrow}0.
$$
Denote the event 
$$
E_{1,n} = \left\{\sup_{\theta\in\Theta} \left|\bar H_n(\theta) - \bar H(\theta)\right|\leq -\frac{1}{3}\lambda^*_{\max}\right\}.
$$
Note that $\lambda^*_{\max}<0$ so $-\frac{1}{3}\lambda^*_{\max}$ is a positive number.
Clearly, $P(E^C_{1,n})\rightarrow 0$. 
Under event $E_{1,n}$, for any point $\theta \in B(\theta^*, \zeta_1)$, 
the maximal eigenvalues $\lambda_{\max}(\bar H_n(\theta))$
satisfies
\begin{align*}
\lambda_{\max}(\bar H_n(\theta)) &\leq  \lambda_{\max}(\bar H(\theta)) + \sup_{\theta\in\Theta}\left|\lambda_{\max}(\bar H_n(\theta))-\lambda_{\max}(\bar H(\theta))\right|\\
&\leq  \lambda_{\max}(\bar H(\theta)) - \frac{1}{3}\lambda^*_{\max}\\
& \overset{\eqref{eq::GA::pf1}}{\leq} \frac{1}{3}\lambda^*_{\max}.
\end{align*}
Note that we use the Weyl's inequality again in the first inequality.

Now consider the event 
$$
E_{2,n} = \left\{\|\hat \theta_n - \theta^*\|\leq \frac{1}{2}\zeta_1 \right\}. 
$$
Since we know that $\|\hat \theta_n - \theta^*\|\overset{P}{\rightarrow}0$, 
$P(E^C_{2,n})\rightarrow 0$. 
Under the event $E_{2,n}$, the ball $B\left(\hat\theta_n, \frac{1}{2}\zeta_1\right)\subset B(\theta^*, \zeta_1)$.
Thus, we choose 
$$
\zeta_0 = \frac{1}{2}\zeta_1 =\frac{-\lambda^*_{\max}}{6 \psi_3}.
$$
Under such choice of $\zeta_0$ and when events $E_{1,n},E_{2,n}$ occur, we have $B\left(\hat\theta_n,\zeta_0\right)\subset B(\theta^*, \zeta_1)$,
so that the maximal eigenvalues 
$$
\sup_{\theta \in B\left(\hat\theta_n,\zeta_0\right)} \lambda_{\max}(\bar H_n(\theta)) \leq \frac{1}{3}\lambda^*_{\max} < 0.
$$
Thus, the function $\bar \ell_n(\theta)$ is strongly concave within $B\left(\hat\theta_n,\zeta_0\right)$
when both $E_{1,n} $ and $E_{2,n}$ occur, which 
has a probability greater than or equal to $1-P(E^C_{1,n}) - P(E^C_{2,n})\rightarrow 1$.

{\bf Part 3: algorithmic convergence.}
Part of this proof is
from standard analysis in convex optimization \citep{boyd2004convex}. 
The technical difficulty is that all our assumptions (A1-3)
are on the population log-likelihood but the gradient ascent algorithm is applied to a
sample log-likelihood.
So the key is to control the smoothness of the sample log-likelihood.

Recall that our gradient ascent is 
$$
\theta^{(t+1)} = \theta^{(t)} + \xi \frac{1}{n}S_n (\theta^{(t)}) = \theta^{(t)} + \xi \bar S_n (\theta^{(t)}), 
$$
where $\bar S_n(\theta) = \frac{1}{n}S_n(\theta) = \frac{1}{n}\nabla \ell_n(\theta) = \nabla \bar\ell_n(\theta)$.

For $\theta \in B\left(\hat\theta_n,\zeta_0\right)$, the analysis in Part 2 shows that 
$\bar \ell_n(\theta)$ is $M^*$-strongly concave with $M^* = \frac{1}{3}\lambda^*_{\max}<0$. So we have 
$$
\bar \ell_n(\hat \theta) - \bar \ell_n(\theta^{(t)}) \leq (\hat \theta_n - \theta^{(t)})^T \bar S_n(\theta^{(t)}) + \frac{M^*}{2}\|\hat \theta_n - \theta^{(t)}\|^2,
$$
which implies
\begin{equation}
2(\hat \theta_n - \theta^{(t)})^T \bar S_n(\theta^{(t)}) \geq 2(\bar \ell_n(\hat \theta) - \bar \ell_n(\theta^{(t)})) - M^*\|\hat \theta_n - \theta^{(t)}\|^2.
\label{eq::GA::pf3}
\end{equation}

On the other hand, 
the score function is smooth 
in the sense that 
\begin{equation}
\|\bar S_n(\theta_1) - \bar S_n(\theta_2)\| \leq 2 H_{\max} \|\theta_1-\theta_2\|,
\label{eq::GA::pf4}
\end{equation}
where $H_{\max} = \sup_{\theta} \|\bar H(\theta)\|_2$ is the maximal spectral norm ($2$-norm) of the population Hessian matrix. 
Equation \eqref{eq::GA::pf4} follows from the fact that the spectral norm
$$
\|\bar S_n(\theta_1) - \bar S_n(\theta_2)\| \leq \sup_{\theta} \|\bar H_n(\theta)\|_2 \|\theta_1-\theta_2\|
$$
and the uniform convergence of the Hessian in Lemma \ref{lem::ERM}  implies that 
$$
\sup_{\theta}| \bar H_n(\theta) - \bar H(\theta)|\overset{P}{\rightarrow } 0 ,
$$
so we have 
$$
P\left(\sup_{\theta} \|\bar H_n(\theta)\|_2  \leq 2 \sup_{\theta}\|\bar H(\theta)\|_2\right) \rightarrow 1.
$$
Let event $E_{3,n}$ be
$$
E_{3,n} = \left\{\sup_{\theta} \|\bar H_n(\theta)\|_2  \leq 2 \sup_{\theta}\|\bar H(\theta)\|_2\right\}.
$$ 
Thus, under event $E_{3,n}$, equation \eqref{eq::GA::pf4} holds.
In the language of convex optimization \citep{boyd2004convex}, equation \eqref{eq::GA::pf4} can be interpreted as the function $\ell_n$ is $L^*$-smooth with $L^* = 2 H_{\max}$.

Note that for an $L^*$-smooth function $f$, i.e., $\|\nabla f(x) - \nabla f(y)\|\leq L^* \|x-y\|$, we have 
\begin{equation}
\begin{aligned}
f(x) - f(y) &\leq (x-y)^T \nabla f(y) + \frac{L^*}{2}\|x-y\|^2\\
f(y) - f(x) &\geq (y-x)^T \nabla f(y) - \frac{L^*}{2}\|x-y\|^2.
\end{aligned}
\label{eq::GA::pf5}
\end{equation}

When $\bar \ell_n(\theta)$ is $L^*$-smooth, we have the following inequalities:
\begin{align*}
\bar \ell_n(\hat\theta_n) - \bar \ell_n\left(\theta + \frac{1}{L^*}\bar S_n(\theta)\right)&\geq 0\\
\Rightarrow \bar \ell_n(\hat\theta_n) - \bar \ell_n\left(\theta + \frac{1}{L^*}\bar S_n(\theta)\right) - \bar \ell_n(\theta) &\geq- \bar \ell_n\left(\theta \right)\\
\Rightarrow  \bar \ell_n(\hat\theta_n)  - \bar \ell_n(\theta) &\geq\bar \ell_n\left(\theta + \frac{1}{L^*}\bar S_n(\theta)\right)- \bar \ell_n\left(\theta \right)\\
\end{align*}
Now applying equation \eqref{eq::GA::pf5} to the last inequality with $y= \theta+\frac{1}{L^*}\bar S_n(\theta)$ and $x = \theta$ and $f = \bar \ell_n$,
we obtain the inequality
\begin{align*}
\bar \ell_n\left(\theta +\frac{1}{L^*}\bar S_n(\theta)\right)- \bar \ell_n\left(\theta \right) &\geq \frac{1}{L^*} \|\bar S_n(\theta)\|^2 - \frac{L^*}{2} \|\frac{1}{L^*}\bar S_n(\theta)\|\\
& = \frac{1}{2L^*} \|\bar S_n(\theta)\|^2.
\end{align*}
Therefore, we conclude that 
\begin{equation}
\bar \ell_n(\hat\theta_n)  - \bar \ell_n(\theta) \geq\bar \ell_n\left(\theta + \frac{1}{L^*}\bar S_n(\theta)\right)- \bar \ell_n\left(\theta \right)\geq \frac{1}{2L^*} \|\bar S_n(\theta)\|^2.
\label{eq::GA::pf6}
\end{equation}


Now we formally derive the convergence of the gradient ascent. 
Recall that $\theta^{(t+1)} = \theta^{(t)} + \xi \bar S_n(\theta^{(t)})$. 
\begin{equation}
\begin{aligned}
\|\theta^{(t+1)} - \hat \theta_n\|^2 &=  \left\|\theta^{(t)} + \xi\bar S_n(\theta^{(t)}) - \hat \theta_n\right\|^2\\
& = \|\theta^{(t)} - \hat \theta_n\|^2 - 2\xi  (\hat \theta_n-\theta^{(t)} )^T \bar S_n(\theta^{(t)}) + \xi^2 \left\|\bar S_n(\theta^{(t)})\right\|^2\\
&\overset{\eqref{eq::GA::pf3}}{\leq }\|\theta^{(t)} - \hat \theta_n\|^2 - \xi\left[2(\bar \ell_n(\hat \theta) - \bar \ell_n(\theta^{(t)})) - M^*\|\hat \theta_n - \theta^{(t)}\|^2\right] + \xi^2 \left\|\bar S_n(\theta^{(t)})\right\|^2\\
& = (1+M^*\xi) \|\theta^{(t)} - \hat \theta_n\|^2  -\xi \left[2(\bar \ell_n(\hat \theta) - \bar \ell_n(\theta^{(t)})) - \xi \left\|\bar S_n(\theta^{(t)})\right\|^2\right]\\
&\overset{\eqref{eq::GA::pf6}}{\leq }
(1+M^*\xi) \|\theta^{(t)} - \hat \theta_n\|^2  -\xi \left[\frac{1}{L^*} \|\bar S_n(\theta^{(t)})\|^2- \xi \left\|\bar S_n(\theta^{(t)})\right\|^2\right]\\
& = (1+M^*\xi) \|\theta^{(t)} - \hat \theta_n\|^2 - \frac{\xi}{L^*} (1- \xi L^*) \|\bar S_n(\theta^{(t)})\|^2.
\end{aligned}
\label{eq::GA::pf7}
\end{equation}
Note that $M^* = \frac{1}{3}\lambda^*_{\max}<0$ and $L^* = 2 H_{\max}>0$,
so when the step size 
$$
\xi < \min\left\{\frac{-1}{M^*}, \frac{1}{L^*}\right\}, 
$$
equation \eqref{eq::GA::pf7} becomes
\begin{align*}
\|\theta^{(t+1)} - \hat \theta_n\|^2 &\leq (1+M^*\xi) \|\theta^{(t)} - \hat \theta_n\|^2 - \frac{\xi}{L^*} (1- \xi L^*) \|\bar S_n(\theta^{(t)})\|^2\\
&\leq 
(1+M^*\xi) \|\theta^{(t)} - \hat \theta_n\|^2 
\end{align*}
and by telescoping, 
$$
\|\theta^{(t)} - \hat \theta_n\| \leq (1+M^*\xi)^{t/2} \|\theta^{(0)} - \hat \theta_n\| = \rho^t  \|\theta^{(0)} - \hat \theta_n\|
$$
with $\rho =\sqrt{1+M^*\xi} = \sqrt{1+\frac{1}{3}\lambda^*_{\max}\xi}\in(0,1)$, which is the desire result.

Throughout the proof, we need events $E_{1,n}, E_{2,n}, E_{3,n}$ to hold
and choose the radius $\zeta_0 = \frac{-\lambda^*_{\max}}{6\psi_3} $ and stepsize $\xi < \xi_0 = \min\left\{\frac{-1}{M^*}, \frac{1}{L^*}\right\} = \min\left\{\frac{-3}{\lambda^*_{\max}}, \frac{1}{2H_{\max}}\right\}$. 
Thus, the above algorithmic convergence holds with a probability 
\begin{align*}
P(E_{1,n}\cap E_{2,n}\cap E_{3,n}) \geq 1- P(E_{1,n}) - P(E_{2,n}) - P(E_{3,n}) \rightarrow 1.
\end{align*}

\end{proof}

\subsection{Proof of Theorem~\ref{thm::MOO::ID}}

We will utilize the following reparameterization method in equation \eqref{eq::sep2} for our proof. 
\begin{lemma}[reparameterization method]
Let $s\in\{0,1\}^d$ be a response vector and $s \neq 0_d$
and $q$ is an imputation model.
Then we have
\begin{align*}
\sum_{j \in s} \log q(x_j &|x_{s\ominus e_j} , R = s\ominus e_j)
=\sum_{r:r\neq 1_d}    \sum_{j \in \bar r} I(s = r\oplus e_j)\log q(x_j |x_{r} , r) 
\end{align*}
\label{lem::repar}
\end{lemma}

\begin{proof}

A key to this proof is that the following two conditions are equivalent:
\begin{equation}
\{(r,s,j): s_j = 1, r=s\ominus e_j\} \equiv \{(r,s,j):r\neq 1_d, s = r\oplus e_j, j \in \bar r\}.
\label{eq::repar1}
\end{equation}
Equation \eqref{eq::repar1} implies that 
$I(s_j=1, r=s\ominus e_j) = I(s = r\oplus e_j, j\in \bar r)$. 
Thus,
\begin{align*}
\sum_{j \in s} \log q(x_j &|x_{s\ominus e_j} , R = s\ominus e_j)\\
& = \sum_{j } I(s_j = 1)\log q(x_j |x_{s\ominus e_j} , R = s\ominus e_j)\\
& = \sum_{r:r\neq 1_d} I(r=s\ominus e_j)\sum_{j } I(s_j = 1)\log q(x_j |x_{s\ominus e_j} , R = s\ominus e_j)\\
& = \sum_{r:r\neq 1_d}\sum_{j} I(r=s\ominus e_j, s_j = 1)\log q(x_j |x_{r} , R = r)\\
& = \sum_{r:r\neq 1_d}\sum_{j} I(s = r\oplus e_j, j\in \bar r)\log q(x_j |x_{r} , R = r)\\
& = \sum_{r:r\neq 1_d}    \sum_{j \in \bar r} I(s = r\oplus e_j)\log q(x_j |x_{r} , r) .
\end{align*}


\end{proof}

With the above reparameterization Lemma, we can formally prove Theorem~\ref{thm::MOO::ID}. 

\begin{proof}
By Lemma \ref{lem::repar}, we have
\begin{align*}
\sum_{j \in \bR_i} \log q_\theta(x_j = \bX_{ij}&|x_{r} = \bX_{i, r}, r = \bR_i\ominus e_j)\\
& = \sum_{r:r\neq 1_d}    \sum_{j \in \bar r} I(\bR_i = r\oplus e_j)\log q_\theta(x_j = \bX_{ij}|x_{r} = \bX_{i,r}, r) .
\end{align*}
The power of the above equality is that 
in the first summation, we have a random indices $j\in\bR_i$
while in the second summation, there is no randomness in the summation $\sum_{r:r\neq 1_d}    \sum_{j \in \bar r}$.

Using the reparameterization method,
we can rewrite the MOO log-likelihood as
\begin{align*}
\bar \ell(q) & = \E\left\{\sum_{j \in \bR_1} \log q(x_j = \bX_{1j}|x_{r} = \bX_{1, r}, r = \bR_1\ominus e_j)\right\}\\
&= \E\left\{ \sum_{r:r\neq 1_d}  \sum_{j \in \bar r} I(\bR_1 = r\oplus e_j)\log q(x_j = \bX_{1j}|x_{r} = \bX_{1,r}, r) \right\}\\
& = \sum_{r:r\neq 1_d} \sum_{j\in \bar r}  \int p(x_{j}, x_r, r\oplus e_j)\log q(x_j |x_{r}, r) dx_j dx_r\\
& = \sum_{r:r\neq 1_d} \sum_{j\in \bar r}\int  \underbrace{\left[\int p(x_{j}| x_r, r\oplus e_j)\log q(x_j |x_{r}, r) dx_j\right]}_{= (A)}  p(x_r, r\oplus e_j)dx_r.
\end{align*}

By definition of $\mathcal{Q}^*_{MOO}$ in equation \eqref{eq::MOO::opt::set}, 
any imputation model $q\in \mathcal{Q}^*_{MOO}$ satisfies
$$
q(x_j|x_r, r) = p(x_j|x_r, r\oplus e_j)
$$
so the quantity (A) in the above equality 
is maximized since it is the cross-entropy.
This holds for every $r$ and every $j\in \bar r$.
As a result, for any $q\in \mathcal{Q}^*_{MOO}$, $\bar \ell(q)$ is maximized, which completes the proof.
\end{proof}

\subsection{Proof of Theorem~\ref{thm::mcar}}
\begin{proof}
Under MCAR and use the fact that the true data are generated from $p(x)=f_{\theta^*}(x)$,
the imputation model under $f_{\theta^*}$ for pattern $(x_r,r)$ is
$f_{\theta^*}(x_{j}|x_r)$
for any $j \in \bar r$.

By the definition of MCAR, $R\perp X$, so 
\begin{align*}
p(x_j|x_r, R= r\oplus e_j)  = p(x_j|x_r)
 = f_{\theta^*}(x_j|x_r).
\end{align*}
So the imputation model $f_{\theta^*}\in\mathcal{Q}_{MOO^*}$.
Thus, by Theorem \ref{thm::MOO::ID}, $\bar \ell(f_{\theta^*}) = \sup_q \bar \ell(q)$.


\end{proof}

\subsection{Proof of Theorem~\ref{thm::BIC}}

%
%
%


\begin{proof}[Proof of Theorem~\ref{thm::BIC}]

Here is the overview of the proof. 
For models $k=1,\cdots, k^*-1$ (models of lower order), we will show that 
the BIC 
$$
 \ell_{n, BIC}(q_k) -  \ell_{n, BIC}(q_{k^*}) \approx -c \cdot n
$$
for some constant $c>0$.
So asymptotically, their BIC values will be lower than $q_{k^*}$. 

For models $k = k^*+1,\cdots, K$ (models of higher order),
since the optimal model is in $\mathcal{Q}_{k^*}$, 
all these models contain the optimal model. 
So the log-likelihood value of $ \ell_{n}(q_k)\approx\ell_{n}(q_{k^*})$
so the BIC value 
$$
\ell_{n, BIC}(q_k) - \ell_{n, BIC}(q_{k^*}) \approx -(d_{k} - d_{k^*}) \log n.
$$
Therefore, the BIC values are also lower than $q_{k^*}$.

{\bf Model of lower order ($k<k^*$).}
For model $\mathcal{Q}_k$, 
let $\theta^*_{[k]}$ be the its population MLE, i.e.,
$$
\theta^*_{[k]} = {\sf argmax}_{\theta_{[k]}\in\Theta_{[k]}} \E\left[ \ell_{n}(q_{\theta_{[k]}})\right]
$$
and recall that 
$$
\hat \theta_{[k]} = {\sf argmax}_{\theta_{[k]}\in\Theta_{[k]}}  \ell_{n}(q_{\theta_{[k]}})
$$
is the sample MLE such that $q_k = q_{\hat \theta_{[k]}}$.

Note that 
\begin{equation}
\begin{aligned}
\E\left[ \ell_{n}(q_{\theta_{[k]}})\right] & =n \cdot \E\left[ \sum_{j \in \bR_i} \log q_{\theta_{[k]}}(x_j = \bX_{ij}|x_{r} = \bX_{i, r}, r = \bR_i\ominus e_j)\right]
& = n \bar \ell(q_{\theta_{[k]}}).
\end{aligned}
\label{eq::BIC::L1}
\end{equation}
By condition (B1), the optimal model does not appear for this $k$ 
and by condition (B2), the optimal parameter is $\theta^*_{[k^*]}$
so we conclude
$$
\Delta_k \equiv \bar \ell(q_{\theta^*_{[k^*]}}) -\bar  \ell(q_{\theta^*_{[k]}}) >0 . 
$$
Putting this back to equation \eqref{eq::BIC::L1}, we conclude 
\begin{equation}
\E\left[ \ell_{n}(q_{\theta^*_{[k]}})\right] - \E\left[ \ell_{n}(q_{\theta^*_{[k^*]}})\right]  = -n \cdot \Delta_k.
\label{eq::BIC::L2}
\end{equation}
Use the fact that $q_k = q_{\hat \theta_{[k]}}$ is the model under the MLE,
we have the following results:
\begin{align*}
0 &\leq \ell_{n}(q_k) - \ell_{n}(q_{\theta^*_{[k]}})\\
&\leq \ell_{n}(q_k) -\E\left[ \ell_{n}(q_{\theta^*_{[k]}})\right] + {\left|\underbrace{\E\left[ \ell_{n}(q_{\theta^*_{[k]}})\right]}_{= n \bar \ell(q_{\theta^*_{[k]}})} - \ell_{n}(q_{\theta^*_{[k]}})\right|}\\
&\Rightarrow \ell_{n}(q_k)- \E\left[ \ell_{n}(q_{\theta^*_{[k]}})\right]  \geq -\left|n \bar \ell(q_{\theta^*_{[k]}}) - \ell_{n}(q_{\theta^*_{[k]}})\right|
\end{align*}
On the other hand, since $q_{\theta^*_{[k]}}$ is the MLE of $\E\left[ \ell_{n}(q_{\theta_{[k]}})\right] = n \bar \ell(q_{\theta_{[k]}}) $,
\begin{align*}
0& \leq  n \bar \ell(q^*_{\theta_{[k]}}) - n\bar \ell(q_k)\\
& \leq  n \bar \ell(q^*_{\theta_{[k]}}) -  \ell_{n}(q_k) +\left| \ell_{n}(q_k) - n\bar \ell(q_k)\right|\\
&\Rightarrow\ell_{n}(q_k)  - \E\left[ \ell_{n}(q_{\theta^*_{[k]}})\right]  \leq \left| \ell_{n}(q_k) - n\bar \ell(q_k)\right|. 
\end{align*}
The uniform bound in  Lemma \ref{lem::ERM} (assumption (AS) allows us to apply it
to every $k$) implies that 
$$
\sup_{\theta_{[k]}\in\Theta_{[k]}}\left| \ell_{n}(q_{\theta_{[k]}}) - n\bar \ell(q_{\theta_{[k]}})\right| = o_P(n),
$$
so we conclude
\begin{equation}
\left|\ell_{n}(q_k) - \E\left[ \ell_{n}(q_{\theta^*_{[k]}})\right]\right| 
 = o_P\left(n\right)
\label{eq::BIC::L3}
\end{equation}
for any $k=1,\cdots, K$. Note that the above bound is the standard empirical risk minimization bound. 

Combining equations \eqref{eq::BIC::L2} and \eqref{eq::BIC::L3}, 
we conclude that 
\begin{align*}
\ell_n(q_k) - \ell_n(q_{k^*})  
& = \E\left[ \ell_{n}(q_{\theta^*_{[k]}})\right] - \E\left[ \ell_{n}(q_{\theta^*_{[k^*]}})\right] + o_P\left(n\right)\\
& = -n \cdot \Delta_k + o_P(n).
\end{align*}
Recall that BIC is $\ell_{n ,BIC}(q) = \ell_n(q) - \frac{1}{2} d(q) \log n$, so 
we conclude
\begin{align*}
\ell_{n, BIC}(q_k) - \ell_{n, BIC}(q_{k^*})  
& = -n \cdot \Delta_k + o_P(n) + \frac{1}{2} (d_k -d_{k^*} ) \log n \\
& = -n \cdot \Delta_k + o_P(n).
\end{align*}
Therefore, 
$$
P(\ell_{n, BIC}(q_{k^*}) >\ell_{n, BIC}(q_k)) \rightarrow 1,
$$
so we will not choose any $k<k^*$ with a probability tending to $1$. 

{\bf Model of higher order ($k>k^*$).}
We will show that the fluctuation of $\ell_n(q_k) - \ell_n(q_{k^*})$ 
will be of the order $O_P(1)$.
So the penalty term in the BIC $(d_k - d_{k^*})\log n$ will eventually dominate. 
Since the model is nested, by condition (B2) we have 
\begin{equation}
 \ell_n( q_{\theta^*_{[k]}}) =  \ell_n( q_{\theta^*_{[k^*]}})
 \label{eq::BIC::U1}
\end{equation}
for all $k>k^*$.

Using the fact that the MLE $\hat \theta_{[k]}$
solves the score equation 
$$
\nabla_\theta\ell_n(\hat \theta_{[k]}) = 0, 
$$
we can perform a Taylor expansion:
\begin{align*}
\ell_n(q_{ \theta^*_{[k]}}) - \ell_n(q_{\hat \theta_{[k]}}) 
& = (\theta^*_{[k]} - \hat \theta_{[k]})^T \nabla_\theta\nabla_\theta \ell_n(q_{\hat \theta_{[k]}}) (\theta^*_{[k]} - \hat \theta_{[k]})  +o_P(n \|\theta^*_{[k]} - \hat \theta_{[k]}\|^2)\\
& = \underbrace{\sqrt{n} (\theta^*_{[k]} - \hat \theta_{[k]})^T}_{=u_n^T} \underbrace{\nabla_\theta\nabla_\theta \frac{1}{n}\ell_n(q_{\hat \theta_{[k]}})}_{=\Omega_n} \sqrt{n}(\theta^*_{[k]} - \hat \theta_{[k]}) +o_P(n \|\theta^*_{[k]} - \hat \theta_{[k]}\|^2).
\end{align*}

The above result shows an asymptotic quadratic form of $u_n^T \Omega_n u_n^T$.
The quantity $\Omega_n$ will converges to a fixed matrix based on assumptions (A2) and (A3). 
The vector $u_n$ has asymptotic normality by Theorem~\ref{thm::MLE}. 
Therefore, we conclude that 
$$
\ell_n(q_{ \theta^*_{[k]}}) - \ell_n(q_{\hat \theta_{[k]}}) = O_P(1).
$$
Since this holds for every $k$, using the fact that $q_k = q_{\hat \theta_{[k]}}$, we conclude that 
\begin{align*}
\ell_n(q_k) - \ell_n(q_{k^*})  & =  \ell_n(q_{\hat \theta_{[k]}})  - \ell_n(q_{\hat \theta_{[k^*]}}) \\
& = \ell_n(q_{ \theta^*_{[k]}}) - \ell_n(q_{ \theta^*_{[k^*]}}) +O_P(1)\\
& = O_P(1).
\end{align*}
Thus, the BIC values will be 
\begin{align*}
\ell_{n,BIC}(q_k) - \ell_{n, BIC}(q_{k^*}) & = \ell_n(q_k) - \ell_n(q_{k^*}) -\frac{1}{2}(d_k-d_{k^*})\log n\\
& =  O_P(1) -\frac{1}{2}(d_k-d_{k^*})\log n.
\end{align*}
Thus, 
$$
P(\ell_{n,BIC}(q_k) - \ell_{n, BIC}(q_{k^*})<0)\rightarrow 1. 
$$
So we conclude that the chance of selecting $k^*$ is approaching $1$.

\end{proof}

\subsection{Proof of Proposition \ref{prop::acmv}}
\begin{proof}

The proof is immediate.
For the case of NCMV, 
the set $\mathcal{Q}^*_{MOOLC}$ in equation \eqref{eq::MOOLC::opt::set}
only has the constraint
$$
q(x_{t+1}|x_{\leq t}, t) = p(x_{t+1}|x_{\leq t},T=t+1) 
$$
for each $t$.
Clearly, the NCMV in equation \eqref{eq::acmv} satisfies this constraint with $\tau=t$. 
So $q_{NCMV} \in \mathcal{Q}^*_{MOOLC}$. 

For the case of ACMV, similarly the set $\mathcal{Q}^*_{MOOBL}$ in equation \eqref{eq::MOOBL::opt::set}
has the constraint 
$$
q(x_{t+1}|x_{\leq t}, t) = p(x_{t+1}|x_{\leq t},T\geq t+1)
$$
and the ACMV in equation \eqref{eq::acmv} satisfies it with $\tau = t$.
Thus, $q_{ACMV}\in \mathcal{Q}^*_{MOOBL}$,
which completes the proof.

\end{proof}

\subsection{Proof of Theorem~\ref{thm::mko}}

\begin{proof}
Similar to the proof of Theorem~\ref{thm::moo}, we first decompose the population risk:
\begin{align*}
\mathcal{E}_K(q) = \E\{L_K(q| \bX_{1,\bR_1},\bR_1)\} &= \sum_{r:r\neq 1_d}\int L_K(q|x_{r} , r) p(x_r,r)dx_r,\\
L_K(q|x_{r} , r) & = \sum_{\ell \in J_K(r)}\sum_{j \in \ell} \int L(x_j, x_j') q(x'_j|x_{r\ominus \ell},R=r\ominus \ell)dx'_j.
\end{align*}
Using the fact that for imputation model $q(x'_j|x_{r\ominus \ell},R=r\ominus \ell)$,
the only relevant variables in $x_r$ are $x_j$ and $x_{r\ominus \ell}$,
we can further decompose it as 
\begin{equation}
\begin{aligned}
\mathcal{E}(q) & = 
 \sum_{r:r\neq 1_d}\int\sum_{\ell \in J_K(r)}\sum_{j \in \ell} L(x_j, x_j') q(x'_j|x_{r\ominus \ell},R=r\ominus \ell)dx'_j p(x_r,r)dx_r,\\
&=  \sum_{r:r\neq 1_d} \sum_{\ell \in J_K(r)}\sum_{j \in \ell} \int L(x_j, x_j') q(x'_j|x_{r\ominus \ell},R=r\ominus \ell)dx'_j p(x_j, x_{r\ominus \ell},r)dx_j dx_{r\ominus \ell},\\
&=  \sum_{r:r\neq 1_d} \sum_{\ell \in J_K(r)}\sum_{j \in \ell} \int L(x_j, x_j') p(x_j, x_{r\ominus \ell},r) q(x'_j|x_{r\ominus \ell},R=r\ominus \ell)dx'_j dx_j dx_{r\ominus \ell}.
\end{aligned}
\label{eq::mko::pf1}
\end{equation}
The above summation starts with each observed pattern $r$ and then sums over all possible $K$-masking scenario
and over all possible variable $j$ that can be masked.
The pattern index for the imputation model is $r\ominus \ell$, which is changing from one masked variable to the other, making
the analysis complicated.

To analyze the imputation model,
we use the reparameterization technique in Lemma \ref{lem::repar} and Section \ref{sec::sep}
that changes the individual-view (summation over $i$ first)
to the imputation model-view (summation over $r$ first).
We consider
the imputation model $q(x_j|x_s, s)$ and  find out all possible cases in the above summation
such that this imputation model is used. 
For response pattern $R = s\oplus e_j$, this case will be used when we mask  $x_j$.
For $R = s\oplus e_j \oplus e_k$ with $k\notin s, k\neq j$, this case will be included when 
we mask two variables $x_j,x_k$. 
By induction, one can clearly see that 
the set
$$
\mathbb{U}_K(s, j)  = \{r\in\{0,1\}^d: r \geq s\oplus e_j, |r-s|\leq K\},
$$
contains all possible patterns that will use imputation model $q(x_j|x_s, s)$
during the mask-K-out process. 

As a result, 
equation \eqref{eq::mko::pf1} can be written as 
\begin{align*}
\mathcal{E}(q)&=  \sum_{r:r\neq 1_d} \sum_{\ell \in J_K(r)}\sum_{j \in \ell} \int L(x_j, x_j') p(x_j, x_{r\ominus \ell},r) q(x'_j|x_{r\ominus \ell},R=r\ominus \ell)dx'_j dx_j dx_{r\ominus \ell}\\
& = \sum_s\sum_{j \in  \bar s} \sum_{\omega\in \mathbb{U}_K(s, j)} \int L(x_j, x_j') p(x_j, x_{s},R=\omega) q(x'_j|x_{s},R=s)dx'_j dx_j dx_{s}\\
& = \sum_s\sum_{j \in  \bar s}  \int L(x_j, x_j') \left[\sum_{\omega\in \mathbb{U}_K(s, j)} p(x_j, x_{s},R=\omega)\right] q(x'_j|x_{s},R=s)dx'_j dx_j dx_{s}\\
& = \sum_s\sum_{j \in \bar s}  \int L(x_j, x_j') p(x_j, x_{s},R\in \mathbb{U}_K(s, j)) q(x'_j|x_{s},R=s)dx'_j dx_j dx_{s}\\
& = \sum_s\sum_{j \in  \bar s}  \underbrace{\left\{\int L(x_j, x_j') p(x_j| x_{s},R\in \mathbb{U}_K(s, j)) q(x'_j|x_{s},R=s)dx'_j dx_j\right\}}_{= \mathbf{L}_s(q)} p( x_{s},R\in \mathbb{U}_K(s, j))dx_{s}
\end{align*}
and clearly, 
if we choose $q$ so that it always impute 
$$
\hat x^*_j = {\sf argmin}_{\theta} \int L( x_j,\theta ) p(x_j|x_{s}, R\in \mathbb{U}_K(s, j))dx_j,
$$
then this imputation value minimizes $\mathbf{L}_s(q)$, so it minimizes $\mathcal{E}(q)$.

%
%
%
%
\end{proof}

\subsection{Proof of Theorem~\ref{thm::MOORT::v}}

\begin{proof}

By the same argument as the proof of Theorem~\ref{thm::MOORT},
$$
\hat{\mathcal{R}}_j(q)\overset{P}{\rightarrow} 0
$$
for each $j$ when $q\in  \mathcal{Q}^*_{MOO}$.
Thus, the summation 
$$
\bar{\mathcal{R}}(q) = \sum_{j=1}^d \hat{\mathcal{R}}_j(q) \overset{P}{\rightarrow}0 .
$$

\end{proof}

\section{Additional details for the simulation studies} \label{sec::appendix_sim}
In the simulation studies, the availability of the complete dataset enables external validation, as the ground truth for missing values is known. We use this characteristic to compute an \textit{oracle imputation risk} evaluated on the truly missing values, which serves as a reference measure of genuine imputation quality. We assess whether the rankings by the masking criteria correspond to the rankings based on this oracle benchmark. Missing values are introduced under missing completely at random (MCAR) and missing-at-random (MAR). The MCAR mechanism, as described in the main text, independently masks each entry with a fixed probability 0.3. To simulate the MAR mechanism, we use the ``ampute'' function in the \textit{mice} package in R with a 30\% missingness proportion and default weighting scheme. 
From Table~\ref{tab:oracle}, the masking criteria produce rankings that are similar to the oracle risks across datasets under both MCAR and MAR. This rank concordance indicates that performance on masked entries provides a reliable proxy for evaluating true imputation quality. This is essential because, in real-world missing data problems, the oracle imputation risks cannot be computed. The concordance under both MCAR and MAR further establishes the robustness of the masking criteria, which provide stable rankings of imputation models across different missingness mechanisms.

\begin{table}[htbp!]
\centering
\caption{Ranks of imputation methods under masking criteria and corresponding oracle imputation risks across the iris, yacht, and concrete data for (a) MCAR and (b) MAR missingness mechanisms. Smaller rank indicates small risk.}
\label{tab:oracle}

\text{(a) MCAR}\\
\begin{tabular}{@{}llcccccc@{}}
\toprule
& & \multicolumn{2}{c}{iris} & \multicolumn{2}{c}{yacht} & \multicolumn{2}{c}{concrete} \\
\cmidrule(lr){3-4}\cmidrule(lr){5-6}\cmidrule(lr){7-8}
Method & Criterion & Mask & Oracle & Mask  & Oracle & Mask  & Oracle \\
\midrule
\multirow{3}{*}{Mean} 
 & MOO & 6 & 6 & 2 & 2 & 1 & 1\\
 & MOORT & 7 & 7 & 7 & 7 & 7 & 7\\
 & MOOEN & 7 & 7 & 7 & 7 & 7 & 7\\
\midrule
\multirow{3}{*}{EM} 
 & MOO & 2 & 2 & 4 & 4 & 3 & 3\\
 & MOORT & 1 & 2 & 1 & 2 & 1 & 1\\
 & MOOEN & 2 & 2 & 3 & 3 & 2 & 2\\
\midrule
\multirow{3}{*}{Nearest Neighbor Hot Deck} 
 & MOO & 5 & 5 & 6 & 6 & 6 & 6\\
 & MOORT & 6 & 6 & 5 & 5 & 6 & 5\\
 & MOOEN & 5 & 5 & 5 & 5 & 3 & 3\\
\midrule
\multirow{3}{*}{Random Hot Deck} 
 & MOO & 7 & 7 & 7 & 7 & 7 & 7\\
 & MOORT & 2 & 1 & 4 & 4 & 4 & 4\\
 & MOOEN & 6 & 6 & 6 & 6 & 6 & 6\\
\midrule
\multirow{3}{*}{MMG} 
 & MOO & 4 & 4 & 3 & 3 & 5 & 5\\
 & MOORT & 3 & 3 & 2 & 1 & 2 & 2\\
 & MOOEN & 3 & 3 & 2 & 2 & 4 & 4\\
\midrule
\multirow{3}{*}{CCMV} 
 & MOO & 3 & 3 & 5 & 5 & 4 & 4\\
 & MOORT & 4 & 4 & 3 & 3 & 3 & 3\\
 & MOOEN & 4 & 4 & 4 & 4 & 5 & 5\\
\midrule
\multirow{3}{*}{MICE} 
 & MOO & 1 & 1 & 1 & 1 & 2 & 2\\
 & MOORT & 5 & 5 & 6 & 6 & 5 & 6\\
 & MOOEN & 1 & 1 & 1 & 1 & 1 & 1\\
\bottomrule
\end{tabular}

\vspace{1em} 

\text{(b) MAR}\\
\begin{tabular}{@{}llcccccc@{}}
\toprule
& & \multicolumn{2}{c}{iris} & \multicolumn{2}{c}{yacht} & \multicolumn{2}{c}{concrete} \\
\cmidrule(lr){3-4}\cmidrule(lr){5-6}\cmidrule(lr){7-8}
Method & Criterion & Mask & Oracle & Mask  & Oracle & Mask  & Oracle \\
\midrule
\multirow{3}{*}{Mean} 
& MOO & 6 & 6 & 6 & 6 & 6 & 6\\
& MOORT & 7 & 7 & 7 & 7 & 7 & 7\\
& MOOEN & 7 & 7 & 7 & 7 & 7 & 7\\
\midrule
\multirow{3}{*}{EM} 
& MOO & 3 & 1 & 3 & 2 & 3 & 3\\
& MOORT & 1 & 1 & 2 & 2 & 1 & 1\\
& MOOEN & 2 & 1 & 3 & 3 & 3 & 3\\
\midrule
\multirow{3}{*}{Nearest Neighbor Hot Deck} 
& MOO & 1 & 5 & 5 & 5 & 4 & 5\\
 & MOORT & 6 & 6 & 5 & 5 & 6 & 6\\
 & MOOEN & 5 & 5 & 5 & 5 & 2 & 2\\
\midrule
\multirow{3}{*}{Random Hot Deck} 
& MOO & 7 & 7 & 7 & 7 & 7 & 7\\
 & MOORT & 4 & 5 & 4 & 4 & 4 & 4\\
 & MOOEN & 6 & 6 & 6 & 6 & 6 & 6\\
\midrule
\multirow{3}{*}{MMG}
& MOO & 5 & 4 & 2 & 3 & 5 & 4\\
 & MOORT & 2 & 2 & 1 & 1 & 2 & 3\\
 & MOOEN & 4 & 4 & 2 & 2 & 5 & 5\\
\midrule
\multirow{3}{*}{CCMV} 
& MOO & 4 & 2 & 4 & 4 & 2 & 2\\
 & MOORT & 3 & 3 & 3 & 3 & 3 & 2\\
 & MOOEN & 3 & 3 & 4 & 4 & 4 & 4\\
\midrule
\multirow{3}{*}{MICE} 
& MOO & 2 & 3 & 1 & 1 & 1 & 1\\
 & MOORT & 5 & 4 & 6 & 6 & 5 & 5\\
 & MOOEN & 1 & 2 & 1 & 1 & 1 & 1\\
\bottomrule
\end{tabular}
\end{table}

Figure~\ref{fig:PI_sim_appendix} extends Figure~\ref{fig:PI_sim} by including the random hot-deck method (randomly impute from observed entries of the same variable, regardless of other information; \citealt{little2019statistical}) and illustrates a scenario where MOOEN can be preferred to MOORT. The corresponding numerical results are reported in Table~\ref{tab:appendix_simulation_res}. Although the random hot deck method selects donors entirely at random, 
it appears to perform well under MOORT. This behavior can be intuitively explained by the fact that random draws preserve the empirical distribution of the observed data. The method essentially shuffles the observed values and thus results in a small Kolmogorov distance to the uniform distribution. In contrast, MOOEN serves as a more reliable criterion when such an imputation method is included in the comparison.

\begin{figure}[htbp!]
    \centering
    \includegraphics[width=\linewidth]{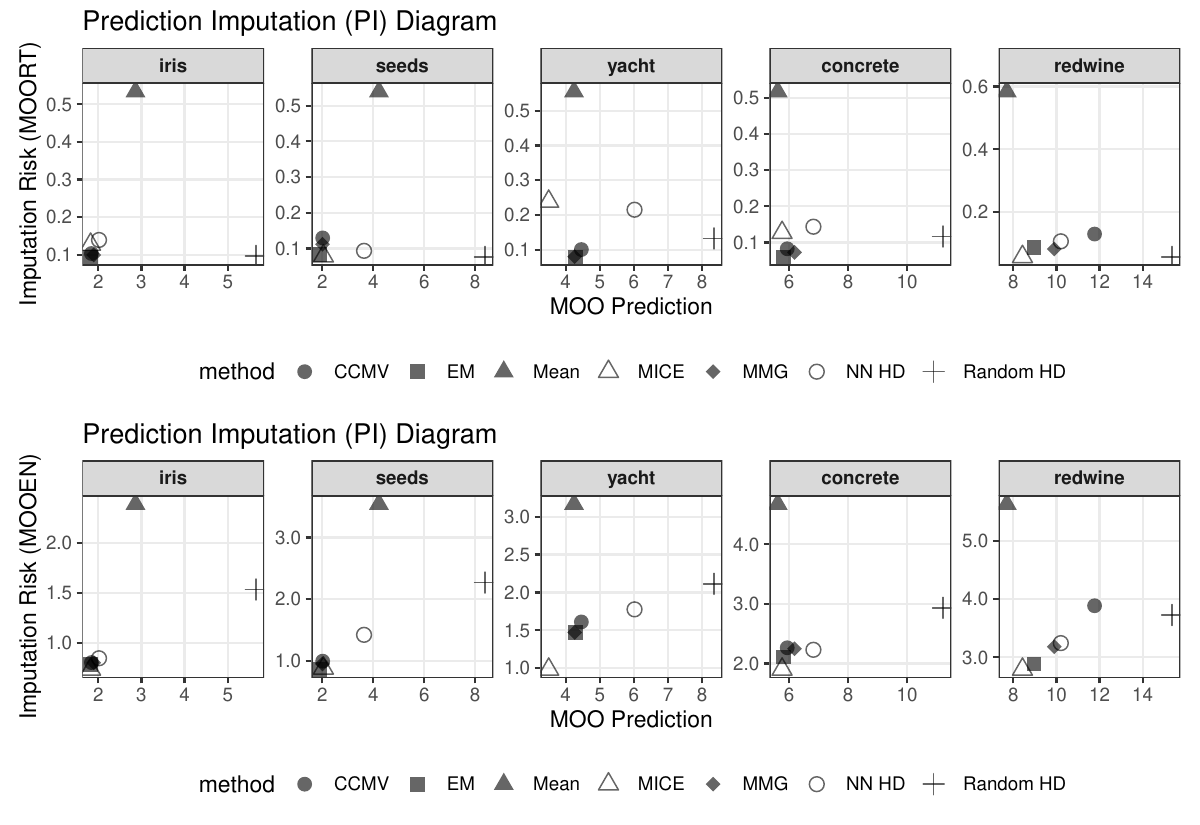}
    \caption{Prediction-Imputation (PI) Diagram comparing imputation methods (CCMV, EM, mean imputation, MICE, MMG, nearest-neighbor hot deck, and random hot deck) under MOO, MOORT, and MOOEN criteria across multiple datasets.}
    \label{fig:PI_sim_appendix}
\end{figure}

\begin{table}[htbp!]
\centering
\caption{Imputation risks across five datasets under the masking criteria.}
\label{tab:appendix_simulation_res}
\begin{tabular}{@{}clccccccc@{}}
\toprule
\text{Dataset} & \text{Criterion} & \text{Mean} & \text{EM} & \text{NN HD} & \text{Random HD} & \text{MMG} & \text{CCMV} & \text{MICE} \\
\midrule
\multirow{3}{*}{\text{iris}} & MOO & 2.862 & 1.827 & 2.017 & 5.654 & 1.894 & 1.842 & 1.823 \\
& MOORT & 0.534 & 0.094 & 0.139 & 0.097 & 0.099 & 0.103 & 0.127 \\
& MOOEN & 2.391 & 0.784 & 0.849 & 1.537 & 0.804  & 0.807 & 0.738 \\
\midrule
\multirow{3}{*}{\text{seeds}} & MOO & 4.223 & 1.895 & 3.327 & 8.386 & 2.005 & 2.013 & 2.042 \\
& MOORT & 0.539 & 0.082 & 0.093 & 0.076 & 0.112 & 0.129 & 0.078 \\
& MOOEN & 3.543 & 0.861 & 1.422 & 2.270 & 0.951 & 0.995 & 0.879 \\
\midrule
\multirow{3}{*}{\text{yacht}} & MOO & 4.244 & 4.282 & 6.019 & 8.352 & 4.263 & 4.460 & 3.502 \\
& MOORT & 0.555 & 0.079 & 0.215 & 0.133 & 0.080 & 0.101 & 0.239 \\
& MOOEN & 3.167 & 1.476 & 1.773 & 2.116 & 1.471 & 1.610 & 0.986 \\
\midrule
\multirow{3}{*}{\text{concrete}} & MOO & 5.617 & 5.821 & 6.825 & 11.22 & 6.184 & 5.937 & 5.764 \\
& MOORT & 0.517 & 0.060 & 0.144 & 0.117 & 0.073 & 0.083 & 0.127 \\
& MOOEN & 4.676 & 2.109 & 2.231 & 2.936 & 2.249 & 2.265 & 1.902 \\
\midrule
\multirow{3}{*}{\text{redwine}} & MOO & 7.695 & 8.955 & 10.20 & 15.37 & 9.898 & 11.77 & 8.424 \\
& MOORT & 0.584 & 0.087 & 0.107 & 0.058 & 0.083 & 0.130 & 0.058 \\
& MOOEN & 5.639 & 2.886 & 3.244 & 3.729 & 3.183 & 3.888 & 2.791 \\
\bottomrule
\end{tabular}
\end{table}


\end{appendix}

\end{document}